\documentclass[12pt]{article}
\usepackage{amsmath,amssymb,amsfonts,amsthm}
\usepackage{mathrsfs}
\usepackage[backref=false]{hyperref} 
\usepackage{graphicx}
\usepackage{enumerate}
\usepackage{natbib}
\usepackage{url} 

\newcommand{\blind}{1}

\usepackage{multirow}
\usepackage{multicol}
\usepackage{mathbbol}
\usepackage[normalem]{ulem} 

\newtheorem{theorem}{Theorem}[section]
\newtheorem{lemma}{Lemma}[section]

\newtheorem{corollary}{Corollary}[section]
\newtheorem{assumption}{Assumption}[section]
\newtheorem{example}{Example}[section]
\newtheorem{proposition}{Proposition}[section]
\theoremstyle{remark}
\newtheorem*{remark}{Remark}
\newtheorem*{remark*}{Remark}
\usepackage{chngcntr}
\counterwithin{table}{section}

\usepackage{algorithm,algpseudocode}
\usepackage{bm}
\usepackage{caption,subcaption}
\counterwithout{figure}{section}
\usepackage{multicol}
\usepackage{tabularx,booktabs}
\usepackage{float}
\usepackage{mathtools}
\floatstyle{plaintop}

\restylefloat{table}
\newcommand{\comment}[1]{}

\allowdisplaybreaks
\usepackage{etoolbox}
\newcommand{\zerodisplayskips}{%
  \setlength{\abovedisplayskip}{0pt}%
  \setlength{\belowdisplayskip}{0pt}%
  \setlength{\abovedisplayshortskip}{0pt}%
  \setlength{\belowdisplayshortskip}{0pt}}
\appto{\normalsize}{\zerodisplayskips}
\appto{\small}{\zerodisplayskips}
\appto{\footnotesize}{\zerodisplayskips}

\usepackage{xcolor}

\begin{document}

\def\spacingset#1{\renewcommand{\baselinestretch}%
{#1}\small\normalsize} \spacingset{1}


\if1\blind
{
  \title{\bf A Distributionally Robust Optimization Framework for Extreme Event Estimation}
  \author{Yuanlu Bai, Henry Lam, Xinyu Zhang\thanks{
    The authors gratefully acknowledge support from the National Science Foundation under grants CAREER CMMI-1834710 and IIS-1849280.}
    \\
  Department of Industrial Engineering and Operations Research\\
  Columbia University\\
  New York, NY 10025, USA\\}
  \date{}
  \maketitle

} \fi

\if0\blind
{
  \bigskip
  \bigskip
  \bigskip
  \begin{center}
    {\LARGE\bf A Distributionally Robust Optimization Framework for Extreme Event Estimation}
\end{center}
  \medskip
} \fi

\bigskip
\begin{abstract}
Conventional methods for extreme event estimation rely on well-chosen parametric models asymptotically justified from extreme value theory (EVT). These methods, while powerful and theoretically grounded, could however encounter a difficult bias-variance tradeoff that exacerbates especially when data size is too small, deteriorating the reliability of the tail estimation. In this paper, we study a framework based on the recently surging literature of distributionally robust optimization. This approach can be viewed as a nonparametric alternative to conventional EVT, by imposing general shape belief on the tail instead of parametric assumption and using worst-case optimization as a resolution to handle the nonparametric uncertainty. We explain how this approach bypasses the bias-variance tradeoff in EVT. On the other hand, we face a conservativeness-variance tradeoff which we describe how to tackle. We also demonstrate computational tools for the involved optimization problems and compare our performance with conventional EVT across a range of numerical examples.

\end{abstract}

\noindent%

\spacingset{1.9} 
\section{Introduction}
\label{sec:intro}
Extreme event estimation is ubiquitous in risk assessment and planning for a wide range of problems, from the century-old coastal flooding, earthquake, financial crisis, to most recently pandemic outbreak. 
In environmental sciences for instance, river flow and wave height data are analyzed to model and predict floods \citep{davisonModelsExceedancesHigh1990}. In engineering, system reliability is assessed via estimating failure or crash probabilities \citep{heidelbergerFastSimulationRare1995,nicolaFastSimulationHighly1993,jonasson2014internal,zhao2016accelerated}. In insurance and finance, pricing and risk management is informed by the prediction of major losses \citep{beirlantModelingLargeClaims1992, embrechtsModellingExtremalEvents1997,glassermanLargeDeviationsMultifactor2007, glassermanFastSimulationMultifactor2008, alexanderj.mcneilQuantitativeRiskManagement2015}. 


A beginning challenge in extreme event estimation is that, by the very definition, there are typically few observations in the data that provide direct information on the tail of the distribution. To this end, the prominent approach is to use extreme value theory (EVT). This approach suggests parametric models for extrapolating tails based on asymptotic theory, and consists of two main line of methods. The first is the block-maxima method \citep{emiljuliusgumbelStatisticsExtremes1958}, which is based on the celebrated Fisher-Tippet-Gnedenko theorem  \citep{fisherLimitingFormsFrequency1928, gnedenkoDistributionLimiteTerme1943} which stipulates that the maximum of i.i.d. random variables $X_i$, under maximum-domain-of-attraction assumptions and suitable normalization, converges in distribution to the generalized extreme value distribution (GEV) as the sample size grows. GEV includes conveniently only three possible more specific distributions: Gumbel, Frechet and  Weibull. Once we fit the maximum of $n$ observations, say $M_n$, we can infer the tail of the data distribution from the relation $\mathbb{P}(M_n\leq x)=\prod_{i=1}^n\mathbb{P}(X_i\leq x)$. This approach is justified as there are natural scenarios where data are collected as the maxima over certain time periods \citep{davisonModelsExceedancesHigh1990}, otherwise one could group the data into blocks and obtain the maximum within each block (e.g., \citet{embrechtsModellingExtremalEvents1997} Section 8.1.2, \citet{alexanderj.mcneilQuantitativeRiskManagement2015} Section 7.1.4).
The second line of methods is the peak-over-threshold (POT) \citep{richardl.smithThresholdMethodsSample1984}. This is based on the Pickands-Balkema-de Haan theorem \citep{balkemaResidualLifeTime1974,jamespickandsiiiStatisticalInferenceUsing1975} which states, under the same domain-of-attraction assumption as the Fisher-Tippet-Gnedenko theorem, that the distribution function of the so-called excess loss above a threshold converges to the generalized Pareto distribution (GPD) as this threshold increases. Thus, GPD serves as a justified choice to fit the tail portion of data. For both lines of method, the involved parametric estimation in either the GEV or the GPD, as well as other implementation details, have been substantially studied, including for example maximum likelihood \citep{smithMaximumLikelihoodEstimation1985}, probability-weighted moments \citep{hoskingParameterQuantileEstimation1987} and the Hill estimator \citep{davisTailEstimatesMotivated1984, hillSimpleGeneralApproach1975}. Lastly, besides EVT, there are also a range of exploratory ``goodness-of-fit'' tools available for tail modeling such as the quantile plot, mean excess plot and max-sum-ratio plots \citep{embrechtsModellingExtremalEvents1997}. 






Despite the theoretical justification and practical usefulness of EVT, it could face difficulties in balancing the intrinsic bias-variance tradeoff. More concretely, the block-maxima method generally requires choosing the block size of data to obtain the maxima. The larger is the number of blocks, the larger number of maxima can be used to estimate the parameters in the GEV and hence the smaller the estimation variance, but this would unavoidably reduce the sample size in each block and lead to larger bias. Similarly, POT requires selecting the threshold level for defining the excess loss. If the threshold is chosen to be low, then more data are above the threshold which can be used to estimate the parameters in the GPD and leads to a smaller estimation variance, but this would cause a larger bias due to the inadequacy of the asymptotic approximation. Choosing the block size or threshold level could intricately depend on the higher-order behaviors of the asymptotic theory \citep{smithEstimatingTailsProbability1987,bladtThresholdSelectionTrimming2020}. Moreover, when the data size is small, it could happen that no choice of block size or threshold level could make both bias and variance small enough simultaneously, thus leading to a significant overall estimation error. 

Motivated by the above challenges in conventional EVT, in this paper we propose an alternative approach for extreme event estimation via the recently surging tool of \emph{distributionally robust optimization (DRO)} \citep{delageDistributionallyRobustOptimization2010,gohDistributionallyRobustOptimization2010,danielkuhnWassersteinDistributionallyRobust2019}. DRO originates as a method for optimization under uncertainty and can be viewed as a generalization of classical robust optimization (RO) \citep{ben-talRobustConvexOptimization1998,bertsimasPriceRobustness2004, ben-talRobustOptimization2009}. When making decision in problems containing uncertain or ambiguous parameters, RO advocates the optimization of decision under the worst-case scenario, where the worst-case is over a feasible region called the uncertainty set or ambiguity set that postulates the likely value of the parameters. It thus often involves a minimax problem where the inner maximization is to compute the worst-case parameter value. When the parameter in the problem is the underlying probability distribution in a stochastic problem, then one would compute the worst-case distribution, in which case it becomes DRO. Here, in this work, we will take a more liberal view of DRO to refer to it as the computation of the worst-case distribution or the resulting worst-case value, not necessarily involving a minimax problem.

Our DRO operates as a nonparametric alternative to EVT as follows. Instead of extrapolating tail using asymptotically justified parametric assumptions such as GEV or GPD, we make stylized geometric shape assumptions on the tail. Examples of these geometric assumptions include monotonicity and convexity, which are intentionally mild and cover not only the tails of GEV and GPD but essentially all common parametric tail distributions. The challenge, however, is that these mild geometric assumptions do not locate specific tail models, or in other words there could be many ways to extrapolate tails under these premises. This is where the worst-case notion in DRO kicks in -- Suppose we are interested in estimating a target extremal quantity such as tail probability, we could compute, among all the extrapolated tail that satisfies the geometric assumptions, the one that give rise to the worst-case value of the target extremal quantity. This therefore comes down to solving an optimization over the space of tail distributions under shape constraints and other auxiliary conditions to ensure the consistency of extrapolation, where these constraints comprise precisely the uncertainty set in the DRO framework. We will demonstrate that, when these constraints are correctly calibrated, the resulting worst-case value would give rise to statistically correct bounds on the extremal quantity.

The above DRO framework bypasses the bias-variance tradeoff in EVT in the sense that we no longer require asymptotic distributional approximation, thus free of model misspecifications due to the use of parametric models. However, instead of getting a consistent estimator, the DRO approach gives bounds on the target extremal quantity, which poses a challenge of conservativeness. That is, the generated bounds, while correct, could be loose. More concretely, as in the POT approach, our DRO approach would need to select a threshold that defines the tail region. In POT, this threshold choice faces a bias-variance tradeoff where the bias comes from the approximation error using GPD while the variance comes from the parameter estimation therein. In DRO, the threshold choice faces a conservativeness-variance tradeoff, where the conservativeness arises if our threshold is chosen too low, in which case there are more tail distributions satisfying our geometric shape conditions and hence loosening the worst-case value, while the variance enlarges if our threshold is chosen too high, in which case few observations are available to calibrate our auxiliary constraints. To this end, we will study the conservativeness of DRO by quantifying the looseness of the resulting bounds in relation to the maximum domain of attraction of the underlying tail distribution. This also provides a mathematical link between DRO and EVT. Furthermore, we will study approaches to alleviate the conservativeness in DRO. These approaches include the addition of auxiliary  moment constraints that inject maximal information about the tail to reduce the uncertainty set size, and also procedures to select the threshold that defines the tail region. 

In terms of computation, we will also present reformulation approaches and procedures to solve our proposed DRO. We note that DRO by nature are infinite-dimensional optimization problems as its decision variable is a probability distribution. To this end, we leverage results in the optimization literature to reformulate our DRO problems, which involve geometric shape constraints, into moment problems that can be dualized into semidefinite programs. Finally, with our solvable and statistically calibratable DRO formulations, we compare our approach with conventional EVT across a range of numerical examples. In particular, we show how DRO provides more reliable estimation on target extremal quantities than EVT, in the sense that the generated confidence bounds are correct more often than EVT. On the other hand, DRO also pays a price of conservativeness, which is also consistent with our theoretical understanding.

The rest of this paper is as follows. Section \ref{sec:framework} first presents our DRO framework for extreme event analysis, including the selection of constraints and thresholds in the involved optimization problem, and discusses related literature. Section \ref{sec:stat} quantifies the conservativeness of our approach by connecting to EVT. Section \ref{sec:computation} discusses the solution approach to the proposed DRO. Section \ref{sec:numerics} shows numerical performances of our approach and compares with EVT.

\section{Framework and Basic Statistical Guarantees}\label{sec:framework}
We are interested in estimating a target quantity $\psi(P)$ that depends on the  distribution $P$ that is unknown but observed from data. The quantity $\psi(P)$ is assumed to be an extremal quantity, i.e., $\psi(P)$ depends on the tail of $P$, e.g., the tail probability $P(X\geq L)$ for some large $L$, or $P(L\leq X\leq R)$ for some interval $[L,R]$, where $X$ is distributed according to $P$.

We consider estimating $\psi(P)$ by setting up an optimization problem that, on a high level, can be written as follows:
\begin{equation}
\begin{array}{ll}
\max_P&\psi(P)\\
\text{subject to}&\text{geometric shape condition on $P$ holds for $x\geq a$}\\
&\text{auxiliary constraints on $P$}
\end{array}\label{opt basic}
\end{equation}
where the decision variable is the unknown true distribution $P$. Let us first explain formulation \eqref{opt basic} and how to use it to estimate $\psi(P)$ intuitively before drilling into details. First, instead of using EVT to fit $P$ as a parametric distribution such as GEV or GPD, we impose a geometric shape condition on the tail of $P$, where the tail region is defined by the condition $x\geq a$ for some large threshold level $a$. When imposing the shape condition, we also have to ensure consistent extrapolation from the non-tail region, and inject any additional worthwhile information which comprises the auxiliary constraints. With all these, \eqref{opt basic} is designed to provide a confidence upper bound for $\psi(P)$. The following is an immediate statistical guarantee:
\begin{proposition}
\label{proposition:abstrac_statistical_guarantee}
Suppose that the constraints in \eqref{opt basic} are correct with statistical confidence $(1-\alpha)$, namely
$$\mathbb P\left(\begin{array}{l}\text{geometric shape condition holds on $P_{true}$ for $x\geq a$}\\
\text{auxiliary constraints holds on $P_{true}$ for $x\geq a$}\end{array}\right)\geq1-\alpha$$
where $P_{true}$ denotes the true $P$ distribution. Then the optimal value of the optimization problem \eqref{opt basic}, denoted $Z^*$, is an upper confidence bound for the true value of $\psi(P_{true})$, denoted $Z_{true}$, with at least the same level of confidence, namely
$$\mathbb P(Z^*\geq Z_{true})\geq1-\alpha.$$
Moreover, the same assertion holds for the asymptotic counterpart. That is, if 
$$\liminf \mathbb P\left(\begin{array}{l}\text{geometric shape condition holds on $P_{true}$ for $x\geq a$}\\
\text{auxiliary constraints holds on $P_{true}$ for $x\geq a$}\end{array}\right)\geq1-\alpha$$
then
$$\liminf \mathbb P(Z^*\geq Z_{true})\geq1-\alpha.$$
In the above, $\mathbb P$ refers to the probability with respect to the data and $\liminf$ is taken as the sample size grows to infinity.\label{basic guarantee DRO}
\end{proposition}
\begin{proof}Suppose $P_{true}$ satisfies all the constraints in \eqref{opt basic}. Then clearly it is a feasible solution to the optimization problem \eqref{opt basic}, and thus $Z^*\geq\psi(P_{true})$. Therefore $$\mathbb P(Z^*\geq\psi(P_{true}))\geq\mathbb P\left(\begin{array}{l}\text{geometric shape condition holds on $P_{true}$ for $x\geq a$}\\
\text{auxiliary constraints holds on $P_{true}$ for $x\geq a$}\end{array}\right)\geq1-\alpha.$$
The asymptotic counterpart also holds in the same manner.
\end{proof}

Problem \eqref{opt basic} can also be phrased as a minimization problem, for which a lower confidence bound guarantee analogous to Proposition \ref{proposition:abstrac_statistical_guarantee} would hold but to avoid repetition we have skipped this. In light of Proposition \ref{basic guarantee DRO}, the question becomes what constraints we put in \eqref{opt basic} to attain good confidence bounds for $\psi(P)$. We discuss this in the next subsection.




\subsection{Optimization Objective and Constraints}
\label{subsec:optimization_objective_constraints}
To facilitate discussion, we focus mainly on $\psi(P)$ that is an expectation or quantile of $X$. In the former case, $\psi(P)$ is $\mathbb E[h(X)]$ for some function $h$, where we assume that $h(x)$ is non-zero only on the region $x\geq a$ (i.e., a tail-related quantity; if not, we can always split $\mathbb E[h(X)]$ into two portions  $\mathbb{E}[h(X);X<a]$ and $\mathbb{E}[h(X);X\geq a]$, where the former is a non-tail estimation problem that can be handled by other standard statistical tools). In the latter case, $\psi(P)=\min \{q:\mathbb{P}(X\leq q)\geq p\}$ for some target probability level $p$. 

The constraints in \eqref{opt basic} are imposed to provide information on $P_{true}$, in that the smaller the resulting feasible region, the tighter is the bound. In the tail region, however, typically little is known about $P$. To this end, we impose geometric conditions that are shared by essentially all common parametric distributions and generally believed to hold. Two such conditions are monotonicity and convexity, which we place in the ``geometric shape condition" in \eqref{opt basic}. 

More precisely, let $f(x)$ and $F(x)$ be the density and distribution function respectively for the probability distribution $P$. By monotonicity we mean to restrict $f(x)$ to be right-continuous and non-increasing for $x\geq a$, and by convexity we mean to restrict $f(x)$ to be convex for $x\geq a$ (note that in this case the existence of one-sided derivatives is guaranteed; e.g., \citet{rockafellarConvexAnalysis1970} Theorem 24.1). For coherence, we can leverage the notion of $D$-th order monotonicity (\citet{pestanaHigherorderMonotoneFunctions2001,vanparysDistributionallyRobustExpectation2019}) to represent these two cases. Let $\mathscr{P}^D(a)$ denote the set of all distribution functions that are $D-1$ times differentiable, and the $D$-th order right derivative exists and is finite and monotone on $[a,\infty)$. Then for $D=0$, $\mathscr{P}^D(a)$ reduces to all probability distributions on $\mathbb{R}$ without any monotonicity assumptions, 
$\mathscr{P}^1(a)$ corresponds to monotone tail distributions, and standard convex analysis shows $\mathscr{P}^2(a)$ is precisely the class of convex tail distributions. We utilize the class $\mathscr{P}^D(a),\ D=1,2$.



The geometric shapes described above need to be accompanied by certain distributional information at the threshold $a$. This information at $a$ serves as ``boundary" conditions" for two purposes. One is to ensure the extrapolation from the non-tail to the tail region at $a$ is natural, in the sense that the density at $a$ still follows the imposed shape condition in a small neighborhood before $x$ hits $a$. Second is to avoid trivializing the optimization and the resulting confidence bounds since without these boundary conditions the worst-case density can be unrealistically pessimistic. More precisely, in the monotone case, this information refers to the density value at $a$, while in the convex case this refers to the density and also its right derivative at $a$. These quantities can be estimated in the form of confidence intervals using standard statistical tools (which we describe later this section). With these, the ``geometric shape condition" in \eqref{opt basic}, for the monotone and convex case, become
\begin{align}
\mathscr{P}^{1}_{a,\eta}:=\{P\in \mathscr{P}^{1}(a)|F_+'(a)\leq \eta\}, \ \mathscr{P}^{2}_{a,\underline{\eta},\bar{\eta},\nu}:=\{P\in \mathscr{P}^{2}(a)|\underline{\eta}\leq f(a)\leq \bar{\eta},f_+'(a)\geq -\nu\}.\label{monotone_convex_assumption}
\end{align}
 respectively where $F_+'(\cdot)$ and $f_+'(\cdot)$ denote the right-derivatives of $F(\cdot)$ and $f(\cdot)$, and $\eta,\underline\eta,\overline\eta,\nu$ are parameters calibratable from data. Note that in \eqref{monotone_convex_assumption} we have used only the upper bound for $F_+'(a)$ in the monotone case and the lower bound for $f_+'(a)$ in the convex case. This is because the lower bound for $F_+'(a)$ or the upper bound for $f_+'(a)$, when coupled with the set $\mathscr P^1(a)$ or $\mathscr P^2(a)$, turns out to be redundant in the optimization problem.


The auxiliary constraints in \eqref{opt basic} are in the form of moment conditions that are calibratable using data above $a$. These constraints serve to reduce the conservativeness in only using the shape belief to extrapolate tails. In general, we consider conditions in the form
\begin{equation}
\mathbb{E}[\bm{g}(X)]\in\mathbb{S}\label{moment constraint general}
\end{equation}
where $\bm g(\cdot)=(g_j(\cdot))_{j=1,\ldots,d}$ and $g_j(\cdot)$ can take the form of an indicator function, e.g., $g_j(x)=\mathbb{I}(a\leq x\leq b)$, or power moment function, e.g., $g_j(x)=x\mathbb{I}(x\geq a)$ (the latter implicitly impose finiteness of power moments which is not always appropriate). Moreover, we always set one of $g_j(\cdot)$ to be $\mathbb I(\cdot\geq a)$, which encodes the probability mass on $x\geq a$ and can be regarded as a boundary condition described earlier. 

For $\mathbb S$, we consider two types, namely ellipsoid $\mathbb{S}_{E}(\bm{\mu},\bm{\Sigma},r)$ and rectangle $\mathbb{S}_{R}(\underline{\bm{\mu}},\bar{\bm{\mu}})$, each parameterized by $(\bm{\mu},\bm{\Sigma},r)\in\mathbb R^d\times\mathbb R^{d\times d}\times\mathbb R_+$ and $(\underline{\bm{\mu}},\bar{\bm{\mu}})\in\mathbb R^d\times\mathbb R^d$ respectively, as
\begin{align}
\mathbb{S}_{E}(\bm{\mu},\bm{\Sigma},r)&:=\big\{\bm{y}:[\bm{y}-\bm{\mu}]^\top \bm{\Sigma}^{-1}[\bm{y}-\bm{\mu}]\leq r\big\}\label{ellipsoid},\\
\mathbb{S}_{R}(\underline{\bm{\mu}},\bar{\bm{\mu}})&:=\big\{\bm{y}: \underline{\bm{\mu}} \leq \bm{y}\leq \bar{\bm{\mu}}\big\}\label{rectangle}.
\end{align}

Before we proceed further, we discuss a bit more on the role of these auxiliary constraints. One may question whether they can be used as the sole constraints with the shape information discarded. A main issue with this approach is that this would lead to a highly conservative result. This is because of a fact in infinite-dimensional programming that solving a distributional optimization problem with only moment constraints typically lead to a finite-supported optimal solution, where the number of support points is the number of independent constraints \citep{winklerExtremePointsMoment1988}. Such a discrete distribution is unrealistic and thus give a very pessimistic bound. 


With the geometric shape conditions in \eqref{monotone_convex_assumption} and the auxiliary moment constraints \eqref{moment constraint general}, \eqref{opt basic} can now be written as
\begin{align}\label{general_framework}
\begin{split}
\mathfrak{P}(\mathscr{P},\bm{g},\mathbb{S}): \quad \max_{P}\ \psi(P) \text{\ \ subject to\ }\ P\in\mathscr{P},\ \mathbb{E}_P[ \bm{g}(X)]\in\mathbb{S}
\end{split}
\end{align}
where $\mathscr P$ can be $\mathscr{P}^1_{a,\eta}$ or $\mathscr{P}^2_{a,\underline{\eta},\bar{\eta},\nu}$ and $\mathbb S$ can be $\mathbb{S}_{E}(\bm{\mu},\bm{\Sigma},r)$ or $\mathbb{S}_{R}(\underline{\bm{\mu}},\bar{\bm{\mu}})$ described above, and $\mathbb E_P$ denotes the expectation under $P$. For convenience, we also denote $\mathbb F(\mathfrak{P}(\mathscr{P},\bm{g},\mathbb{S}))=\{P:P\in\mathscr{P},\ \mathbb{E}_P[ \bm{g}(X)]\in\mathbb{S}\}$ as the feasible region in problem \eqref{general_framework}. We have the following immediate corollary from Proposition \ref{proposition:abstrac_statistical_guarantee}: 



\begin{theorem}\label{proposition:sg1}
We have the following, given threshold $a$.

\begin{enumerate}
\item[(i)] Suppose the true distribution $P_{true}$ lies in $\mathscr P^1(a)$. Suppose the value $\eta$ and the set $\mathbb{S}$ are selected such that the true distribution function satisfies $F_+'(a)\leq \eta$ and $\mathbb{E}_{P_{true}}[ \bm{g}(X)]\in\mathbb{S}$ jointly with confidence $1-\alpha$. Then the optimal value of \eqref{general_framework}, with $\mathscr P=\mathscr{P}^1_{a,\eta}$, is an upper bound for $\psi(P_{true})$ also with confidence $1-\alpha$. Analogous assertion holds when the confidences are satisfied asymptotically.


\item[(ii)] Suppose the true distribution $P_{true}$ lies in $\mathscr P^2(a)$. Suppose the values $\underline\eta,\overline\eta,\nu$ and the set $\mathbb{S}$ are selected such that the true density function satisfies $\underline\eta\leq f(a)\leq\overline\eta$ and $f_+'(a)\geq-\nu$ and $\mathbb{E}_{P_{true}}[ \bm{g}(X)]\in\mathbb{S}$ jointly with confidence $1-\alpha$. Then the optimal value of \eqref{general_framework}, with $\mathscr P=\mathscr{P}^2_{a,\underline{\eta},\bar{\eta},\nu}$, is an upper bound for $\psi(P_{true})$ also with confidence $1-\alpha$. Analogous assertion holds when the confidences are satisfied asymptotically.
\end{enumerate}
\end{theorem}

In view of Theorem \ref{proposition:sg1}, optimization problem \eqref{general_framework} provides a statistically valid confidence upper bound as long as our geometric belief on the tail is correct and the parameters in the constraints are calibrated via the correct confidence bounds. In the monotone case, $\eta$ can be plugged in as the upper confidence bound for the density at $a$. In the convex case, $\underline\eta,\overline\eta$ can be plugged in as the limits of confidence interval for the density at $a$, and $\nu$ as the lower confidence bound for the density derivative at $a$.

For the moment functions $\bm g$ and set $\mathbb S$, there are two main approaches to choose and calibrate them, corresponding to the ellipsoid and rectangle respectively. In the first approach, we could consider several functions, for example the indicator function of lying in an interval or power function, and calibrate $\bm{\mu}$, $\bm{\Sigma}$ and $r$ as the point estimates of these (generalized) moments, their covariance matrix estimate (scaled by the sample size), and a $\chi^2$-quantile respectively, based on the elementary multivariate central limit theorem.  In the second approach, we could use a Kolmogorov-Smirnov statistic to calibrate $\underline{\bm{\mu}},\bar{\bm{\mu}}$, in which case each $g_j(\cdot)=\mathbb I(\cdot\leq x_j)$ for a data point $x_j$.

Note that the conditions for $\eta$ or $\underline\eta,\overline\eta,\nu$, and $\mathbb{E}_{P_{true}}[ \bm{g}(X)]\in\mathbb{S}$, need to be held jointly. Thus one would need to use the Bonferroni correction to calibrate all these parameter values to ensure a family-wise confidence level $1-\alpha$. These procedures are routine and we show the details in Appendix \ref{sec:calibration} for completeness purpose.



Lastly, we discuss how to choose the threshold $a$. A simple guideline is to leave out a small amount of data above $a$, enough so that elementary central limit theorem can be applied, so that we can calibrate the moment constraints and estimate the distributional information at $a$ at an adequate accuracy. Alternately, we can consider solving the optimization problem \eqref{general_framework} at several values of $a$, say $a_i$ for $i=1,\ldots,m$. Then, suppose the feasible region $\mathbb F(\mathfrak{P}(\mathscr{P}_i,\bm{g}_i,\mathbb{S}_i))$ for each problem at $a_i$, configured by $\mathscr{P}_i,\bm{g}_i,\mathbb{S}_i$, is calibrated such that $P_{true}\in\bigcap_{i=1,\ldots,m}\mathbb F(\mathfrak{P}(\mathscr{P}_i,\bm{g}_i,\mathbb{S}_i))$ with confidence $1-\alpha$ (e.g., one way is to make sure $P_{true}\in\mathbb F(\mathfrak{P}(\mathscr{P}_i,\bm{g}_i,\mathbb{S}_i))$ with confidence $1-\alpha/m$ for each $i$, but this is not the only way). 
Then we can set $\min_{i=1,\ldots,m}Z_i^*$ as the $1-\alpha$ upper confidence bound, where $Z_i^*$ denotes the optimal value of $\mathfrak{P}(\mathscr{P}_i,\bm{g}_i,\mathbb{S}_i)$. This ultimate bound is the minimum of all the individual bounds and thus could be tighter, though it needs to be balanced with the impact from the additional Bonferroni correction to guarantee $\mathbb P(P_{true}\in\bigcap_{i=1,\ldots,m}\mathbb F(\mathfrak{P}(\mathscr{P}_i,\bm{g}_i,\mathbb{S}_i)))\geq1-\alpha$. We summarize the guarantee for the above procedure as follows.

\begin{theorem}\label{proposition:sg2}
Given a range of threshold values $a_i,i=1,\ldots,m$, suppose $\psi(P)$ depends on $P$ only on the region $x\geq\max_ia_i$. Suppose either of the following holds:

\begin{enumerate}
\item[(i)] The true distribution $P_{true}\in\mathscr P^1(\min_ia_i)$. We set $\eta=\eta_i$ and $\mathbb S=\mathbb{S}_i$ in \eqref{general_framework} for each $a_i$ such that $P_{true}\in\bigcap_{i=1,\ldots,m}\mathbb F(\mathfrak{P}(\mathscr{P}_i,\bm{g}_i,\mathbb{S}_i))$ with confidence $1-\alpha$.



\item[(ii)] The true distribution $P_{true}\in\mathscr P^2(\min_ia_i)$. We set $\underline\eta=\underline\eta_i,\overline\eta=\overline\eta_i,\nu=\nu_i$ and $\mathbb S=\mathbb{S}_i$ in \eqref{general_framework} for each $a_i$ such that $P_{true}\in\bigcap_{i=1,\ldots,m}\mathbb F(\mathfrak{P}(\mathscr{P}_i,\bm{g}_i,\mathbb{S}_i))$ with confidence $1-\alpha$.
\end{enumerate}

Then $\min_iZ_i^*$ is an upper bound of $\psi(P_{true})$ with confidence $1-\alpha$, where $Z_i^*,i=1,\ldots,m$ are the optimal values of $\mathfrak{P}(\mathscr{P}_i,\bm{g}_i,\mathbb{S}_i)$. Analogous assertion holds when the confidences are satisfied asymptotically.

\end{theorem}

\subsection{Related Literature}

Problem \eqref{opt basic} can be viewed as a worst-case optimization problem in the DRO framework. As described in the introduction, DRO advocates a worst-case perspective for decision-making under ambiguous stochastic uncertainty. The latter means that the decision-maker faces a stochastic optimization problem where the underlying probability distribution that controls the stochasticity is unknown or ambiguous. In this case, DRO solves a minimax problem in which the inner maximization is over the worst-case distribution, among a set that is believed to contain the true distribution, often known as the uncertainty set or ambiguity set. The idea dates back to \cite{scarfMinMaxSolutionInventory1957}, and has found growing applications across various disciplines including economics \citep{hansen2008robustness}, stochastic control \citep{pjd00,doi:10.1287/moor.1120.0540,iyengar2005robust}, finance \citep{glasserman2014robust}, revenue management \citep{lim2006model} and most recently machine learning \citep{rahimian2019distributionally,danielkuhnWassersteinDistributionallyRobust2019,blanchet2021statistical}. Problem \eqref{opt basic} takes a more general view of DRO that refers to the worst-case optimization over distributions but does not necessarily involve a decision. The assertion in Proposition \ref{proposition:abstrac_statistical_guarantee} is an immediate guarantee from data-driven DRO, namely a DRO problem where the uncertainty set is calibrated using data. In particular, when the set is constructed via confidence region, the confidence guarantee for the uncertainty translates into the confidence guarantee of the resulting optimal value of the DRO.

The key of DRO lies in the construction of the uncertainty set. To this end, there are two main approaches. First is to use a neighborhood ball surrounding a baseline distribution, where the ball size is measured by a statistical distance including the $\phi$-divergence \citep{gupta2019near,bayraksan2015data,hu2013kullback,gotoh2018robust,ghosh2019robust,atar2015robust,dey2010entropy,jiang2018risk} and the Wasserstein distance \citep{esfahani2018data,blanchet2021sample,gao2016distributionally,xie2019tractable,shafieezadeh2019regularization,chen2018robust}. This approach has been used as a nonparametric approach to sensitivity analysis \citep{Lam2016robust,Lam2018sensitivity}. It also bears statistical consistency properties in that the resulting optimal value converges to the truth when the ball size is suitably calibrated from data \citep{jiang2016data,bertsimas2018robust}, and has a close relation with the empirical likelihood and its generalizations \citep{LAM2017301,duchi2016statistics,lam2019recovering,Blanchet2019,blanchet2019confidence}. The second approach to construct uncertainty set is to use summary distributional information including moments and support \citep{delageDistributionallyRobustOptimization2010,bertsimasOptimalInequalitiesProbability2005,Wiesemann2014,gohDistributionallyRobustOptimization2010,ghaoui2003worst}, marginal information \citep{doan2015robustness,dhara2021worst}, and geometric shape  \citep{van2016generalized,li2019ambiguous,chen2021discrete}. Though statistical consistency is not guaranteed, this approach is useful to obtain bounds for problems with significant distributional ambiguity. Among these choices, shape condition requires minimal input from data and thus fits into extreme event analysis where data are by definition scarce in the tail. 

There are several recent works that consider DRO in extreme event analysis. The most relevant is \cite{lamTailAnalysisParametric2017} that proposes tail extrapolation using convexity information. \cite{lamTailAnalysisParametric2017} focuses on the light versus heavy-tailed behavior of the extrapolated tail, and proposes nonlinear optimization procedure to distinguish the two cases as well as compute the worst-case distribution. \cite{blanchet2020distributionally} studies the robustification of GEV by considering DRO with Renyi divergence neighborhood ball as the uncertainty set. They focus on the preservation of the maximum domain of attraction in the GEV for the worst-case distribution, and also suggest approaches to calibrate the ball size in practice. Our work differs from these works in that we consider an alternative approach to estimate extremal quantity without GEV, thus different from \cite{blanchet2020distributionally}, and our uncertainty set construction, including the choice of constraints and selection of threshold, and relatedly the computational approach, are more general than \cite{lamTailAnalysisParametric2017}. 

Besides \cite{lamTailAnalysisParametric2017} and \cite{blanchet2020distributionally}, other works that use DRO in extremes include \cite{engelke2017robust}  which studies robust bounds on extremal quantities subject to neighborhood balls measured by $\chi^2$ and the first moment, and \cite{birghila2021distributionally} which studies bounds using the Wasserstein distance and $f$-divergence neighborhood on a heavy-tailed distribution. Moreover, in the multivariate context, a line of works has investigated worst-case bounds when marginal distributions are assumed known but dependence structure is open. These include \cite{embrechts2006bounds1,embrechts2006bounds2,puccetti2013sharp} on tail probabilities related to financial risks, \cite{wang2011complete} on expectations of convex functions of sums, and \cite{dhara2021worst} on conditional value-at-risk. Moreover, \cite{yuen2020distributionally} studies worst-case value-at-risk subject to constraints on the extremal coefficients. These works have different focuses from ours as they primarily focus on dependence structure, and less on the statistical issue in extrapolating tails.

\section{Conservativeness}\label{sec:stat}
In this section, we develop theoretical guarantees for conservativeness, where we leverage EVT to derive asymptotic results. More specifically, to quantify the conservativeness of our DRO framework, we investigate the limiting behavior of the relative error ($:=\text{[Estimated Value]}/\text{[True Value]}-1$) as the threshold value as well as the target quantity become more extreme. First, we introduce an abstract formulation which can be considered as a special case of our general framework, and we focus on two types of problem settings: estimating tail probabilities and tail quantiles. Next, for both types of problems, we present asymptotic results for the relative error under heavy-tailed and light-tailed distributions. Lastly, we provide examples of commonly-used distributions to facilitate understanding.

We derive that heavy/light-tailed distributions engender different performance in terms of the looseness of the resulting bounds, and we summarize the conclusions in Table \ref{conclusion}. Intuitively, the upper bound for the tail probability given by our DRO approach turns out to be less conservative under a heavy-tailed distribution compared to a light-tailed one as the true heavy tail tends to have a slow decay. Inversely, estimating tail quantiles under a light-tailed distribution is less conservative than under a heavy-tailed one.

\spacingset{1.0}
\begin{table}[htbp]
    \centering
    \begin{tabular}{c|cc}
    \toprule
    \hline
    &Heavier Tail & Lighter Tail\\
    \hline
    Estimating Tail Probabilities & Less Conservative & More Conservative\\
    \hline
    Estimating Tail Quantiles & More Conservative & Less Conservative \\
    \hline
    \bottomrule    
    \end{tabular}
    \caption{Conservativeness Performance\strut}
    \label{conclusion}
\end{table}
\spacingset{1.9}
\subsection{Abstract Formulation} 
Since it is hard to develop theoretical results in general, we focus on a specific abstract formulation where the geometric assumption is chosen as convexity, the moment region is chosen as a singleton set, and the information up to the threshold is known explicitly instead of calibrated with data. While we only consider a special case of the entire framework discussed in this paper, we note that these asymptotic results build a connection with EVT and also provide us with useful insights on conservativeness.

Given a continuous random variable $X$ with distribution function $F$ and density function $f$, we denote the right endpoint as $x_F:=\sup\{x\in\mathbb{R}:F(x)<1\}$. We assume that $F$ and $f$ are known up to a large threshold $a<x_F$. Our goal is to estimate quantities that are related to the tail region above $a$. 

We first make some assumptions on the true probability distribution. Under $1+\xi x>0$, $$H_{\xi}(x):=\left\{\begin{array}{ll}{\exp \left\{-(1+\xi x)^{-1 / \xi}\right\}} & {\text { if } \xi \neq 0} \\ {\exp \{-\exp \{-x\}\}} & {\text { if } \xi=0}\end{array}\right.$$
denotes the GEV distribution \citep{embrechtsModellingExtremalEvents1997}. 
\begin{assumption}
There exists $z<x_F$ such that on the interval $(z, x_F)$, $F$ is twice differentiable and $f$ is positive, decreasing and convex. 
\label{smoothness and shape}
\end{assumption}
\begin{assumption}$F$ belongs to the maximum domain of attraction (MDA) of $H_\xi$, i.e., $F\in MDA(H_{\xi})$. In another words, there exist normalization constants $c_n>0,d_n\in\mathbb{R}$ such that $c_n^{-1}(M_n-d_n)\stackrel{}{\Rightarrow}H_{\xi}$ for some $\xi\in\mathbb{R}$ as $n\rightarrow \infty$ where $M_n:=\max(X_1,\cdots,X_n)$ is the sample maxima and $X_1,\cdots,X_n$ are random realizations of $X$.  
\label{MDA}
\end{assumption}
We recall that there are three different types of GEV distribution: the Fr\'{e}chet distribution $\phi_{\alpha}$ with $\xi=\alpha^{-1}>0$, the Gumbel distribution $\Lambda$ with $\xi=0$ and the Weibull distribution $\psi_{\alpha}$ with $\xi=-\alpha^{-1}<0$. $\xi$ is a shape parameter to govern the tail behavior of the distribution. Indeed, according to \citet{embrechtsModellingExtremalEvents1997}, in the Fr\'{e}chet case, $x_F=\infty$ and $\lim_{x\rightarrow\infty}\frac{\bar{F}(tx)}{\bar{F}(x)}=t^{-\frac{1}{\xi}}, t>0$. Such an aymptotic behavior is also known as regularly varying at $\infty$ with index $-1/\xi$, denoted for short by $\bar{F}\in RV_{-1/\xi}$. The larger is $\xi$, the more slowly the tail of $\bar{F}$ decays, and hence the heavier the tail of $F$ is. Similarly, when $\xi=0$ with $x_F=\infty$ (Gumbel distribution), the tail is lighter than those in the Fr\'{e}chet case. 
\par 
In the case that $\xi=0$, we make an additional assumption that is broadly satisfied by the textbook distributions such as normal, Gamma and exponential distributions. Note that Assumption \ref{additional assumption} together with the above assumptions with $\xi=0$ imply that $F$ is a von Mises function.
\begin{assumption}$\lim_{x\uparrow x_F}\frac{\overline{F}(x)f'(x)}{f^2(x)}=-1.$
\label{additional assumption}
\end{assumption}

It is known from \citet{embrechtsModellingExtremalEvents1997} that if $\xi>0$, then $x_F=\infty$; if $\xi<0$, then $x_F<\infty$; if $\xi=0$, then $x_F$ can be either finite or infinite. In fact, under our assumptions, we can focus on the case that $x_F=\infty$, which is justified by the following two propositions:
\begin{proposition}\label{conserv_1}
    Suppose that $F$ satisfies Assumption \ref{smoothness and shape} and Assumption \ref{MDA} with $\xi<0$. Let $Y=1/(x_F-X)$. Use $F_Y$ and $f_Y$ to denote the distribution function and density function of $Y$. Then there exists $z_Y<\infty$ such that on $(z_Y,\infty)$, $F_Y$ is twice differentiable and $f_Y$ is positive, decreasing and convex. Moreover, $F_Y\in MDA(H_{-\xi})$.
\end{proposition}
\begin{proposition}\label{conserv_2}
    Suppose that $F$ satisfies Assumption \ref{smoothness and shape}, Assumption \ref{MDA} with $\xi=0$ and additionally, Assumption \ref{additional assumption}. Moreover, we suppose that $x_F<\infty$. Let $Y=1/(x_F-X)$. Use $F_Y$ and $f_Y$ to denote the distribution function and density function of $Y$. Then there exists $z_Y<\infty$ such that on $(z_Y,\infty)$, $F_Y$ is twice differentiable and $f_Y$ is positive, decreasing and convex. Moreover, $F_Y\in MDA(\Lambda)$ with infinite right endpoint. In addition, we have that 
    \begin{equation}
        \lim_{x\rightarrow\infty}\frac{\bar{F}_Y(x)f_Y'(x)}{f_Y^2(x)}=-1.
        \label{additional assumption for Y}
    \end{equation}
    \vspace{-1cm}
\end{proposition}
By the above two propositions, under our assumptions, if $x_F<\infty$, then we may define $Y=1/(x_F-X)$, which also satisfies our assumptions. Knowing $F(x),x\leq a$ implies that we know $F_Y$, the distribution function of $Y$, up to $1/(x_F-a)$. Therefore, we may transform the problem into an equivalent one with infinite right endpoint. From now on, we assume that $x_F=\infty$ without loss of generality and either $F\in MDA(H_{\xi})$ for some $\xi>0$, or $F\in MDA(\Lambda)$. 

First we consider estimating tail probabilities $\psi(P)=E[h(X)]$ with $h(x)=\mathbb{I}(x> b)$, i.e. $\psi(P)=\mathbb{P}(X>b)=\int_b^{\infty} f(x)\mathrm{d} x$, where $b=b(a)$ satisfies that $a\leq b<\infty$. We extract three constants $\beta,\eta,\nu$ from the known information $f(x),x\leq a$:
\begin{align*}
    \beta=\beta(a):=1-F(a)=\bar{F}(a);\  \eta=\eta(a):=f(a); \ \nu=\nu(a):=-f_+'(a).
\end{align*}
For simplicity of further use, we define 
\begin{align*}
    \mu=\mu(a):=\frac{\eta}{\nu},\sigma=\sigma(a):=\frac{2\beta}{\nu}.
\end{align*}
For $\beta,\eta,\nu>0$, we consider the RO problem formulated in the following form:
\begin{equation}
\begin{aligned}
\max_f\  & \int_{a}^{\infty}\mathbb{I}(x>b)f(x)\mathrm{d}x\\
\textit{ s.t. } & \int_a^{\infty}f(x)\mathrm{d}x=\beta,\\
& f(a)=f(a+)=\eta,\\
& f_+'(a)\ge -\nu,\\
& f \mathrm{\ convex\ for\ }x\ge a,\\
& f(x)\ge 0 \mathrm{\ for\ }x\ge a.
\end{aligned}
\label{optimization problem}
\end{equation}
We note that \eqref{optimization problem} is equivalent to $\mathfrak{P}(\mathscr{P}^2_{a,\eta,\eta,\nu}, \mathbb{I}(x>a), \{\beta\})$, so we are indeed considering a special case of the previously defined framework. From now on, the optimal value to the above problem is denoted by $z^*=z^*(a, b)$. By applying the results in \citet{lamTailAnalysisParametric2017}, we get the following theorem:
\begin{theorem}
    If $\eta^2<2\beta\nu$, then the optimal value of \eqref{optimization problem}, denoted by $z^*$, is 
    \begin{equation}
    z^*=\begin{cases}
    \frac{\nu}{2}(\sigma-\mu^2) & \text{if } \mu\leq b-a;\\
    \frac{\nu}{2}[\sigma-2(b-a)\mu+(b-a)^2] & \text{if } \mu > b-a.
    \end{cases} 
    \label{optimal value}
    \end{equation}
    \label{optimal value theorem}
\end{theorem}
To get some intuition on Theorem \ref{optimal value theorem}, note that if we draw a line from $(a,\eta)$ with slope $-\nu$, then it hits 0 at $(a+\mu,0)$. The tail extrapolation of any feasible density function must be above this straight line. On the other hand, they can be as close as possible. Thus intuitively $z^*$ is equal to $\beta$ subtracted by the area of the shaded region in Figure \ref{intuition}, which exactly coincides with the results in the theorem.
\begin{figure}[htbp] 
\centering    
\begin{subfigure}{0.45\textwidth}
\centering
\includegraphics[width=1\textwidth]{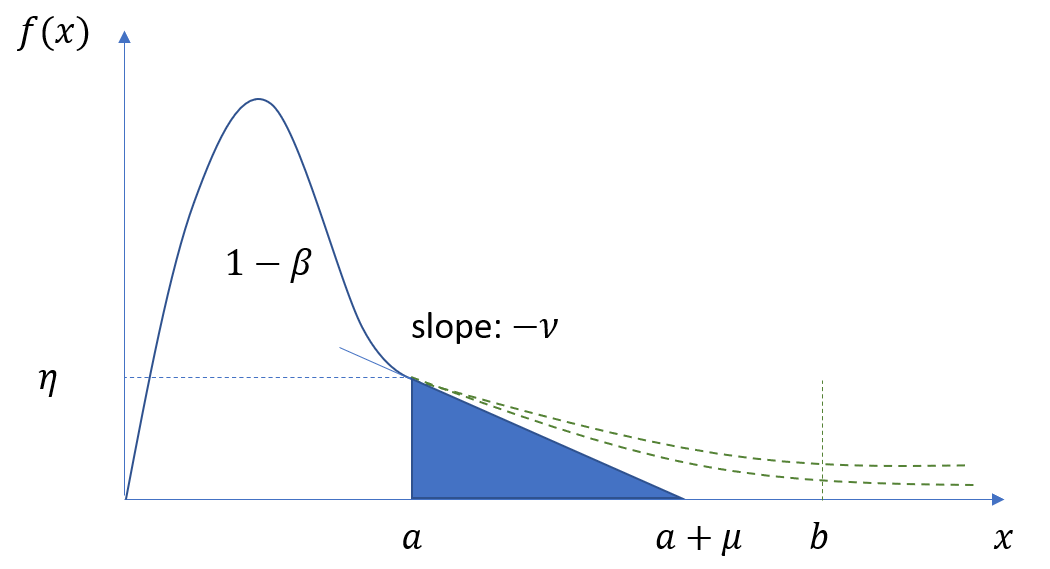}  
\caption{$b\geq a+\mu$}
\end{subfigure}
\begin{subfigure}{0.45 \textwidth}
\centering
\includegraphics[width=1\textwidth]{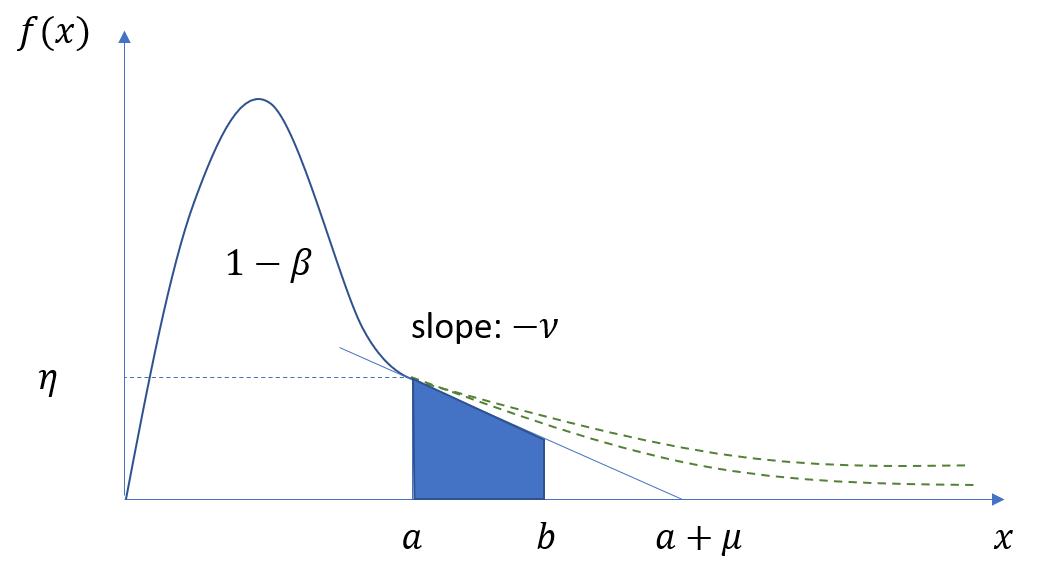} 
\caption{$b<a+\mu$}
\end{subfigure}    
\caption{Intuition for the Optimal Value $z^*$}
\label{intuition}     
\end{figure}
To quantify the conservativeness in estimating tail probabilities, we need to specify a proper function $b=b(a)$ of $a$ and consider the limit of the relative error as $a\rightarrow\infty$, which is defined as 
\begin{equation}
    \lim_{a\rightarrow\infty}\frac{z^*(a,b(a))-\bar{F}(b(a))}{\bar{F}(b(a))}.
    \label{relative error for p}
\end{equation}
Heuristically, the larger the value of \eqref{relative error for p}, the more conservative the RO approach is. We will present the selection of $b(a)$ in later subsections.

Now we consider estimating tail quantiles. Similarly, we can apply the RO approach to get a worst-case estimation. More specifically, given the information of $F(x)$ for $x\leq a$, suppose that our goal is to estimate $q=F^{-1}(p)$ where $p\geq 1-F(a)=1-\beta$. In order to get a worst-case estimation, we maximize $q$ among all the potential convex tail extrapolations. That is, the worst-case estimation for $q$ is obtained by solving 
\begin{equation}
\begin{aligned}
\max_f\  & q\\
\text{ s.t. } & \int_{-\infty}^q f(x)\mathrm{d}x=p,\\
&\int_a^{\infty}f(x)\mathrm{d}x=\beta,\\
& f(a)=f(a+)=\eta,\\
& f_+'(a)\ge -\nu,\\
& f \mathrm{\ convex\ for\ }x\ge a,\\
& f(x)\ge 0 \mathrm{\ for\ }x\ge a.
\end{aligned}
\label{quantile optimization problem}
\end{equation}
We define
\begin{equation}
    q^*:=\inf\{b\geq a:z^*(a,b)=1-p\}.
    \label{q^*}
\end{equation}
The curve in Figure \ref{plot z^*} reflects the shape of $z^*$ against $b$. For $p_1$ such that $\beta-\eta^2/(2\nu)\leq 1-p_1\leq \beta$, we may find a corresponding point $q_1^*$ such that $z^*(a,q_1^*)=1-p_1$. In particular, if $1-p_1=\beta-\eta^2/(2\nu)$, then by the definition in \eqref{q^*}, we have that $q_1^*=a+\mu$. However, for $p_2$ such that $1-p_2<\beta-\eta^2/(2\nu)$, the corresponding $q_2^*$ is defined as $\infty$. The following theorem justifies this intuitive definition of $q^*$. 
\begin{figure}[htbp]
    \centering
    \includegraphics[width=6cm]{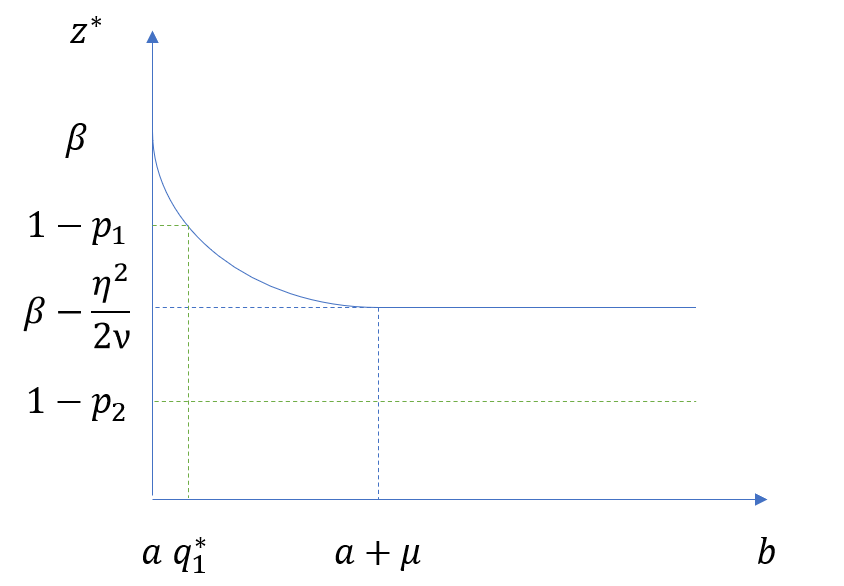}
    \caption{The Shape of $z^*$ against $b$ and the Definition of $q^*$ }
    \label{plot z^*}
\end{figure}
\begin{theorem}\label{conserv_4}
    $q^*$ defined in \eqref{q^*} is the optimal value of \eqref{quantile optimization problem}, and is expressed by
    \begin{equation}
    q^*=\begin{cases}
    a+\mu-\sqrt{\mu^2-\sigma+\frac{2(1-p)}{\nu}} & \text{if }1-\beta \le p\le 1-\beta+\frac{\eta^2}{2\nu};\\
    \infty & \text{if }p>1-\beta+\frac{\eta^2}{2\nu}.
    \end{cases}
    \label{value of q^*}
\end{equation}
    \label{quantile optimal value}
\end{theorem}

\par 
Similar to estimation of tail probabilities, we manage to find a proper function $p=p(a)$ of $a$ and compute the limit of the relative error as $a\rightarrow\infty$, which is defined as 
\begin{equation}
    \lim_{a\rightarrow\infty}\frac{q^*(p(a))-q(p(a))}{q(p(a))}. \label{relative error of q}
\end{equation}
Recall that $q=q(p(a)):=F^{-1}(p(a))$ is the true $p$-quantile. Similarly, the larger the value of \eqref{relative error of q}, the more conservative the RO approach in estimating tail quantiles is. 

\subsection{Conservativeness in Estimating Tail Probabilities}

We consider two cases as follows.

\textbf{Case 1: $\xi>0$.} Suppose that $F$ and $f$, the true distribution function and density function of $X$, satisfy Assumptions \ref{smoothness and shape} and \ref{MDA} with $\xi>0$. Since $\bar{F}\in RV_{-1/\xi}$, by the Karamata representation theorem \citep{embrechtsModellingExtremalEvents1997}, $\bar{F}$ has the following representation:
\begin{equation}
    \bar{F}(x)=c(x)\exp\left\{-\int_z^x\frac{1}{u(t)}\mathrm{d}t\right\},z<x<\infty
    \label{Karamata}
\end{equation}
where $c(x)\rightarrow c>0$, $u(x)/x\rightarrow \xi$ as $x\rightarrow\infty$. Moreover, it is known that when the threshold $a$ is sufficiently large,
\begin{align*}
P(X>x|X>a)\approx\bar{G}_{\xi;a,u(a)}(x)=\left(1+\xi\frac{x-a}{u(a)}\right)^{-\frac{1}{\xi}},x\geq a,
\end{align*}
which is exactly the mathematical foundation for the GPD method. 
\par 
Heuristically, if $P(X>x|X>a)$ is exactly equal to $\bar{G}_{\xi;a,u(a)}(x)$ for any $x\geq a$, then we get that 
\begin{align*}
    \bar{F}(x)=P(X>a)\bar{G}_{\xi;a,u(a)}(x)=\beta\left(1+\xi\frac{x-a}{u(a)}\right)^{-\frac{1}{\xi}}.
\end{align*}
Thus,
\begin{align*}
f(x)&=-\frac{\mathrm{d}}{\mathrm{d}x}P(X>x)=\frac{\beta}{u(a)}\left(1+\xi\frac{x-a}{u(a)}\right)^{-\frac{1}{\xi}-1},\\
-f'(x)&=-\frac{\mathrm{d}}{\mathrm{d}x}f(x)=\frac{(\xi+1)\beta}{u^2(a)}\left(1+\xi\frac{x-a}{u(a)}\right)^{-\frac{1}{\xi}-2}.
\end{align*}
If we substitute $x$ with $a$, then we get that 
\begin{align*}
    \eta=\frac{\beta}{u(a)},\nu=\frac{(\xi+1)\beta}{u^2(a)}.
\end{align*}
Therefore, it seems reasonable to use $\beta/u(a)$ and $(\xi+1)\beta/u^2(a)$ to approximate $\eta$ and $\nu$ respectively. Indeed, this is true as $a\rightarrow\infty$, which is justified by the following proposition:
\begin{proposition}\label{conserv_5}
    Suppose that distribution function $F$ and the corresponding density function $f$ satisfy Assumptions \ref{smoothness and shape} and \ref{MDA} with $\xi>0$, which implies that $\bar{F}$ has the representation \eqref{Karamata}. The following statements are true:
    \begin{align}\label{Frechet}
        \lim_{x\to\infty}\frac{f(x)u(x)}{\bar{F}(x)}=1;\ \lim_{x\to\infty}-\frac{f'(x)u^2(x)}{(\xi+1)\bar{F}(x)}=1.
    \end{align}
\end{proposition}
\begin{remark}
Note that this theorem implies that as $a\rightarrow\infty$,
\begin{equation}
    \eta\sim \frac{\beta}{u(a)},\nu\sim \frac{(\xi+1)\beta}{u^2(a)},\mu\sim \frac{u(a)}{\xi+1},\sigma\sim \frac{2u^2(a)}{\xi+1}.
    \label{Frechet approximations}
\end{equation}
Throughout this section, we use $g_1(a)\sim g_2(a)$ to denote $\lim_{a\to\infty}g_1(a)/g_2(a)=1$.
\end{remark}
\par 
For simplicity, we first consider $b=a+xu(a)$ where $x\geq 0$ is a fixed number. In this case, it is known that $\bar{F}(b)/\bar{F}(a)\rightarrow (1+\xi x)^{-1/\xi}$ as $a\rightarrow\infty$ \citep{embrechtsModellingExtremalEvents1997}. Using the conclusions in Proposition \ref{conserv_5} above, we can get the following theorem:
\begin{theorem}\label{conserv_6}
    Suppose that distribution function $F$ and the corresponding density function $f$ satisfy Assumptions \ref{smoothness and shape} and \ref{MDA} with $\xi>0$, which implies that $\bar{F}$ has the representation \eqref{Karamata}. $b=b(a)$ is chosen as $b=a+xu(a)$ where $x\geq 0$ is a fixed number. Then we have that
    \begin{equation}
        \lim_{a\rightarrow\infty}\frac{z^*(a,b)}{\bar{F}(b)}=\begin{cases}
        \left(1-\frac{1}{2(\xi+1)}\right)(1+\xi x)^{\frac{1}{\xi}} & \text{if } x\geq\frac{1}{\xi+1};\\
        \left(1-x+\frac{\xi+1}{2}x^2\right)(1+\xi x)^{\frac{1}{\xi}} & \text{if } x<\frac{1}{\xi+1}.
    \end{cases}
    \label{Frechet error}
    \end{equation}
    \label{Frechet results}
\end{theorem}

In fact, we may generalize the results in Theorem \ref{Frechet results} to the case that $b=a+x(a)u(a)$ where $\lim_{a\rightarrow\infty}x(a)=x_0\ge 0$. By Proposition 0.5 in \citet{resnickExtremeValuesRegular1987}, we get that the convergence $\lim_{x\rightarrow\infty}\frac{\bar{F}(tx)}{\bar{F}(x)}=t^{-\frac{1}{\xi}}$ holds locally uniformly on $(0,\infty)$. Since the limit $t^{-\frac{1}{\xi}}$ is continuous in $t$, continuous convergence holds. We have that 
\begin{align*}
\lim_{a\rightarrow\infty}\frac{a+x(a)u(a)}{a}=1+\xi x_0>0,
\end{align*}
so 
\begin{align*}
\lim_{a\rightarrow\infty}\frac{\bar{F}(a+x(a)u(a))}{\bar{F}(a)}=(1+\xi x_0)^{-\frac{1}{\xi}}.
\end{align*}
Moreover, we can also follow the discussions in the above proof to get the limit of $z^*/\bar{F}(a)$, and hence the limit of $z^*/\bar{F}(b)$. For example, if we choose $b=2a$, then we can set $x(a)=a/u(a)\rightarrow 1/\xi$. Since $1/\xi>1/(\xi+1)$, we get that
\begin{equation}
    \lim_{a\rightarrow\infty}\frac{z^*}{\bar{F}(b)}=\left(1-\frac{1}{2(\xi+1)}\right)2^{\frac{1}{\xi}}.
    \label{Frechet 2u}
\end{equation}
We note that \eqref{Frechet 2u} is decreasing in $\xi$ for $\xi>0$, and it converges to 1 as $\xi\to\infty$. This shows that the DRO approach is less conservative in estimating tail probabilities for heavier-tailed distributions.

\noindent\textbf{Case 2: $\xi=0$ and $x_F=\infty$.} Now we suppose that $F$ and $f$ satisfy Assumption \ref{smoothness and shape}, Assumption \ref{MDA} with $\xi=0$, and additionally Assumption \ref{additional assumption}. Substituting $\xi$ with 0 in \eqref{Frechet approximations}, we naturally guess that as $a\to\infty$,
\begin{equation}
    \eta\sim \frac{\beta}{u(a)}, \nu\sim \frac{\beta}{u^2(a)}, \mu\sim u(a),\sigma\sim 2u^2(a).
    \label{Gumbel approximations}
\end{equation}
This guess is justified by the proposition below. 
\begin{proposition}\label{conserv_7}
    Suppose that distribution function $F$ and the corresponding density function $f$ satisfy Assumption \ref{smoothness and shape}, Assumption \ref{MDA} with $\xi=0$ and Assumption \ref{additional assumption}. We have that $\bar{F}$ is a von Mises function with the following representation:
    \begin{equation}
        \bar{F}(x)=c\exp\left\{-\int_z^x\frac{1}{u(t)}\mathrm{d}t\right\},z<x<\infty
        \label{von Mises representation}
    \end{equation}
    where $c$ is some positive constant and $u(x)=\bar{F}(x)/f(x)$ is positive and absolutely continuous with $\lim_{x\rightarrow \infty}u'(x)=0$. In addition, the following statements hold:
   \begin{align}
        \lim_{x\rightarrow \infty}\frac{f(x)u(x)}{\bar{F}(x)}&=1;\label{eta_Gumbel}\\
        \lim_{x\rightarrow \infty}-\frac{f'(x)u^2(x)}{\bar{F}(x)}&=1.\label{nu_Gumbel}
    \end{align}
\end{proposition}

Similar to Case 1, we may choose $b=a+xu(a)$ where $x\geq 0$ is a fixed number, and then we have the following results:
\begin{theorem}\label{conserv_8}
    Suppose that distribution function $F$ and the corresponding density function $f$ satisfy Assumption \ref{smoothness and shape}, Assumption \ref{MDA} with $\xi=0$ and additionally, Assumption \ref{additional assumption}. $b=b(a)$ is chosen as $b=a+xu(a)$ where $x\geq 0$ is a fixed number. Then we have that 
    \begin{equation}
        \lim_{u\rightarrow\infty}\frac{z^*(a,b)}{\bar{F}(b)}=
        \begin{cases}
        \frac{1}{2}e^x & \text{if }x\geq1;\\
        \left(1-x+\frac{1}{2}x^2\right)e^x & \text{if }x<1.
        \end{cases}
        \label{Gumbel error}
    \end{equation}
\end{theorem}

Indeed, $\lim_{a\rightarrow\infty} u(a)/a=0$, so no matter how large is $x$, we always have that $(b-a)/a=xu(a)/a\rightarrow 0$ as $a\rightarrow\infty$. Thus if we choose $b=2a$ instead, then for any $x\geq 0$, $\lim_{a\rightarrow\infty}\bar{F}(b)/\bar{F}(a)\leq \lim_{a\rightarrow\infty}\bar{F}(a+xu(a))/\bar{F}(a)=e^{-x}$, and hence $\bar{F}(b)/\bar{F}(a)\rightarrow 0$ as $a\rightarrow\infty$. We also know that in this case $z^*(a,b)/\bar{F}(a)\rightarrow 1/2$. Therefore, in the Gumbel case, $z^*(a,2a)/\bar{F}(2a)\rightarrow \infty$ as $a\rightarrow\infty$. Compared with \eqref{Frechet 2u}, we conclude that in estimating the tail probabilities, the DRO approach is more conservative in the light-tail case than in the heavy-tail case.

\subsection{Conservativeness in Estimating Tail Quantiles}
We again consider two cases as follows.

\textbf{Case 1: $\xi>0$.} Suppose that $F$ and $f$ satisfy Assumption \ref{smoothness and shape} and Assumption \ref{MDA} with $\xi>0$. For simplicity, we choose $p=p(a)$ such that $1-p=x\beta$ where $1-\frac{1}{2(\xi+1)}< x\leq 1$ is a fixed number. Recall that $q^*=\infty$ if $p>1-\beta+\eta^2/(2\nu)$. Using \eqref{Frechet approximations}, we get that 
\begin{align*}
\lim_{a\rightarrow\infty}\left(\beta-\frac{\eta^2}{2\nu}\right)/\beta=1-\frac{1}{2(\xi+1)}.
\end{align*}
Thus the requirement $x>1-\frac{1}{2(\xi+1)}$ guarantees that  $1-p=x\beta>\beta-\eta^2/(2\nu)$ for sufficiently large $a$, and hence the RO approach can give a non-trivial estimate. \par 
\begin{theorem}\label{conserv_9}
    Suppose that distribution function $F$ and the corresponding density function $f$ satisfy Assumption \ref{smoothness and shape} and Assumption \ref{MDA} with $\xi>0$. $p$ is chosen as $1-p=x\beta$ where $1-\frac{1}{2(\xi+1)}< x\leq 1$ is a fixed number. Then we have that 
    \begin{equation}
        \lim_{a\rightarrow\infty}\frac{q^*}{q}=x^{\xi}\left[\frac{\xi}{\xi+1}\left(1-\sqrt{1-2(1-x)(\xi+1)}\right)+1\right].
        \label{Frechet error q}
    \end{equation}
\end{theorem}

We note that as $\xi$ grows, the feasible interval for $x$, i.e. $(1-\frac{1}{2(\xi+1)},1]$, becomes narrower. We also note that \eqref{Frechet error q} has an upper bound which only depends on $\xi$. Roughly speaking, for any $\xi>0$ and $1-\frac{1}{2(\xi+1)}<x\leq 1$, the value of \eqref{Frechet error q} is always bounded by $\frac{\xi}{\xi+1}+1$, which is increasing with $\xi$. Thus, for heavier-tailed distributions, estimating the tail quantiles is more conservative.

\textbf{Case 2: $\xi=0$ and $x_F=\infty$.} Now we suppose that $F$ and $f$ satisfy Assumption \ref{smoothness and shape}, Assumption \ref{MDA} with $\xi=0$, and additionally, Assumption \ref{additional assumption}. Similarly, we choose $p$ such that $1-p=x\beta$ where $1/2<x\leq 1$ is a fixed number, which guarantees that $q^*<\infty$ for sufficiently large $a$. In this case, we have the following theorem:
\begin{theorem}\label{conserv_10}
    Suppose that distribution function $F$ and the corresponding density function $f$ satisfy Assumption \ref{smoothness and shape}, Assumption \ref{MDA} with $\xi=0$ and additionally, Assumption \ref{additional assumption}. $p$ is chosen as $1-p=x\beta$ where $1/2< x\leq 1$ is a fixed number. Then we have that 
    \begin{equation}
        \lim_{a\rightarrow\infty}\frac{q^*}{q}=1.
        \label{Gumbel error q}
    \end{equation}
\end{theorem}

Compared to Theorem \ref{conserv_9}, in the light-tail case not only the feasible interval for $x$ is wider, but also the limit of the relative error is smaller, and thus it is less conservative to estimate tail quantiles.

\subsection{Examples}
For Cases 1 and 2, we respectively consider the Pareto distribution and the standard normal distribution as examples, which verify our main conclusions about the relationship between the conservativeness of the DRO approach and the heaviness of the tail.
\begin{example}[Pareto distribution]
Suppose that $X$ has a Pareto distribution with scale parameter $x_m>0$ and shape parameter $\alpha>0$. Then the tail distribution function is $\bar{F}(x)=(x_m/x)^{\alpha}$ for $x\in [x_m,\infty)$. It is known that $F\in MDA(H_{\xi})$ where $\xi=\alpha^{-1}>0$. The larger the $\alpha$ is, the smaller $\xi$ is, and the lighter is the tail. In estimating tail probabilities, we choose $b=ka$ where $k\geq 1$. Then by the discussion following Theorem \ref{Frechet results}, we get that the limit of relative error is
\begin{align*}
\lim_{a\rightarrow\infty}\frac{z^*(a,b(a))-\bar{F}(b(a))}{\bar{F}(b(a))}=
\begin{cases}
\left(1-\frac{1}{2(\xi+1)}\right)k^{\frac{1}{\xi}}-1 &\text{ if }k\geq 1+\frac{\xi}{\xi+1};\\
\left(1-\frac{k-1}{\xi}+\frac{(\xi+1)(k-1)^2}{2\xi^2}\right)k^{\frac{1}{\xi}}-1 & \text{ if }k<1+\frac{\xi}{\xi+1}.
\end{cases}
\end{align*}
In estimating tail quantiles, we choose $1-p=x\beta$ where $1-\frac{1}{2(\xi+1)}<x\leq 1$. Then by Theorem \ref{conserv_9}, we get that the limit of relative error is 
\begin{align*}
\lim_{a\rightarrow\infty}\frac{q^*(p(a))-q(p(a))}{q(p(a))}=x^{\xi}\left[\frac{\xi}{\xi+1}\left(1-\sqrt{1-2(1-x)(\xi+1)}\right)+1\right]-1.
\end{align*}
Figure \ref{Pareto} $(a)$ shows how the limit of relative error changes with $k$ and $\alpha$ in estimating probabilities while $(b)$ shows the change with $x$ and $\alpha$ in estimating quantiles. It can be seen that in estimating probabilities, heavier tail gives less conservativeness; however, in estimating quantiles, for the heavier tail, the limit of relative error is larger and meanwhile the range of $x$ that gives non-trivial estimate $q^*$ is smaller, so the DRO approach is more conservative.
\par 
\begin{figure}[htbp] 
\centering    
\begin{subfigure}{0.4\textwidth}
\centering
\includegraphics[width=1\textwidth]{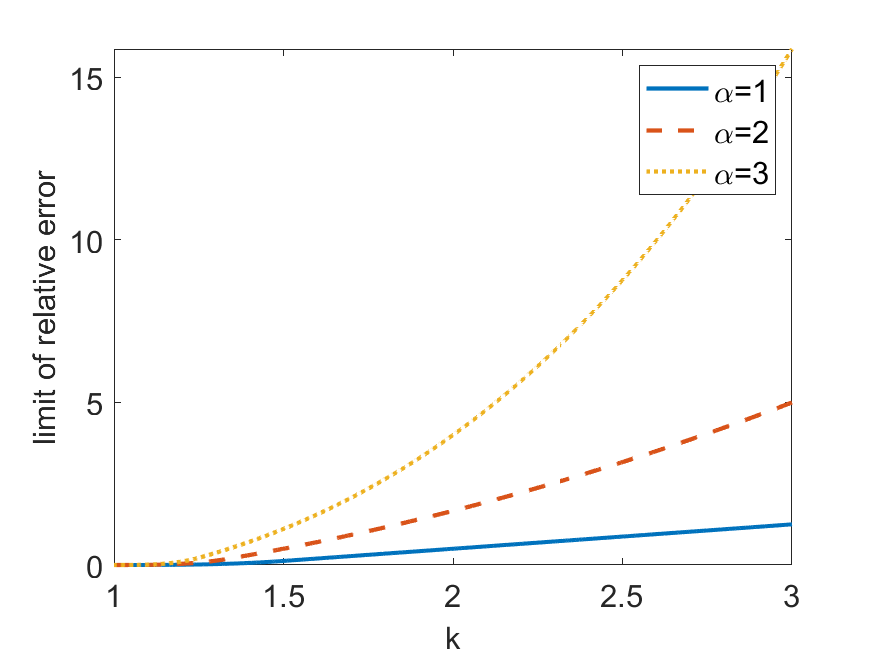}  
\caption{Estimating Tail Probabilities}
\end{subfigure}
\begin{subfigure}{0.4\textwidth}
\centering
\includegraphics[width=1\textwidth]{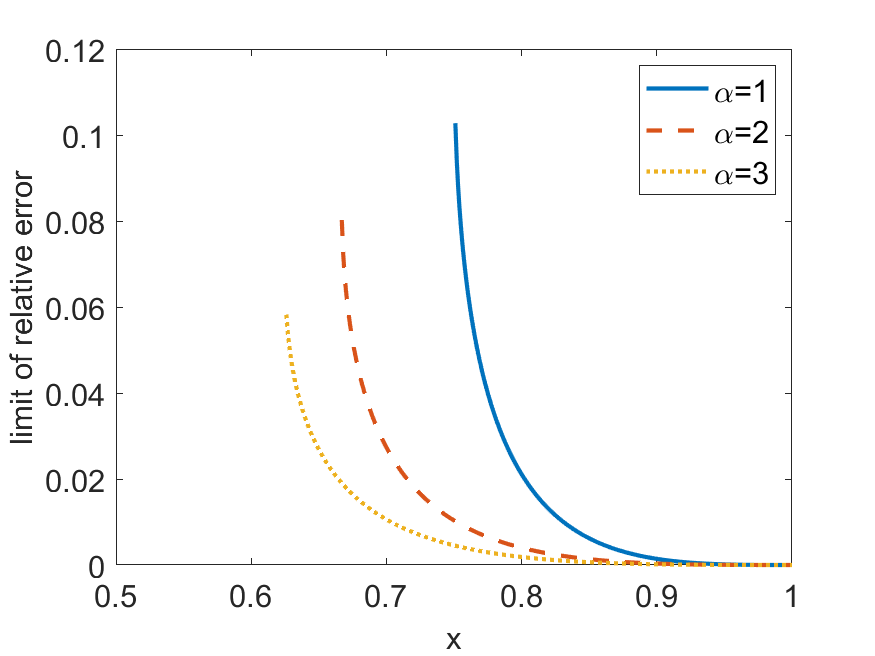}     
\caption{Estimating Tail Quantiles}
\end{subfigure}
\caption{Limit of Relative Errors for the Pareto Distribution}
\label{Pareto}     
\end{figure}
\end{example}
\begin{example}[Standard normal distribution]
Suppose that $X$ has a standard normal distribution. It is known that the distribution function $F\in MDA(H_{\xi})$ where $\xi=0$. In particular, $F$ is a von Mises function with auxiliary function $u(x)=\bar{F}(x)/f(x)\sim 1/x,x\rightarrow\infty$. In estimating tail probabilities, if $b=a+xu(a)$, we always have $b-a\rightarrow 0$ as $a\rightarrow \infty$. Nevertheless, the value in \eqref{Gumbel error} still grows exponentially in $x$. Therefore, the DRO approach is very conservative in estimating probabilities. In estimating tail quantiles, we may choose $1-p=x\beta$ for $1/2<x\leq 1$. This range is larger than that for any $\xi>0$, and yet the limit of relative error is always equal to 0, which means that the DRO approach is less conservative than in the Fr\'{e}chet case.  
\end{example}

\section{Optimization Reformulations and Solution Tractability}\label{sec:computation}
In this section, we focus mainly on $\psi(P)$ as an expectation, i.e., $\mathbb E_P[h(X)]$ for some function $h$. When the target $\psi(P)$ is a quantile, the analysis can be reduced to the expectation case because $\psi(P)= \min \{q:\mathbb{P}(X\leq q)\geq p\}$ can be written as $\psi(P)= \min \{q:\mathbb{E}[h(X)]\geq p-\mathbb{P}(X\leq a)\}$ where $h(X)=\mathbb{I}(a\leq X\leq q)$, and the non-tail part $\mathbb{P}(X\leq a)$ is supposedly handleable by standard statistical tools such as via the empirical distribution. Therefore, in finding the quantile, one could consider line search methods such as bisection on $q$ to obtain the minimum $q$ such that the following holds:
\begin{align*}
    \max_{P\in \mathbb F(\mathfrak{P}(\mathscr{P},\bm{g},\mathbb{S}))} \mathbb{E}[\mathbb{I}(a\leq X\leq q)]\geq p-\mathbb{P}(X\leq a).
\end{align*}

There exists some challenges in solving \eqref{general_framework} as it is an infinite-dimensional optimization problem, and the geometric shape constraint imposes extra complication. In the following subsections, we show how we leverage techniques in the optimization literature to reduce \eqref{general_framework} into a moment problem which can then be dualized into standard solvable program classes.


\subsection{Transformation to Moment Problems }


Given $\psi(P)=\mathbb E_P[h(X)]$, moment constraints configured by $\bm g$ and $\mathbb S$, and shape information $\mathscr{P}$ as either $\mathscr{P}^1_{a,\eta}$ or $\mathscr{P}^2_{a,\underline{\eta},\bar{\eta},\nu}$, \eqref{general_framework} can be written as
\begin{align}\label{general_framework1}
\begin{split}
\mathfrak{P}(h,\mathscr{P},\bm{g},\mathbb{S}): \quad \max_{P}\ \mathbb E_P[h(X)] \text{\ \ subject to\ }\ P\in\mathscr{P},\ \mathbb{E}_P[ \bm{g}(X)]\in\mathbb{S}
\end{split}
\end{align}
where we have now highlighted the role of $h$ in the notation $\mathfrak{P}(h,\mathscr{P},\bm{g},\mathbb{S})$. Our first step in handling \eqref{general_framework1} is to convert it to an equivalent moment problem:
\begin{align}
\mathfrak{M}_a(H,\bm{G},\mathbb{S}'): \ \max_{X\sim Q}\   \mathbb{E}_Q [H(X)]\text{\ \ subject to\ }\  \mathbb{E}_Q[ \bm{G}(X)]\in\mathbb{S}',\ Q\in \mathscr{P}[a,\infty)
\label{moment problem}
\end{align}
where the set $\mathbb{S}'$ is derived from $\mathbb{S}$, $\mathscr{P}[a,\infty)$ denotes the class of probability distributions with support $[a,\infty)$, and function $H$ and vector function $\bm{G}$ are derived from $h$ and $\bm{g}$ respectively. More precisely, we have the following result. 

\begin{theorem}
\label{thm:mon}
Suppose $h$ is bounded and each $g_i$ is bounded from below with support $x\geq a$. 

\begin{enumerate}
\item $\mathfrak{P}(h,\mathscr{P}^{1}_{a,\eta},\bm{g},\mathbb{S})$ is equivalent to $\mathfrak{M}_a\big(\eta\tilde{H},\eta \tilde{\bm{G}},\mathbb{S}\big)$
where $\tilde{H}(x)=\int_{a}^x h(u)du$ and $\tilde{G}_i(x)=\int_{a}^x g_i(u)du$ for each component $i$ in the vector function $\tilde{\bm{G}}$. A bijective transformation between a feasible solution $P$ of $\mathfrak{P}(h,\mathscr{P}^{1}_{a,\eta},\bm{g},\mathbb{S})$ and $Q$ of $\mathfrak{M}_a(\eta\tilde{H},\eta \tilde{\bm{G}},\mathbb{S})$ is given by $P'_+(x)=\eta(1-Q(x))$ (viewing $P$ and $Q$ as distribution functions).

\item $\mathfrak{P}(h,\mathscr{P}^{2}_{a,\underline{\eta},\bar{\eta},\nu},\bm{g},\mathbb{S})$ is equivalent to $\mathfrak{M}_a\Big(\nu H,\big(\nu (x-a),\nu \bm{G}^\top\big)^\top,\mathbb{S}_{R}(\underline{\eta},\bar{\eta})\times\mathbb{S}\Big)$
where $H(x)=\int_{a}^x\int_{a}^u h(v)dvdu$ and $G_i(x)=\int_{a}^x\int_{a}^u g_i(v)dvdu$ for each component $i$ in the vector function $\bm{G}$. A bijective transformation between a feasible solution $P$ in $\mathfrak{P}(h,\mathscr{P}^{2}_{a,\underline{\eta},\bar{\eta},\nu}, \bm{g},\mathbb{S})$ and $Q$ in $\mathfrak{M}_a\Big(\nu H,\big(\nu (x-a),\nu \bm{G}^\top\big)^\top,\mathbb{S}_{R}(\underline{\eta},\bar{\eta})\times\mathbb{S} \Big)$ is given by $-P^{(2)}_+(x)=\nu(1-Q(x))$, where $P^{(2)}_+$ denotes the second-order right derivative of $P$ (viewing $P$ and $Q$ as distribution functions).
\end{enumerate}
\end{theorem}

To illustrate Theorem \ref{thm:mon}, consider for example a tail interval probability as the objective, in which $h(x) = \mathbb{I}(L\leq x \leq R)$ for some given number $L,R$. We have 
\begin{align*}
\tilde{H}(x) &= (x-L)\mathbb{I}(L\leq x\leq R)+(R-L)\mathbb{I}(x\geq R); \\ H(x) &= \frac{1}{2}(x-L)^2\mathbb{I}(L\leq x\leq R)+(R-L)(x-\frac{R+L}{2})\mathbb{I}(x\geq R).
\end{align*}
If $g_{i}(x)=x^{i-1}\mathbb{I}(x\geq a),i\geq 1$, we have 
\begin{align*}
\tilde{G}_i(x)&=(x^{i}-a^{i})/i,\ G_i(x)=x^{i+1}/(i^2+i)-a^{i}x/i+a^{i+1}/(i+1).
\end{align*}

Theorem \ref{thm:mon} shows that \eqref{general_framework} is equivalent to a moment-constrained program, by identifying the decision variable (as a probability distribution) via a one-to-one map with a probability distribution function with support on $[a,\infty)$. We give two distinct methods to prove Theorem \ref{thm:mon}. The first one is an integration-by-parts technique that involves replacing the distribution function in the decision variable by its derivative. This approach is built on \cite{lamTailAnalysisParametric2017} that considers a more restrictive formulation. The second method is Choquet's theory, which in convex analysis implies the representation of any point in a compact convex set by a mixture of its extreme points. In our context, the stipulated class of probability distributions can be written as a mixture representation of simpler distributions. This then allows one to rewrite the optimization problem in terms of the mixture distribution as the decision variable, and subsequently remove the shape constraint. Though the main idea of this method follows from some existing DRO works \citep{popescuSemidefiniteProgrammingApproach2005,van2016generalized,li2019ambiguous}, our theorem allows for general moment set $\mathbb{S}$ and inequality constraints for the accompanying parameters in the shape constraint $\mathscr{P}^1_{a,\eta}$ or $\mathscr{P}^2_{a,\underline{\eta},\bar{\eta},\nu}$, instead of singleton used in these works.

\subsection{Dualization to Semidefinite Programs}
The moment problem \eqref{moment problem} has a finite number of constraints but an infinite-dimensional decision variable. In the following, we transform it into a dual program with finite-dimensional decision variable but an infinite number of constraints, which we can further reduce to a more tractable formulation. Since our set $\mathbb S'$ in the moment problem \eqref{moment problem} generally consists of both ellipsoid and rectangular sets, we write our theorem in this generality as well. First, we introduce the following assumption for guaranteeing strong duality:
\begin{assumption}[Slater Condition]\label{slater_condition}
In $\mathfrak{M}_a(H,\mathbf{G},\mathbb{S})$ where $(G_j)_{j=1}^d$ are measurable functions, there exists $P\in \mathscr{P}[a,\infty)$ such that $\mathbb{E}_P[\bm{G}(X)]$ lies in the interior of $\mathbb{S}$. 
\end{assumption}

We have the following duality result:


\begin{theorem}\label{p2d}
The dual problem of $\mathfrak{M}_a\big( H,(\bm{G_1}^\top,\bm{G_2}^\top)^\top,\mathbb{S}_{E}(\bm{\mu},\bm{\Sigma},r)\times \mathbb{S}_{R}(\underline{\bm{\mu}},\bar{\bm{\mu}})\big)$, with given constant values $\bm{\mu},\bm{\Sigma},\bm{\bar{\mu}},\bm{\underline{\mu}},a$, function $H$ and vector functions $\bm{G_1},\bm{G_2}$, is 
\begin{subequations}\label{mix_F_d}
\begin{align}
\min_{\substack{\kappa,\bm{\lambda_1}\geq 0,\bm{\lambda_2}\geq 0,\\   \|\bm{u}\|_2\leq \lambda} } \ \ \  &\kappa+\lambda+r^{-1/2}\bm{u}^\top\bm{\Sigma}^{-1/2}\bm{\mu}+\bm{\lambda_1}^\top \bm{\bar{\mu}}-\bm{\lambda_2}^\top \bm{\underline{\mu}}\label{mix_F_d_a}\\
\text{ s.t. } \ \  &- H(x)+r^{-1/2}\bm{u}^\top \bm{\Sigma}^{-1/2}\bm{G_1}(x)+(\bm{\lambda_1}-\bm{\lambda_2})^\top\bm{G_2}(x)+\kappa \geq 0 ,\ \forall x\geq a. \label{mix_F_d_b}
\end{align}
\end{subequations}
Here, $\kappa,\bm{u},\lambda,\bm{\lambda_1},\bm{\lambda_2}$ are decision variables. The optimal value of \eqref{mix_F_d} is at least that of $\mathfrak{M}_a\big( H,(\bm{G_1}^\top,\bm{G_2}^\top)^\top,\mathbb{S}_{E}(\bm{\mu},\bm{\Sigma},r)\times \mathbb{S}_{R}(\underline{\bm{\mu}},\bar{\bm{\mu}})\big)$ and, under Assumption \ref{slater_condition}, they attain equality.
\end{theorem}

Theorem \ref{p2d} shows the dual program of \eqref{moment problem} when $\mathbb{S}$ is a combination of ellipsoid and rectangle. The Slater condition in Assumption \ref{slater_condition} that ensures strong duality can be checked routinely case-by-case. Moreover, even if this condition does not hold, the ultimate statistical guarantees provided by Theorems \ref{proposition:sg1} and \ref{proposition:sg2} are still valid since weak duality allows us to obtain a more conservative bound. Theorem \ref{p2d} follows immediately from the duality theory of conic programs (e.g.,  \cite{shapiro2001duality}).

Putting Theorems \ref{thm:mon} and \ref{p2d} together, we can convert \eqref{general_framework1} with $(\mathbb{S},\mathscr{P})$ being $(\mathbb{S}_E,\mathscr{P}^1_{a,\eta})$, $(\mathbb{S}_R,\mathscr{P}^1_{a,\eta})$, $(\mathbb{S}_E,\mathscr{P}^2_{a,\underline{\eta},\bar{\eta},\nu})$ or $(\mathbb{S}_R, \mathscr{P}^2_{a,\underline{\eta},\bar{\eta},\nu})$ into the following dual program. 

\begin{corollary}\label{coro_4cases}
Given any functions $h,\bm{g}$ satisfying the assumptions in Theorem \ref{thm:mon}, parameters $\eta,\nu, a, \bm{\mu},\underline{\bm{\mu}},\bar{\bm{\mu}},\bm{\Sigma}\succ 0$ and $\tilde{H}(x)=\int_{a}^x h(u)du$, $H(x)=\int_{a}^x\int_{a}^u h(v)dvdu$, $\tilde{G}_i(x)=\int_{a}^x g_i(u)du$, $G_i(x)=\int_{a}^x\int_{a}^u g_i(v)dvdu $ for each component $i$ of $\bm{g}$, we have 

1. For problem $\mathfrak{P}(h,\bm{g},\mathbb{S}_{E}(\bm{\mu},\bm{\Sigma},r),\mathscr{P}^{1}_{a,\eta})$, the dual of the converted moment problem is 
\begin{align*}
\min_{\kappa,\|\bm{u}\|_2\leq \lambda }\kappa+\lambda+r^{-1/2}\bm{u}^\top\bm{\Sigma}^{-\frac{1}{2}}\bm{\mu} \textup{ s.t. } -\eta \tilde{H}(x)+\eta\bm{u}^\top \bm{\Sigma}^{-\frac{1}{2}}\tilde{\bm{G}}(x)+\kappa \geq 0 ,\ \forall x\geq a.
\end{align*}

2. For problem $\mathfrak{P}(h,\bm{g},\mathbb{S}_{R}(\underline{\bm{\mu}},\bar{\bm{\mu}}),\mathscr{P}^{1}_{a,\eta})$, the dual of the converted moment problem is 
\begin{align*}
\min_{\kappa,\bm{\lambda_1}\geq 0,\bm{\lambda_2}\geq 0 } \kappa+\bm{\lambda_1}^\top \bm{\bar{\mu}}-\bm{\lambda_2}^\top \bm{\underline{\mu}},\textup{ s.t. }-\eta \tilde{H}(x)+\eta(\bm{\lambda_1}-\bm{\lambda_2})^\top\tilde{\bm{G}}(x)+\kappa \geq 0 ,\ \forall x\geq a.
\end{align*}

3. For problem $\mathfrak{P}(h,\bm{g},\mathbb{S}_{E}(\bm{\mu},\bm{\Sigma},r),\mathscr{P}^{2}_{a,\underline{\eta},\bar{\eta},\nu})$, the dual of the converted moment problem is  
{\footnotesize
\begin{align*}
\min_{\substack{\kappa,\|\bm{u}\|_2\leq \lambda \\ \delta_1\geq 0,\delta_2\geq 0 } } \kappa+\lambda+r^{-1/2}\bm{u}^\top\bm{\Sigma}^{-1/2}\bm{\mu}+\delta_1 \bar{\eta}-\delta_2 \underline{\eta}, \textup{ s.t. }-\nu H(x)+\nu\bm{u}^\top \bm{\Sigma}^{-1/2}\bm{G}(x)+\nu(\delta_1-\delta_2)(x-a)+\kappa \geq 0 ,x\geq a.
\end{align*}}

4. For problem $\mathfrak{P}(h,\bm{g},\mathbb{S}_{R}(\underline{\bm{\mu}},\bar{\bm{\mu}}),\mathscr{P}^{2}_{a,\underline{\eta},\bar{\eta},\nu})$, the dual of the converted moment problem is
{\footnotesize\begin{align*}
\min_{\substack{\kappa,\delta_1\geq 0,\delta_2\geq 0 ,\\ \bm{\lambda_1}\geq 0,\bm{\lambda_2}\geq 0} } \kappa+\bm{\lambda_1}^\top \bm{\bar{\mu}}-\bm{\lambda_2}^\top \bm{\underline{\mu}}+\delta_1 \bar{\eta}-\delta_2 \underline{\eta}, \textup{ s.t. } -\nu H(x)+\nu(\bm{\lambda_1}-\bm{\lambda_2})^\top\bm{G}(x)+\nu(\delta_1-\delta_2)(x-a)+\kappa \geq 0 ,\ \forall x\geq a.
\end{align*}}
In each case, the optimal value of the dual problem is at least that of the corresponding $\mathfrak{P}(h,\bm{g},\mathbb{S},\mathscr{P})$ and, under Assumption \ref{slater_condition} applied to the corresponding moment problem, they attain equality.
\end{corollary}

Program \eqref{mix_F_d} and the specialized versions in Corollary \ref{coro_4cases} have infinite numbers of constraints, with \eqref{mix_F_d_b} being a condition for any $x\geq a$. We can convert \eqref{mix_F_d} into a semidefinite program under suitable assumptions.
\begin{theorem}\label{d2sdp}
If the constraint \eqref{mix_F_d_b} can be written as a series of polynomial inequalities, i.e.,
\begin{equation}
\label{polying}
\begin{aligned}
\sum\limits_{i=0}^{k}y_{1i} x^i\geq 0, \ \forall x\geq a \textup{\ \ or\ \ } \sum\limits_{i=0}^{k}y_{2i} x^i\geq 0, \ \forall b\leq x\leq c,
\end{aligned}
\end{equation}
then program \eqref{mix_F_d} is equivalent to a mixed semidefinite and second-order cone program (SDP-SOCP) 
\begin{equation}\label{thm::socp-sdp_expr}
\begin{aligned}
\min_{\substack{\bm{u},\lambda,\kappa,\bm{\lambda_1}\geq 0,\bm{\lambda_2}\geq 0,\\ \bm{V}\succeq 0,\bm{W}\succeq 0, \|\bm{\Sigma}^{1/2}\bm{u}\|_2 \leq \lambda}} \ \ \ \ \ \ \   &\kappa+\lambda+\bm{u}^\top\bm{\mu}+\bm{\lambda_1}^\top \bm{\bar{\mu}}-\bm{\lambda_2}^\top \bm{\underline{\mu}}\\
\textup{ s.t. }\sum\limits_{i,j:i+j=2l-1} &v_{ij}=\sum\limits_{i,j:i+j=2l-1} w_{ij} =0, \ \ l=1,\cdots,k,\\
\sum\limits_{i,j:i+j=2l} &v_{ij}=\sum\limits_{r=l}^k \binom{r}{l}y_{1r} a^{r-l},\ \ l=0,\cdots,k,\\
\sum\limits_{i,j:i+j=2l} &w_{ij}=\sum\limits_{m=0}^l \sum\limits_{r=m}^{k+m-l} \binom{r}{m}\binom{k-r}{l-m}y_{2r} b^{r-m}c^{m}, \ \ l=0,\cdots,k
\end{aligned}
\end{equation}
where $\bm{V}=[v_{ij}]_{i,j=0,\cdots,k}$ and $\bm{W}=[w_{ij}]_{i,j=0,\cdots,k}$.
\end{theorem}
Theorem \ref{d2sdp} shows that problem \eqref{moment problem} can be formulated into a tractable SDP-SOCP program when $H(x)$ and $\bm{G}(x)$ belong to the polynomial function class. If $h(x)$ and $\bm{g}(x)$ are indicator functions or piecewise polynomial functions, then the polynomial class condition is satisfied which guarantees the tractability. Theorem \ref{d2sdp} can be proved using the semidefinite representation of moments (e.g., \cite{lasserre2009moments,bertsimasOptimalInequalitiesProbability2005}).


In Section \ref{sec:numerics}, we use Corollary \ref{coro_4cases} and Theorem \ref{d2sdp} to produce our numerical results in estimating extremal quantities. In Appendix \ref{sec:sensitivity} we also present some sensitivity analysis tools for our DRO problems.



\section{Numerical Results}
\label{sec:numerics}
We illustrate the numerical performance of our DRO framework and compare with conventional EVT tools. In our experiment, we consider synthetic data  of size $500$.
The quantities of interest are tail interval probabilities and quantiles, which we will specify later. 
In the experiments, we generate samples from ``true"  distributions that range from light-tailed to heavy-tailed, including 1) Gamma distribution with shape parameter $0.5$ and scale parameter $1$, 2) log-normal distribution with mean parameter 0 and standard deviation parameter 1, 3) Pareto distribution with shape parameter $1.5$ and scale parameter $1$. We aim to obtain a one-sided $95\%$ confidence upper bound by using program \eqref{general_framework}. 
To approximate the coverage probability, we repeat each experiment $200$ times, from which we would also output the sample mean of the estimated confidence bounds. Calibration of parameters is conducted by bootstrapping with a resample size $500$, where for densities and their derivatives we use the standard kernel estimator from the R package \textit{ks}.

We conduct experiments with different 1) shape constraints, 2) moment constraints 3) cutoff thresholds $a$, 4) objective functions, and also compare with POT. These experiments aim to: 1) validate our methodology by demonstrating how it generates valid confidence bounds under a wide range of settings, as supported by the statistical guarantees in Theorems \ref{proposition:sg1} and \ref{proposition:sg2}; 2) provide guidance for users in implementing our approach; 3) compare our approach against existing methods like POT. Sections \ref{sec:exp_constraints}-\ref{sec:exp_objective} will discuss results pertinent to the first two goals, while Section \ref{sec:exp_POT} will focus on the third goal above.

In the following tables, we use $(D,\chi^2)$ to denote the setting of $D$-th order monotonicity and ellipsoid moment constraint depicted in Section \ref{subsec:optimization_objective_constraints} where $D\in \{0,1, 2\}$. Similarly, $(D,\text{KS})$ denotes the setting of $D$-th order monotonicity and rectangular moment constraint. 


\subsection{Selection of Shape and Moment Constraints}\label{sec:exp_constraints}

In Table \ref{tb1_tpe}, we consider the estimation of tail interval probabilities $\mathbb{P}(q_{0.99}\leq X\leq q_{0.995})$ using the synthetic data set from three distributions, where $q_{0.99}$ and $q_{0.995}$ are theoretical $99^{th}$ and $99.5^{th}$-percentile respectively. The threshold $a$ is chosen as the $70^{th}$ sample percentile of this synthetic data set. In this and the following tables, the ``Upper Bound'' and ``Coverage Probability'' columns respectively show the sample mean and the ratio of coverage of the upper confidence bounds in the 200 repetitions, and the ``Relative Ratio'' column is defined as the mean upper confidence bound divided by the true value. The confidence intervals are given by $sample\ mean\pm 1.96*sample\ standard\ deviation/\sqrt{200}$. We note that in some cases the average coverage probability is 1, i.e., the confidence bound covers the truth in all the 200 repetitions, in which case the sample standard deviation as well as the confidence interval width are 0.

\begin{table}[!ht]
    \centering\scalebox{0.6}{\begin{tabular}{cc|ccc}
    \toprule
    \hline
    Data Source & Constraint Setting & Relative Ratio & Upper Bound & Coverage Probability\\\hline
\multirow{6}{*}{Gamma}&$(0,\chi^2)$& $14.46\ (\pm0.57)$ & $7.23\times 10^{-2}\ (\pm 2.85\times 10^{-3})$ & $1.000\ (\pm 0.000)$\\
&$(1,\chi^2)$& $5.36\ (\pm0.17)$ & $2.68\times 10^{-2}\ (\pm 8.26\times 10^{-4})$ & $1.000\ (\pm 0.000)$\\
&$(2,\chi^2)$& $3.05\ (\pm0.08)$ & $1.53\times 10^{-2}\ (\pm 4.12\times 10^{-4})$ & $1.000\ (\pm 0.000)$\\
&$(0,\text{{KS}})$& $14.17\ (\pm0.12)$ & $7.08\times 10^{-2}\ (\pm 5.88\times 10^{-4})$ & $1.000\ (\pm 0.000)$\\
&$(1,\text{{KS}})$& $6.81\ (\pm0.07)$ & $3.40\times 10^{-2}\ (\pm 3.39\times 10^{-4})$ & $1.000\ (\pm 0.000)$\\
&$(2,\text{{KS}})$& $4.12\ (\pm0.04)$ & $2.06\times 10^{-2}\ (\pm 2.04\times 10^{-4})$ & $1.000\ (\pm 0.000)$\\\hline 
\multirow{6}{*}{Lognorm}&$(0,\chi^2)$& $16.62\ (\pm0.80)$ & $8.31\times 10^{-2}\ (\pm 3.99\times 10^{-3})$ & $1.000\ (\pm 0.000)$\\
&$(1,\chi^2)$& $6.52\ (\pm0.26)$ & $3.26\times 10^{-2}\ (\pm 1.32\times 10^{-3})$ & $1.000\ (\pm 0.000)$\\
&$(2,\chi^2)$& $3.98\ (\pm0.14)$ & $1.99\times 10^{-2}\ (\pm 7.11\times 10^{-4})$ & $1.000\ (\pm 0.000)$\\
&$(0,\text{{KS}})$& $14.08\ (\pm0.13)$ & $7.04\times 10^{-2}\ (\pm 6.49\times 10^{-4})$ & $1.000\ (\pm 0.000)$\\
&$(1,\text{{KS}})$& $7.98\ (\pm0.09)$ & $3.99\times 10^{-2}\ (\pm 4.33\times 10^{-4})$ & $1.000\ (\pm 0.000)$\\
&$(2,\text{{KS}})$& $5.09\ (\pm0.05)$ & $2.55\times 10^{-2}\ (\pm 2.72\times 10^{-4})$ & $1.000\ (\pm 0.000)$\\\hline 
\multirow{6}{*}{Pareto}&$(0,\chi^2)$& $20.90\ (\pm1.68)$ & $1.05\times 10^{-1}\ (\pm 8.39\times 10^{-3})$ & $1.000\ (\pm 0.000)$\\
&$(1,\chi^2)$& $9.85\ (\pm0.76)$ & $4.92\times 10^{-2}\ (\pm 3.78\times 10^{-3})$ & $0.995\ (\pm 0.010)$\\
&$(2,\chi^2)$& $6.83\ (\pm0.44)$ & $3.41\times 10^{-2}\ (\pm 2.21\times 10^{-3})$ & $0.995\ (\pm 0.010)$\\
&$(0,\text{{KS}})$& $14.13\ (\pm0.13)$ & $7.07\times 10^{-2}\ (\pm 6.30\times 10^{-4})$ & $1.000\ (\pm 0.000)$\\
&$(1,\text{{KS}})$& $9.76\ (\pm0.10)$ & $4.88\times 10^{-2}\ (\pm 5.15\times 10^{-4})$ & $1.000\ (\pm 0.000)$\\
&$(2,\text{{KS}})$& $6.81\ (\pm0.07)$ & $3.40\times 10^{-2}\ (\pm 3.65\times 10^{-4})$ & $1.000\ (\pm 0.000)$\\
    \hline
    \bottomrule
    \end{tabular}}\caption{Tail probability estimation under different constraint settings. The true value is 0.005.}
    \label{tb1_tpe}
\end{table}

In Table \ref{tb2_qe}, we consider the estimation of quantile $q_{0.99}$ using the synthetic data set from three distributions. The threshold is still chosen as the $70^{th}$ sample percentile. The true value of the $99^{th}$-quantile of each distribution is shown in the table. We note that $(D,\text{KS})$ cannot always obtain a valid quantile estimation due to the possible assignment of probability mass at $\infty$, so we do not include these settings in this table.

\begin{table}[ht]
    \centering\scalebox{0.6}{\begin{tabular}{cc|ccc}
    \toprule
    \hline
    Data Source & Constraint Setting & Relative Ratio & Upper Bound & Coverage Probability \\\hline
\multirow{3}{*}{Gamma w. true quantile point 3.32.}&$(0,\chi^2)$& $2.13\ (\pm0.04)$ & $7.07\times 10^{0}\ (\pm 1.34\times 10^{-1})$ & $1.000\ (\pm 0.000)$\\
&$(1,\chi^2)$& $1.50\ (\pm0.03)$ & $4.96\times 10^{0}\ (\pm 9.66\times 10^{-2})$ & $1.000\ (\pm 0.000)$\\
&$(2,\chi^2)$& $1.39\ (\pm0.03)$ & $4.62\times 10^{0}\ (\pm 9.20\times 10^{-2})$ & $0.995\ (\pm 0.010)$\\\hline 
\multirow{3}{*}{Lognorm w. true quantile point 10.24.}&$(0,\chi^2)$& $2.54\ (\pm0.10)$ & $2.60\times 10^{1}\ (\pm 1.05\times 10^{0})$ & $1.000\ (\pm 0.000)$\\
&$(1,\chi^2)$& $1.80\ (\pm0.07)$ & $1.84\times 10^{1}\ (\pm 7.40\times 10^{-1})$ & $1.000\ (\pm 0.000)$\\
&$(2,\chi^2)$& $1.69\ (\pm0.07)$ & $1.73\times 10^{1}\ (\pm 6.97\times 10^{-1})$ & $0.990\ (\pm 0.014)$\\\hline 
\multirow{3}{*}{Pareto w. true quantile point 21.54.}&$(0,\chi^2)$& $6.05\ (\pm1.33)$ & $1.30\times 10^{2}\ (\pm 2.87\times 10^{1})$ & $1.000\ (\pm 0.000)$\\
&$(1,\chi^2)$& $3.92\ (\pm0.71)$ & $8.45\times 10^{1}\ (\pm 1.54\times 10^{1})$ & $0.990\ (\pm 0.014)$\\
&$(2,\chi^2)$& $3.59\ (\pm0.61)$ & $7.73\times 10^{1}\ (\pm 1.31\times 10^{1})$ & $0.985\ (\pm 0.017)$\\
    \hline
    \bottomrule
    \end{tabular}}\caption{Quantile estimation under different constraint settings. }
    \label{tb2_qe}
\end{table}

Regarding the selection of shape constraints, for tail probability estimation problem in Table \ref{tb1_tpe}, as $D$ increases, we observe a decreasing upper bound and confidence interval width. For example, with Gamma data in Table \ref{tb1_tpe}, the upper bound decreases from $7.23\times 10^{-2}$ to $1.53\times 10^{-2}$ (the true value is $5.00\times 10^{-3}$) and the confidence interval width decreases from $2.85\times 10^{-3}$ to $2.04\times 10^{-4}$ as the constraint setting changes from $(0,\chi^2)$ to $(2,\chi^2)$, leading to a tighter result. Similarly, for quantile estimation problem in Table \ref{tb2_qe}, we observe that the upper bound and confidence interval width decrease when the shape constraint becomes stronger. Again, with Gamma data, the upper bound decreases from $7.07$ to $4.62$ (the true value is 3.32) and the confidence interval width decreases from $0.134$ to $0.092$.

Indeed, by assuming shape property, we restrict the feasible distribution to a smaller set compared to that without shape assumptions ($D=0$). Moreover, convexity ($D=2$) implies monotonicity ($D=1$) so the former set is a subset of the latter. Our results also show that the extra errors in the additional estimation tasks needed in calibrating the parameters under the stronger shape conditions do not seem to outweigh the benefits of imposing the stronger constraints. In practice, one could visualize the distribution around the threshold $a$ to evaluate the plausibility of the shape assumption.

Now we compare the moment constraints for tail probability estimation. From Table \ref{tb1_tpe}, we can see that the difference between the two moment constraints is not as substantial as the one among the three shape constraints. Without any shape constraint (i.e., $D=0$), the rectangular constraint has a smaller relative ratio than the ellipsoidal one for all the three distributions. In the presence of shape constraint (i.e., $D=1,2$), the ellipsoidal constraint is less conservative than the rectangular one for Gamma and log-normal distributions, and only slightly more conservative for Pareto distribution. In fact, for the Pareto data and $D=1,2$, we cannot reject the null hypothesis that $(D,\chi^2)$ and $(D,\text{KS})$ are equally conservative using Welch's t-test with the significance level 0.05. Overall, if one decides to choose $D=0$ after observing the data, then the rectangular constraint seems a better choice for the moment constraint. Otherwise, the ellipsoidal constraint should have a better or comparable performance.

\subsection{Selection of Threshold}\label{sec:exp_threshold}

In this section, we compare different selections of the threshold. In Table \ref{tb3_tpe}, we consider the estimation of tail interval probabilities $\mathbb{P}(q_{0.99}\leq X\leq q_{0.995})$ under constraint setting $(2,\chi^2)$. For each data distribution, we test four single cutoff thresholds: $60^{th}$, $70^{th}$, $80^{th}$, $90^{th}$ sample percentiles. We also test using all four of them as multiple thresholds (see Supplement \ref{sec:calibration} for details of using multiple thresholds). In Table \ref{tb3_qe}, we consider the estimation of quantile $q_{0.99}$ also under constraint setting $(2,\chi^2)$ with different choices of cutoff thresholds. 

\begin{table}[!ht]
    \centering\scalebox{0.6}{\begin{tabular}{cc|ccc}
    \toprule
    \hline
    Data Source & Constraint Setting & Relative Ratio & Upper Bound & Coverage Probability \\\hline
\multirow{5}{*}{Gamma}&$60^{th}$ & $3.130\ (\pm0.088)$ & $1.57\times 10^{-2}\ (\pm4.39\times 10^{-4})$ & $1.000\ (\pm0.000)$\\
&$70^{th}$ & $3.053\ (\pm0.082)$ & $1.53\times 10^{-2}\ (\pm4.12\times 10^{-4})$ & $1.000\ (\pm0.000)$\\
&$80^{th}$ & $2.948\ (\pm0.077)$ & $1.47\times 10^{-2}\ (\pm3.85\times 10^{-4})$ & $0.995\ (\pm0.010)$\\
&$90^{th}$ & $2.823\ (\pm0.077)$ & $1.41\times 10^{-2}\ (\pm3.85\times 10^{-4})$ & $0.995\ (\pm0.010)$\\
&$(60^{th},70^{th},80^{th},90^{th})$ & $2.946\ (\pm0.079)$ & $1.47\times 10^{-2}\ (\pm3.93\times 10^{-4})$ & $0.995\ (\pm0.010)$\\\hline 
\multirow{5}{*}{Lognorm}&$60^{th}$ & $4.113\ (\pm0.147)$ & $2.06\times 10^{-2}\ (\pm7.37\times 10^{-4})$ & $1.000\ (\pm0.000)$\\
&$70^{th}$ & $3.984\ (\pm0.142)$ & $1.99\times 10^{-2}\ (\pm7.11\times 10^{-4})$ & $1.000\ (\pm0.000)$\\
&$80^{th}$ & $3.831\ (\pm0.140)$ & $1.92\times 10^{-2}\ (\pm6.99\times 10^{-4})$ & $1.000\ (\pm0.000)$\\
&$90^{th}$ & $3.629\ (\pm0.129)$ & $1.81\times 10^{-2}\ (\pm6.45\times 10^{-4})$ & $0.995\ (\pm0.010)$\\
&$(60^{th},70^{th},80^{th},90^{th})$ & $3.792\ (\pm0.132)$ & $1.90\times 10^{-2}\ (\pm6.61\times 10^{-4})$ & $1.000\ (\pm0.000)$\\\hline 
\multirow{5}{*}{Pareto}&$60^{th}$ & $7.093\ (\pm0.500)$ & $3.55\times 10^{-2}\ (\pm2.50\times 10^{-3})$ & $0.995\ (\pm0.010)$\\
&$70^{th}$ & $6.825\ (\pm0.443)$ & $3.41\times 10^{-2}\ (\pm2.21\times 10^{-3})$ & $0.995\ (\pm0.010)$\\
&$80^{th}$ & $6.428\ (\pm0.367)$ & $3.21\times 10^{-2}\ (\pm1.84\times 10^{-3})$ & $0.995\ (\pm0.010)$\\
&$90^{th}$ & $5.559\ (\pm0.233)$ & $2.78\times 10^{-2}\ (\pm1.17\times 10^{-3})$ & $0.995\ (\pm0.010)$\\
&$(60^{th},70^{th},80^{th},90^{th})$ & $5.821\ (\pm0.240)$ & $2.91\times 10^{-2}\ (\pm1.20\times 10^{-3})$ & $0.995\ (\pm0.010)$\\
    \hline
    \bottomrule
    \end{tabular}}\caption{Tail probability estimation under different cutoff threshold(s). The true value is 0.005.}
    \label{tb3_tpe}
\end{table}

\begin{table}[!ht]
    \centering\scalebox{0.6}{\begin{tabular}{cc|ccc}
    \toprule
    \hline
    Data Source & Constraint Setting & Relative Ratio & Upper Bound & Coverage Probability \\\hline
\multirow{5}{*}{Gamma w. true quantile point 3.32}&$60^{th}$ & $1.426\ (\pm0.027)$ & $4.73\times 10^{0}\ (\pm8.95\times 10^{-2})$ & $1.000\ (\pm0.000)$\\
&$70^{th}$ & $1.392\ (\pm0.028)$ & $4.62\times 10^{0}\ (\pm9.20\times 10^{-2})$ & $0.995\ (\pm0.010)$\\
&$80^{th}$ & $1.359\ (\pm0.029)$ & $4.51\times 10^{0}\ (\pm9.47\times 10^{-2})$ & $0.990\ (\pm0.014)$\\
&$90^{th}$ & $1.331\ (\pm0.029)$ & $4.41\times 10^{0}\ (\pm9.76\times 10^{-2})$ & $0.980\ (\pm0.019)$\\
&$(60^{th},70^{th},80^{th},90^{th})$ & $1.354\ (\pm0.031)$ & $4.49\times 10^{0}\ (\pm1.02\times 10^{-1})$ & $0.980\ (\pm0.019)$\\\hline 
\multirow{5}{*}{Lognorm w. true quantile point 10.24}&$60^{th}$ & $1.702\ (\pm0.067)$ & $1.74\times 10^{1}\ (\pm6.89\times 10^{-1})$ & $0.995\ (\pm0.010)$\\
&$70^{th}$ & $1.686\ (\pm0.068)$ & $1.73\times 10^{1}\ (\pm6.97\times 10^{-1})$ & $0.990\ (\pm0.014)$\\
&$80^{th}$ & $1.670\ (\pm0.069)$ & $1.71\times 10^{1}\ (\pm7.04\times 10^{-1})$ & $0.985\ (\pm0.017)$\\
&$90^{th}$ & $1.650\ (\pm0.068)$ & $1.69\times 10^{1}\ (\pm6.97\times 10^{-1})$ & $0.970\ (\pm0.024)$\\
&$(60^{th},70^{th},80^{th},90^{th})$ & $1.690\ (\pm0.071)$ & $1.73\times 10^{1}\ (\pm7.30\times 10^{-1})$ & $0.980\ (\pm0.019)$\\\hline 
\multirow{5}{*}{Pareto w. true quantile point 21.54}&$60^{th}$ & $3.637\ (\pm0.648)$ & $7.83\times 10^{1}\ (\pm1.40\times 10^{1})$ & $0.990\ (\pm0.014)$\\
&$70^{th}$ & $3.587\ (\pm0.607)$ & $7.73\times 10^{1}\ (\pm1.31\times 10^{1})$ & $0.985\ (\pm0.017)$\\
&$80^{th}$ & $3.542\ (\pm0.588)$ & $7.63\times 10^{1}\ (\pm1.27\times 10^{1})$ & $0.985\ (\pm0.017)$\\
&$90^{th}$ & $3.347\ (\pm0.450)$ & $7.21\times 10^{1}\ (\pm9.69\times 10^{0})$ & $0.985\ (\pm0.017)$\\
&$(60^{th},70^{th},80^{th},90^{th})$ & $3.584\ (\pm0.552)$ & $7.72\times 10^{1}\ (\pm1.19\times 10^{1})$ & $0.985\ (\pm0.017)$\\
    \hline
    \bottomrule
    \end{tabular}}\caption{Quantile estimation under different cutoff threshold(s). }
    \label{tb3_qe}
\end{table}

We observe that as the cutoff threshold increases, the result tends to be less conservative for both tail probability estimation and quantile estimation. For example, for Gamma distribution, as the threshold increases from $60^{th}$ sample percentile to $90^{th}$ sample percentile, the relative ratio decreases from 3.130 to 2.823 for tail probability estimation, and from 1.426 to 1.331 for quantile estimation. This phenomenon is reasonable as more information is leveraged with a larger threshold. The performance of multiple thresholds lies in the middle of the ones of single thresholds. Ideally, we should choose a relatively large threshold given that the parameters could be calibrated well. However, in practice, it is usually hard to evaluate which threshold satisfies this condition. Thus, using multiple thresholds is also a reasonable choice as it is less sensitive to the selection.

\subsection{Performance on Different Objective Functions}\label{sec:exp_objective}

In this section, we aim to understand how the performance of our approach would change with the objective function. In Table \ref{tb4_tpe_0.7}, we show the tail probability estimation results for different intervals under constraint settings $(2,\chi^2)$ and $(2,\text{KS})$. More specifically, the target probability is chosen as $\mathbb{P}(q_{\text{LHS}}\leq X\leq q_{\text{LHS}+0.005})$ where LHS takes different values ranging from 0.90 to 0.99. That is, we keep the true probability value as 0.005, and the interval $[q_{\text{LHS}}, q_{\text{LHS}+0.005}]$ moves to the farther part of the tail as LHS increases. The cutoff threshold is chosen as the $70^{th}$ sample percentile.

\begin{table}[!ht]
    \centering\scalebox{0.5}{\begin{tabular}{cc|ccc|ccc}
    \toprule
    \hline
    \multicolumn{1}{c}{} & \multicolumn{1}{c}{} &\multicolumn{3}{|c|}{{$(2,\chi^2)$}} & \multicolumn{3}{|c}{{$(2, \textup{KS})$}}\\
    Data Source & LHS Quantitle & Relative Ratio & Upper Bound & Coverage Probability & Relative Ratio & Upper Bound & Coverage Probability \\\hline
\multirow{10}{*}{Gamma}&$0.900$ & $1.714\ (\pm0.010)$ & $8.57\times 10^{-3}\ (\pm4.94\times 10^{-5})$ & $1.000\ (\pm0.000)$ & $1.766\ (\pm0.014)$ & $8.83\times 10^{-3}\ (\pm7.14\times 10^{-5})$ & $1.000\ (\pm0.000)$\\
&$0.910$ & $1.782\ (\pm0.010)$ & $8.91\times 10^{-3}\ (\pm4.87\times 10^{-5})$ & $1.000\ (\pm0.000)$ & $1.780\ (\pm0.014)$ & $8.90\times 10^{-3}\ (\pm7.12\times 10^{-5})$ & $1.000\ (\pm0.000)$\\
&$0.920$ & $1.860\ (\pm0.010)$ & $9.30\times 10^{-3}\ (\pm4.92\times 10^{-5})$ & $1.000\ (\pm0.000)$ & $1.806\ (\pm0.015)$ & $9.03\times 10^{-3}\ (\pm7.40\times 10^{-5})$ & $1.000\ (\pm0.000)$\\
&$0.930$ & $1.950\ (\pm0.010)$ & $9.75\times 10^{-3}\ (\pm5.24\times 10^{-5})$ & $1.000\ (\pm0.000)$ & $1.844\ (\pm0.016)$ & $9.22\times 10^{-3}\ (\pm7.88\times 10^{-5})$ & $1.000\ (\pm0.000)$\\
&$0.940$ & $2.051\ (\pm0.012)$ & $1.03\times 10^{-2}\ (\pm6.14\times 10^{-5})$ & $1.000\ (\pm0.000)$ & $1.899\ (\pm0.017)$ & $9.50\times 10^{-3}\ (\pm8.55\times 10^{-5})$ & $1.000\ (\pm0.000)$\\
&$0.950$ & $2.163\ (\pm0.016)$ & $1.08\times 10^{-2}\ (\pm8.09\times 10^{-5})$ & $1.000\ (\pm0.000)$ & $1.982\ (\pm0.019)$ & $9.91\times 10^{-3}\ (\pm9.35\times 10^{-5})$ & $1.000\ (\pm0.000)$\\
&$0.960$ & $2.278\ (\pm0.023)$ & $1.14\times 10^{-2}\ (\pm1.16\times 10^{-4})$ & $1.000\ (\pm0.000)$ & $2.111\ (\pm0.021)$ & $1.06\times 10^{-2}\ (\pm1.03\times 10^{-4})$ & $1.000\ (\pm0.000)$\\
&$0.970$ & $2.390\ (\pm0.033)$ & $1.20\times 10^{-2}\ (\pm1.63\times 10^{-4})$ & $1.000\ (\pm0.000)$ & $2.326\ (\pm0.023)$ & $1.16\times 10^{-2}\ (\pm1.15\times 10^{-4})$ & $1.000\ (\pm0.000)$\\
&$0.980$ & $2.565\ (\pm0.041)$ & $1.28\times 10^{-2}\ (\pm2.04\times 10^{-4})$ & $1.000\ (\pm0.000)$ & $2.755\ (\pm0.028)$ & $1.38\times 10^{-2}\ (\pm1.39\times 10^{-4})$ & $1.000\ (\pm0.000)$\\
&$0.990$ & $3.053\ (\pm0.082)$ & $1.53\times 10^{-2}\ (\pm4.12\times 10^{-4})$ & $1.000\ (\pm0.000)$ & $4.115\ (\pm0.041)$ & $2.06\times 10^{-2}\ (\pm2.04\times 10^{-4})$ & $1.000\ (\pm0.000)$\\\hline 
\multirow{10}{*}{Lognorm}&$0.900$ & $1.860\ (\pm0.011)$ & $9.30\times 10^{-3}\ (\pm5.52\times 10^{-5})$ & $1.000\ (\pm0.000)$ & $1.951\ (\pm0.016)$ & $9.75\times 10^{-3}\ (\pm7.75\times 10^{-5})$ & $1.000\ (\pm0.000)$\\
&$0.910$ & $1.959\ (\pm0.011)$ & $9.80\times 10^{-3}\ (\pm5.45\times 10^{-5})$ & $1.000\ (\pm0.000)$ & $1.987\ (\pm0.016)$ & $9.93\times 10^{-3}\ (\pm7.91\times 10^{-5})$ & $1.000\ (\pm0.000)$\\
&$0.920$ & $2.080\ (\pm0.011)$ & $1.04\times 10^{-2}\ (\pm5.49\times 10^{-5})$ & $1.000\ (\pm0.000)$ & $2.038\ (\pm0.017)$ & $1.02\times 10^{-2}\ (\pm8.39\times 10^{-5})$ & $1.000\ (\pm0.000)$\\
&$0.930$ & $2.225\ (\pm0.011)$ & $1.11\times 10^{-2}\ (\pm5.71\times 10^{-5})$ & $1.000\ (\pm0.000)$ & $2.102\ (\pm0.018)$ & $1.05\times 10^{-2}\ (\pm9.04\times 10^{-5})$ & $1.000\ (\pm0.000)$\\
&$0.940$ & $2.402\ (\pm0.013)$ & $1.20\times 10^{-2}\ (\pm6.50\times 10^{-5})$ & $1.000\ (\pm0.000)$ & $2.185\ (\pm0.020)$ & $1.09\times 10^{-2}\ (\pm9.93\times 10^{-5})$ & $1.000\ (\pm0.000)$\\
&$0.950$ & $2.613\ (\pm0.018)$ & $1.31\times 10^{-2}\ (\pm8.86\times 10^{-5})$ & $1.000\ (\pm0.000)$ & $2.298\ (\pm0.022)$ & $1.15\times 10^{-2}\ (\pm1.10\times 10^{-4})$ & $1.000\ (\pm0.000)$\\
&$0.960$ & $2.854\ (\pm0.030)$ & $1.43\times 10^{-2}\ (\pm1.48\times 10^{-4})$ & $1.000\ (\pm0.000)$ & $2.471\ (\pm0.025)$ & $1.24\times 10^{-2}\ (\pm1.23\times 10^{-4})$ & $1.000\ (\pm0.000)$\\
&$0.970$ & $3.109\ (\pm0.051)$ & $1.55\times 10^{-2}\ (\pm2.57\times 10^{-4})$ & $1.000\ (\pm0.000)$ & $2.748\ (\pm0.029)$ & $1.37\times 10^{-2}\ (\pm1.43\times 10^{-4})$ & $1.000\ (\pm0.000)$\\
&$0.980$ & $3.399\ (\pm0.078)$ & $1.70\times 10^{-2}\ (\pm3.89\times 10^{-4})$ & $1.000\ (\pm0.000)$ & $3.305\ (\pm0.036)$ & $1.65\times 10^{-2}\ (\pm1.80\times 10^{-4})$ & $1.000\ (\pm0.000)$\\
&$0.990$ & $3.984\ (\pm0.142)$ & $1.99\times 10^{-2}\ (\pm7.11\times 10^{-4})$ & $1.000\ (\pm0.000)$ & $5.091\ (\pm0.054)$ & $2.55\times 10^{-2}\ (\pm2.72\times 10^{-4})$ & $1.000\ (\pm0.000)$\\\hline 
\multirow{10}{*}{Pareto}&$0.900$ & $2.204\ (\pm0.014)$ & $1.10\times 10^{-2}\ (\pm6.80\times 10^{-5})$ & $1.000\ (\pm0.000)$ & $2.252\ (\pm0.018)$ & $1.13\times 10^{-2}\ (\pm9.23\times 10^{-5})$ & $1.000\ (\pm0.000)$\\
&$0.910$ & $2.338\ (\pm0.013)$ & $1.17\times 10^{-2}\ (\pm6.42\times 10^{-5})$ & $1.000\ (\pm0.000)$ & $2.306\ (\pm0.019)$ & $1.15\times 10^{-2}\ (\pm9.70\times 10^{-5})$ & $1.000\ (\pm0.000)$\\
&$0.920$ & $2.511\ (\pm0.013)$ & $1.26\times 10^{-2}\ (\pm6.28\times 10^{-5})$ & $1.000\ (\pm0.000)$ & $2.370\ (\pm0.020)$ & $1.19\times 10^{-2}\ (\pm1.02\times 10^{-4})$ & $1.000\ (\pm0.000)$\\
&$0.930$ & $2.739\ (\pm0.013)$ & $1.37\times 10^{-2}\ (\pm6.39\times 10^{-5})$ & $1.000\ (\pm0.000)$ & $2.456\ (\pm0.022)$ & $1.23\times 10^{-2}\ (\pm1.10\times 10^{-4})$ & $1.000\ (\pm0.000)$\\
&$0.940$ & $3.042\ (\pm0.014)$ & $1.52\times 10^{-2}\ (\pm6.83\times 10^{-5})$ & $1.000\ (\pm0.000)$ & $2.578\ (\pm0.024)$ & $1.29\times 10^{-2}\ (\pm1.19\times 10^{-4})$ & $1.000\ (\pm0.000)$\\
&$0.950$ & $3.453\ (\pm0.017)$ & $1.73\times 10^{-2}\ (\pm8.53\times 10^{-5})$ & $1.000\ (\pm0.000)$ & $2.739\ (\pm0.026)$ & $1.37\times 10^{-2}\ (\pm1.32\times 10^{-4})$ & $1.000\ (\pm0.000)$\\
&$0.960$ & $4.013\ (\pm0.032)$ & $2.01\times 10^{-2}\ (\pm1.58\times 10^{-4})$ & $1.000\ (\pm0.000)$ & $2.975\ (\pm0.030)$ & $1.49\times 10^{-2}\ (\pm1.50\times 10^{-4})$ & $1.000\ (\pm0.000)$\\
&$0.970$ & $4.759\ (\pm0.077)$ & $2.38\times 10^{-2}\ (\pm3.85\times 10^{-4})$ & $1.000\ (\pm0.000)$ & $3.375\ (\pm0.034)$ & $1.69\times 10^{-2}\ (\pm1.71\times 10^{-4})$ & $1.000\ (\pm0.000)$\\
&$0.980$ & $5.691\ (\pm0.186)$ & $2.85\times 10^{-2}\ (\pm9.28\times 10^{-4})$ & $1.000\ (\pm0.000)$ & $4.178\ (\pm0.045)$ & $2.09\times 10^{-2}\ (\pm2.23\times 10^{-4})$ & $1.000\ (\pm0.000)$\\
&$0.990$ & $6.825\ (\pm0.443)$ & $3.41\times 10^{-2}\ (\pm2.21\times 10^{-3})$ & $0.995\ (\pm0.010)$ & $6.808\ (\pm0.073)$ & $3.40\times 10^{-2}\ (\pm3.65\times 10^{-4})$ & $1.000\ (\pm0.000)$\\
    \hline
    \bottomrule
    \end{tabular}}\caption{Tail probability estimation under different objective functions. The target probability is $\mathbb{P}(q_{\text{LHS}}\leq X\leq q_{\text{LHS}+0.005})$.}
    \label{tb4_tpe_0.7}
\end{table}

From the table, we see that the relative ratio tends to increase as the interval is on the farther tail from the threshold, i.e., as LHS increases. For instance, for Gamma data under $(2,\chi^2)$ constraints, the relative ratio increases from 1.714 to 3.053 as the left endpoint LHS increases from 0.90 to 0.99. The same trend is observed for all the data distributions and constraint settings. Thus, when we infer the tail region from non-tail data, it is more conservative if the target quantity is associated with farther tail. In this case, one could try to increase the threshold if possible as discussed in Section \ref{sec:exp_threshold}, which could reduce the conservativeness to certain degree. Otherwise, one could at least get a conservative but safe estimation with this approach.

\subsection{Comparison with POT}\label{sec:exp_POT}

Finally, we compare our approach with the conventional POT method, where a GPD is fitted from the excess-loss data using maximum likelihood (e.g., \citet{smithEstimatingTailsProbability1987}) and a $95\%$ confidence upper bound for the tail interval probability is then obtained from the delta method. Table \ref{tb5_tpe} shows the POT results, where we choose the threshold for fitting the GPD according to the graphical approach based on the linearity of the mean excess function (see \cite{embrechtsModellingExtremalEvents1997}). 

Comparing Tables \ref{tb4_tpe_0.7} and \ref{tb5_tpe}, we see that while our DRO method obtains looser bounds than POT, it exhibits correct coverage. By contrast, POT undercovers (bold in Table \ref{tb5_tpe}) in many cases. For instance, for an objective function $\mathbb{P}(q_{0.90}\leq X\leq q_{0.905})$ and Gamma dataset, POT gives an upper confidence bound $4.64\times 10^{-3}$ which is even smaller than the truth, while DRO gives $8.57\times 10^{-2}$ with configuration $(2,\chi^2)$. On the other hand, POT gives only $65\%$ coverage while DRO gives $100\%$ coverage. The subpar coverage of POT suggests that the data size is too small to carry out proper estimation. 

Overall, POT gives estimates closer to the true target quantity but its confidence bounds can fall short of the prescribed coverage. Our recommendation is that a modeler whose priority is about the order of magnitude would be better off choosing GPD, whereas a more risk-averse modeler seeking a bound with correct confidence guarantee would be better off choosing our DRO approach.

\begin{table}[!ht]
    \centering\scalebox{0.6}{\begin{tabular}{cc|ccc}
    \toprule
    \hline
    Data Source & LHS Quantitle & Relative Ratio & Upper Bound & Coverage Probability \\\hline
\multirow{10}{*}{Gamma}&$0.900$ & $0.928\ (\pm0.060)$ & $4.64\times 10^{-3}\ (\pm3.01\times 10^{-4})$ & $\bm{0.650\ (\pm0.066)}$\\
&$0.910$ & $1.096\ (\pm0.046)$ & $5.48\times 10^{-3}\ (\pm2.32\times 10^{-4})$ & $\bm{0.780\ (\pm0.057)}$\\
&$0.920$ & $1.159\ (\pm0.040)$ & $5.79\times 10^{-3}\ (\pm2.01\times 10^{-4})$ & $\bm{0.850\ (\pm0.049)}$\\
&$0.930$ & $1.177\ (\pm0.035)$ & $5.88\times 10^{-3}\ (\pm1.75\times 10^{-4})$ & $\bm{0.880\ (\pm0.045)}$\\
&$0.940$ & $1.170\ (\pm0.032)$ & $5.85\times 10^{-3}\ (\pm1.59\times 10^{-4})$ & $\bm{0.905\ (\pm0.041)}$\\
&$0.950$ & $1.164\ (\pm0.030)$ & $5.82\times 10^{-3}\ (\pm1.49\times 10^{-4})$ & $\bm{0.895\ (\pm0.042)}$\\
&$0.960$ & $1.169\ (\pm0.030)$ & $5.85\times 10^{-3}\ (\pm1.50\times 10^{-4})$ & $\bm{0.865\ (\pm0.047)}$\\
&$0.970$ & $1.210\ (\pm0.032)$ & $6.05\times 10^{-3}\ (\pm1.58\times 10^{-4})$ & $\bm{0.900\ (\pm0.042)}$\\
&$0.980$ & $1.346\ (\pm0.037)$ & $6.73\times 10^{-3}\ (\pm1.86\times 10^{-4})$ & $0.950\ (\pm0.030)$\\
&$0.990$ & $1.705\ (\pm0.056)$ & $8.52\times 10^{-3}\ (\pm2.81\times 10^{-4})$ & $0.970\ (\pm0.024)$\\\hline 
\multirow{10}{*}{Lognorm}&$0.900$ & $1.002\ (\pm0.051)$ & $5.01\times 10^{-3}\ (\pm2.55\times 10^{-4})$ & $\bm{0.730\ (\pm0.062)}$\\
&$0.910$ & $1.110\ (\pm0.037)$ & $5.55\times 10^{-3}\ (\pm1.87\times 10^{-4})$ & $\bm{0.830\ (\pm0.052)}$\\
&$0.920$ & $1.185\ (\pm0.023)$ & $5.93\times 10^{-3}\ (\pm1.16\times 10^{-4})$ & $0.910\ (\pm0.040)$\\
&$0.930$ & $1.200\ (\pm0.020)$ & $6.00\times 10^{-3}\ (\pm1.01\times 10^{-4})$ & $0.935\ (\pm0.034)$\\
&$0.940$ & $1.219\ (\pm0.019)$ & $6.09\times 10^{-3}\ (\pm9.45\times 10^{-5})$ & $0.965\ (\pm0.025)$\\
&$0.950$ & $1.241\ (\pm0.020)$ & $6.20\times 10^{-3}\ (\pm9.86\times 10^{-5})$ & $0.955\ (\pm0.029)$\\
&$0.960$ & $1.272\ (\pm0.022)$ & $6.36\times 10^{-3}\ (\pm1.10\times 10^{-4})$ & $0.950\ (\pm0.030)$\\
&$0.970$ & $1.331\ (\pm0.024)$ & $6.65\times 10^{-3}\ (\pm1.19\times 10^{-4})$ & $0.975\ (\pm0.022)$\\
&$0.980$ & $1.461\ (\pm0.029)$ & $7.31\times 10^{-3}\ (\pm1.43\times 10^{-4})$ & $0.990\ (\pm0.014)$\\
&$0.990$ & $1.731\ (\pm0.045)$ & $8.65\times 10^{-3}\ (\pm2.23\times 10^{-4})$ & $0.985\ (\pm0.017)$\\\hline 
\multirow{10}{*}{Pareto}&$0.900$ & $1.061\ (\pm0.029)$ & $5.30\times 10^{-3}\ (\pm1.44\times 10^{-4})$ & $\bm{0.830\ (\pm0.052)}$\\
&$0.910$ & $1.100\ (\pm0.020)$ & $5.50\times 10^{-3}\ (\pm1.00\times 10^{-4})$ & $\bm{0.875\ (\pm0.046)}$\\
&$0.920$ & $1.131\ (\pm0.015)$ & $5.66\times 10^{-3}\ (\pm7.45\times 10^{-5})$ & $0.910\ (\pm0.040)$\\
&$0.930$ & $1.156\ (\pm0.015)$ & $5.78\times 10^{-3}\ (\pm7.37\times 10^{-5})$ & $0.940\ (\pm0.033)$\\
&$0.940$ & $1.183\ (\pm0.016)$ & $5.92\times 10^{-3}\ (\pm7.81\times 10^{-5})$ & $0.950\ (\pm0.030)$\\
&$0.950$ & $1.216\ (\pm0.017)$ & $6.08\times 10^{-3}\ (\pm8.45\times 10^{-5})$ & $0.980\ (\pm0.019)$\\
&$0.960$ & $1.260\ (\pm0.018)$ & $6.30\times 10^{-3}\ (\pm9.19\times 10^{-5})$ & $0.985\ (\pm0.017)$\\
&$0.970$ & $1.323\ (\pm0.021)$ & $6.62\times 10^{-3}\ (\pm1.04\times 10^{-4})$ & $0.990\ (\pm0.014)$\\
&$0.980$ & $1.428\ (\pm0.026)$ & $7.14\times 10^{-3}\ (\pm1.30\times 10^{-4})$ & $0.990\ (\pm0.014)$\\
&$0.990$ & $1.627\ (\pm0.046)$ & $8.13\times 10^{-3}\ (\pm2.28\times 10^{-4})$ & $0.980\ (\pm0.019)$\\
    \hline
    \bottomrule
    \end{tabular}}\caption{Tail probability estimation under different objective functions with the POT method. The target probability is $\mathbb{P}(q_{\text{LHS}}\leq X\leq q_{\text{LHS}+0.005})$. The coverage probabilities under the nominal confidence level are bold in the table.}
    \label{tb5_tpe}
\end{table}

\bibliographystyle{apalike}
\bibliography{RO}
\spacingset{1.9} 

\bigskip
\begin{center}
{\large\bf SUPPLEMENTARY MATERIAL}
\end{center}

\begin{description}

\item[Supplement A] Calibration Procedures.
\item[Supplement B] Sensitivity Analysis. 
\item[Supplement C] Proofs of Results. 
\end{description}
\newpage
\appendix
\setcounter{equation}{0}
\setcounter{page}{1}
\renewcommand{\theequation}{A.\arabic{equation}}
\renewcommand{\thelemma}{A.\arabic{lemma}}
\renewcommand{\thetheorem}{A.\arabic{theorem}}

\section{Supplement: Calibration Procedures}\label{sec:calibration}

The constructions of \eqref{ellipsoid} and \eqref{rectangle} are motivated by the two following weak convergence results.

\text{1. $\chi^2$ distribution.} Given $n$ sample points from $d$-dimensional random vectors $\mathbf{X}$ with positive variance-covariance matrix $\bm{\Sigma}$, according to weak convergence results, e.g., Corollary 2.1 in \citet{j.j.dikDISTRIBUTIONGENERALQUADRATIC1985}, we have
\begin{align}
n(\mathbb{E}[\bm{X}]-\hat{\bm{\mu}})^\top\hat{\bm{\Sigma}}_n^{-1}(\mathbb{E}[\bm{X}]-\hat{\bm{\mu}})
{\Rightarrow}\chi^2_d\label{chi2convergence}
\end{align}
where $\hat{\bm{\mu}}$ and $\hat{\bm{\Sigma}}_{n}$ are the sample mean and sample covariance matrix respectively and $\chi^2_{d}$ is a chi-squared distribution with $d$ degrees of freedom.

\text{2. Kolmogorov distribution.} Given $n$ i.i.d. realization ordered points $\{x_{1},\cdots,x_{n}\}$ from random variable $X$ with continuous cumulative probability distribution $F$ and the corresponding empirical distribution $F_n$, the Kolmogorov-Smirnov statistic $D_n\big(:=\sup_x|F_n(x)-F(x)|\big)$ converges to Kolmogorov distribution, e.g., \citet{noetherNoteKolmogorovStatistic1963}, i.e., 
\begin{align}
\sqrt{n}\max_{i=1,\cdots,n} \big(\max\big(\frac{i-n}{n}+\mathbb{E}[\mathbb{I}(X\geq x_i)],\frac{n+1-i}{n}-\mathbb{E}[\mathbb{I}(X\geq x_i)]\big)\big) \Rightarrow \sup_{t\in [0,1]}|B(t)-tB(1)|\label{kolconvergence}
\end{align}
where $B(t)$ is a standard Wiener process.

We discuss the calibration methods based on empirical observations to achieve the statistical guarantee results in Theorem \ref{proposition:sg1} and \ref{proposition:sg2}. To construct feasible regions that satisfy coverage of true distribution $P_{true}$ with high probability, one could calibrate the parameters $\bm{\Gamma}, \eta,\nu$ in a statistical perspective. Overall, $\mathbb{S}$ can be constructed based on the aforementioned weak convergence results. The estimation of $\eta$ and $\nu$ can be accessible via kernel density estimation and boostrapping. A Bonferroni correction is applied to guarantee simultaneous confidence level of the estimation of those parameters. We illustrate the procedure in both ellipsoidal and rectangular cases. Consider a sample realization $\{x_1,x_2,\cdots,x_n\}$ from random variable $X$ with cumulative distribution $F$. 

\paragraph{Ellipsoidal Constraint}Denote $\hat{\bm{\mu}}=\frac{1}{n}\sum_{i=1}^n \bm{g}(x_i)$, $\hat{\bm{\Sigma}}=\frac{1}{n-1}\sum_{i=1}^n\big(\bm{g}(x_i)-\hat{\bm{\mu}}\big)\big(\bm{g}(x_i)-\hat{\bm{\mu}}\big)^\top$ and assume the dimension of $\bm{g}$ is $d$. To calibrate the ellipsoidal region $\mathbb{S}_{E}$, we utilize the $\chi^2$ weak convergence result depicted in \eqref{chi2convergence}. In particular, we need to determine $z$ in 
\begin{align*}
n(\mathbb{E}[\bm{g}(X)]-\hat{\bm{\mu}})^\top\hat{\bm{\Sigma}}_n^{-1}(\mathbb{E}[\bm{g}(X)]-\hat{\bm{\mu}})\leq z, 
\end{align*}
to construct the required confidence region. 

To that end, value $z$ is chosen as the $\big(1-\frac{\alpha}{2}\big)^{\text{th}}$ or $\big(1-\frac{\alpha}{3}\big)^{\text{th}}$ quantile of chi-squared distribution with $d$ degrees of freedom under $\mathscr{P}^{1}_{a,\eta}$ or $\mathscr{P}^{2}_{a,\underline{\eta},\bar{\eta},\nu}$ respectively. For the former case, $\eta$ is estimated as the $1-\frac{\alpha}{2}$ percentile of the bootstrapped densities at $a$. For the latter case, $\underline{\eta}$ and $\bar{\eta}$ are chosen as the $\frac{\alpha}{6}$ percentile and $1-\frac{\alpha}{6}$ percentile of the bootstrapped densities at $a$ respectively and $-\hat{\nu}$ the $\frac{\alpha}{3}$ percentile of the bootstrapped coefficients $F_+^{(2)}(a)$. 

In the case of multiple threshold levels, we seek to obtain a result which is the optimal of all objective values ranging over different $a_i,i=1,...,m$ and cover the true value with probability $1-\alpha$. Under this scenario, value $z$ is chosen as $\big(1-\frac{\alpha}{m+1}\big)^{\text{th}}$ or $\big(1-\frac{\alpha}{2m+1}\big)^{\text{th}}$ quantile of chi-squared distribution with $d$ degrees of freedom under $\mathscr{P}^{1}_{a_i,\eta}$ or $\mathscr{P}^{2}_{a_i,\underline{\eta},\bar{\eta},\nu}$ respectively. Alternatively, we  may also choose $z$ via bootstrapping, depicted in Algorithm \ref{algo:z_ellipsoidal}. For each $a_i,i=1,...,m$, $\hat{\bm{\mu}},\hat{\bm{\Sigma}}$ are the empirical mean and covariance. 
$\eta$ is chosen as the $1-\frac{\alpha}{m+1}$ percentile of the bootstrapped densities at $a_i$ for $\mathscr{P}^{1}_{a_i,\eta}$. For $\mathscr{P}^{2}_{a_i,\underline{\eta},\bar{\eta},\nu}$, $\underline{\eta}$ and $\bar{\eta}$ are chosen as the $\frac{\alpha}{4m+2}$ percentile and $1-\frac{\alpha}{4m+2}$ percentile of the bootstrapped densities at $a_i$ respectively and $-\hat{\nu}$ the $\frac{\alpha}{2m+1}$ percentile of the bootstrapped coefficients $F_+^{(2)}(a_i)$.

\scalebox{0.8}{
\begin{minipage}{1.2\linewidth}
\begin{algorithm}[H]
  \caption{Bootstrap procedure for computing $z$.}
  \begin{algorithmic}[1]
\For{$b=1,\dots,500$}
 \State Randomly sample $n$ data points with replacement from $\{x_1,\cdots,x_n\}$.
 \State Compute $\hat{z}$ after each sampling:
\begin{align*}
\hat{z} = \max\limits_{a_i,i=1,...,m}\Big\{n\big(\hat{\hat{\bm{\mu}}}-\hat{\bm{\mu}}\Big)^T\hat{\bm{\Sigma}}^{-1}\Big(\hat{\hat{\bm{\mu}}}-\hat{\bm{\mu}}\big)
\Big\}
\end{align*}
where $\hat{\hat{\bm{\mu}}}$ is the empirical mean of $\{\bm{g}(x_i^b)\}_{i=1}^n$ and $\{x^b_1,\cdots,x^b_n\}$ is the data set we obtain after this resampling. 
\EndFor
\State Select the $(1-\frac{\alpha}{m+1})^{\text{th}}$ or $(1-\frac{\alpha}{2m+1})^{\text{th}}$ quantile of the $\hat{z}$'s obtained above under $\mathscr{P}^{1}_{a_i,\eta}$ or $\mathscr{P}^{2}_{a_i,\underline{\eta},\bar{\eta},\nu}$ respectively.
  \end{algorithmic}
  \label{algo:z_ellipsoidal}
\end{algorithm}
\end{minipage}}

\paragraph{Rectangular Constraint}To calibrate the rectangular region $\mathbb{S}_{R}$, we use the Kolmogorov-Smirnov weak convergence result shown in \eqref{kolconvergence}. In particular, we need to determine $z$ in 
\begin{align}
\sqrt{n}\max_{i=1,\cdots,n} \|F(x_i)-\widehat{F}(x_i)\|_\infty\leq z,
\end{align} 
where $\widehat{F}$ is the empirical distribution, to construct the confidence region. To this end, value $z$ is the $\big(1-\frac{\alpha}{2}\big)^{\text{th}}$ or $\big(1-\frac{\alpha}{3}\big)^{\text{th}}$ quantile of Kolmogorov distribution under $\mathscr{P}^{1}_{a,\eta}$ or $\mathscr{P}^{2}_{a,\underline{\eta},\hat{\eta},\nu}$ respectively. For the former case, $\eta$ is the $1-\frac{\alpha}{2}$ percentile of the bootstrapped densities $F^{(1)}(a)$. For the latter case, $\bar{\eta}$ and $\underline{\eta}$ are the upper bound and lower bound of $F^{(1)}(a)$ with joint probability $1-\frac{\alpha}{3}$ and $-\hat{\nu}$ is the lower bound of $F_+^{(2)}(a|x\geq a)$ with probability $1-\frac{\alpha}{3}$.

For the setting of choosing among a range of thresholds $a_i,i=1,...,m$. Value $z$ is chosen as $\big(1-\frac{\alpha}{m+1}\big)^{\text{th}}$ or $\big(1-\frac{\alpha}{2m+1}\big)^{\text{th}}$ quantile of Kolmogorov distribution divided by $\sqrt{n}$ under $\mathscr{P}^{1}_{a_i,\eta}$ or $\mathscr{P}^{2}_{a_i,\underline{\eta},\hat{\eta},\nu}$ respectively. Alternatively, we  may also choose $z$ via bootstrapping, depicted in Algorithm \ref{algo:z_ks}, where $\Delta=1-\frac{\alpha}{m+1}$ for $\mathscr{P}^{1}_{a_i,\eta}$ and $\Delta=1-\frac{\alpha}{2m+1}$ for $\mathscr{P}^{2}_{a_i,\underline{\eta},\bar{\eta},\nu}$. For the first case, $\eta$ is set as the $1-\frac{\alpha}{m+1}$ percentile of the bootstrapped densities $F^{(1)}(a_i)$. For the second case, $\bar{\eta}$ and $\underline{\eta}$ are the upper bound and lower bound of $F^{(1)}_+(a_i)$ with joint probability $1-\frac{\alpha}{2m+1}$ and $-\hat{\nu}$ is the lower bound of $F_+^{(2)}(a_i|x\geq a_i)$ with probability $1-\frac{\alpha}{2m+1}$.

\scalebox{0.8}{
\begin{minipage}{1.2\linewidth}
\begin{algorithm}[H]
  \caption{Bootstrap procedure for computing $z$.}
  \begin{algorithmic}[1]
\For{$b=1,\dots,500$}
 \State Randomly sample $n$ data points with replacement from $\{x_1,\cdots,x_n\}$.
 \State Compute $\hat{z}$ after each sampling:
{\small
\begin{align*}
\hat{z} = \max\limits_{a_i,i=1,...,m}\bigg\{\sqrt{n}\max\limits_{j=1,2,\cdots,n} \Big\{\max\Big\{\frac{j}{n}-\frac{\#\{k:a_i\leq x_k\leq x^b_j\}}{n},\frac{\#\{k:a_i\leq x_k\leq x^b_j\}}{n}-\frac{j-1}{n}\Big\}\Big\}\bigg\}
\end{align*}}
where $\{x^b_1,\cdots,x^b_n\}$ is the non-decreasing ordered dataset obtained after each sampling.
\EndFor
\State Select the $(\Delta)^{\text{th}}$ quantile of the $\hat{z}$'s obtained above.
\State $z$ is output as the selected $(\Delta)^{\text{th}}$ quantile in Step 5 divided by $\sqrt{n}$.
  \end{algorithmic}
  \label{algo:z_ks}
\end{algorithm}
\end{minipage}}

\section{Sensitivity Analysis}\label{sec:sensitivity}
On a high level, we give a perturbation analysis on the optimal value of the parameterized problem $\mathfrak{P}_{\bm{\mu}}$ in (\ref{compact_origin_opt}) with parameters $\bm{\mu}$. In particular, for a given perturbation direction $\bm{dr}$, we give a formula to express the perturbed optimal value of the problem as a linear function of $\bm{dr}$. This analysis shows how the DRO estimation changes with respect to the calibration accuracy, and hence helps us understand how robust the estimation is to the randomness in the data and the calibration procedure.

In program $\mathfrak{P}_{\bm{\mu}}$ with given functions $H(x)$ and $G_j(x), \forall j \in \{0\}\cup [d]$ for some positive integer $d$ where $[d]:=\{1,\ldots,d\}$, we assume \text{(1)}. $G_0(x)=1,\mu_0=1$, restricting the non-negative bounded measures to probability measures; \text{(2)}. the parameters $\mu_j\in \mathbb{R},\forall j\in [d]$. Note that the dual of $\mathfrak{P}_{\bm{\mu}}$, i.e., $\mathfrak{D}_{\bm{\mu}}$, is a linear semi-infinite programming, in which there are finite number of decision variables and infinite number of linear constraints. According to \citet{gobernaFarkasMinkowskiSystemsSemiinfinite1981}, the relation $\bm{a}^\top \bm{x}\geq \beta$ with the associated vector $(\bm{a}^\top,\beta)^\top$ is called a linear consequence relation of the constraints system in some program $\mathfrak{P}$ if every feasible point in $\mathbb{F}(\mathfrak{P})$ satisfies the relation. Program $\mathfrak{P}$ is then \textit{Farkas-Minkowski} (FM) if every linear consequence relation of the constraints system of $\mathfrak{P}$ is a linear consequence relation of a finite subsystem. Lastly, for any program $\mathfrak{P}$, $\mathbb{F}(\mathfrak{P})$, $\mathbb{F}^*(\mathfrak{P})$ and $v(\mathfrak{P})$ denote the feasible region, the optimal solution region and the optimal objective value of $\mathfrak{P}$ respectively.{\small
\begin{multicols}{2}
\noindent 
\begin{equation}
\label{compact_origin_opt}
\begin{aligned}\break
(\mathfrak{P}_{\bm{\mu}}):\max_{X\sim P\in \mathscr{M}^+(\mathbb{\Omega})} \ \  &\mathbb{E}_P[H(X)] \ \\
\text{ s.t. } \ \ &\mathbb{E}_P [\bm{G}(X)]=\bm{\mu}.
\end{aligned}
\end{equation}
\noindent
\begin{align*}
(\mathfrak{D}_{\bm{\mu}}):\min_{\bm{y}} \ \  &\bm{y}^\top\bm{\mu}\\
\text{ s.t. } \ \ &\bm{y}^\top \bm{G}(x)\geq H(x),\forall x\in \mathbb{\Omega}.
\end{align*}
\end{multicols}
\vspace{-0.5cm}
}
\noindent where $\mathscr{M}^+(\mathbb{\Omega})$ denotes the space of non-negative bounded measures on $\mathbb{\Omega}$. 

\begin{theorem}\label{thm:first_order_compactness}
Suppose $\mathfrak{P}_{\bm{\mu}}$ is feasible and $\mathfrak{D}_{\bm{\mu}}$ is a feasible FM system with $\bm{\mu}\in \textup{int}\big(\mathbb{M}(\mathfrak{D}_{\bm{\mu}})\big)$, then for any direction $\bm{dr}$ satisfying that the constraint system $\big\{\mathbb{E}_P [G_j(X)]+\mu_j\pi= dr_j,\forall j\in \{0\}\cup [d]\big\}$ for $(P,\pi)$ is non-empty, there exists $\epsilon>0$ such that $\forall \rho: 0\leq \rho<\epsilon$,
\begin{align}
v(\mathfrak{P}_{\bm{\mu}+\rho \bm{dr}})=v(\mathfrak{P}_{\bm{\mu}})+\rho\min\{\bm{dr}^\top \bm{y}|\bm{y}\in \mathbb{F}^*(\mathfrak{D}_{\bm{\mu}})\}.\label{eq:1st_expansion}
\vspace{-1cm}
\end{align}
\end{theorem}
The derivation mainly follows standard duality and sensitivity analysis for linear semi-infinite linear optimization \citep{miguela.gobernaLinearSemiinfiniteOptimization2000} seen in Theorem 2 of \citet{gobernaSensitivityAnalysisLinear2007}. There are several sufficient conditions for $\mathfrak{D}_{\bm{\mu}}$ to be an FM system. For instance, if $\mathbb{\Omega}$ is a compact set, function $H$ and vector functions $\bm{G}$ are continuous such that $\exists \bm{y},\ \bm{y}^\top \bm{G}(x)>H(x),\forall x\in\Omega$, then $\mathfrak{D}_{\bm{\mu}}$ is an FM system according to Lemma \ref{lemma:cc2fm}. Generally speaking, such a first-order expansion requires solving an auxiliary optimization problem to obtain the difference $v(\mathfrak{P}_{\bm{\mu}+\rho \bm{dr}})-v(\mathfrak{P}_{\bm{\mu}})$. However, if $\bm{dr}$ is a multiple of $\bm{\mu}$, i.e., $\rho\bm{dr}=c\bm{\mu}$ for some $c\neq 0$, then the constraint system $\big\{\mathbb{E}_P [G_j(X)]+\mu_j\pi= dr_j,\forall j\in \{0\}\cup [d]\big\}$ is non-empty for $(\mathcal{P,\pi})$ and therefore \eqref{eq:1st_expansion} is simplified as  
\begin{align*}
v(\mathfrak{P}_{\bm{\mu}+ \bm{dr}})=v(\mathfrak{P}_{(1+c)\bm{\mu}})&= v(\mathfrak{P}_{\bm{\mu}})+c\min\{\bm{\mu}^\top \bm{y}|\bm{y}\in \mathbb{F}^*(\mathfrak{D}_{\bm{\mu}})\}= (1+c)v(\mathfrak{P}_{\bm{\mu}}).
\end{align*}

\section{Supplement: Proofs of Results}
\subsection{Proof of Theorem \ref{proposition:sg1}}
\begin{proof}
For program $\mathfrak{P}(\mathscr{P}^{1}_{a,\eta},\bm{g},\mathbb{S})$, if $P_{true}\in\mathbb{F}(\mathfrak{P}(\mathscr{P}^{1}_{a,\eta},\bm{g},\mathbb{S}))$, we will have $\psi(P_{true})\leq v(\mathfrak{P}(\mathscr{P}^{1}_{a,\eta},\bm{g},\mathbb{S}))$ where $v(\mathfrak{P})$ is the optimal value of program $\mathfrak{P}$. Therefore, 
\begin{align*}
\mathbb{P}(v(\mathfrak{P}(\mathscr{P}^{1}_{a,\eta},\bm{g},\mathbb{S})) \geq \psi(P_{true}))\geq \mathbb{P}(P_{true}\in\mathbb{F}(\mathfrak{P}(\mathscr{P}^{1}_{a,\eta},\bm{g},\mathbb{S})))=1-\alpha.
\end{align*}
Similar arguments hold for $\mathfrak{P}(\mathscr{P}^{2}_{a,\underline{\eta},\bar{\eta},\nu},\bm{g},\mathbb{S})$. 
\end{proof}
\subsection{Proof of Theorem \ref{proposition:sg2}}
\begin{proof}

If $P_{true}\in\bigcap_{i=1,\ldots,m}\mathbb F(\mathfrak{P}(\mathscr{P}_i,\bm{g}_i,\mathbb{S}_i))$, we have 
\begin{align*}
\psi(P_{true})\leq v\big(\mathfrak{P}(\mathscr{P}_i,\bm{g}_i,\mathbb{S}_i)\big),i=1,...,m. 
\end{align*} 
Hence 
\begin{align*}
&\mathbb{P}\Big(\min\limits_{a_i,i=1,...,m} v\big(\mathfrak{P}(\mathscr{P}_i,\bm{g}_i,\mathbb{S}_i)\big) \geq \psi(P_{true})\Big)&\\
\geq &\mathbb{P}\Big(P_{true}\in\bigcap_{i=1,\ldots,m}\mathbb F(\mathfrak{P}(\mathscr{P}_i,\bm{g}_i,\mathbb{S}_i))\Big)\\
=&1-\alpha.
\end{align*}
Similar arguments hold for $P_{true}\in\mathscr P^2(\min_ia_i)$. 
\end{proof}

\subsection{Proof of Proposition \ref{conserv_1}}
\begin{proof}
By definition, we know that $Y>0$ with probability $1$, and that 
\begin{equation}
    F_Y(x)=\mathbb{P}\left( Y\leq x\right)=P\left(X\leq x_F-x^{-1}\right)=F\left(x_F-x^{-1}\right).
    \label{F_Y}
\end{equation}
The fact that $F_Y\in MDA(H_{-\xi})$ is proved in \citet{embrechtsModellingExtremalEvents1997}. Now we prove the remaining statements. By taking derivatives, we get that 
\begin{eqnarray}
&f_Y(x)=f\left(x_F-x^{-1}\right)x^{-2},\label{f_Y}\\
&f_Y'(x)=f'\left(x_F-x^{-1}\right)x^{-4}-2f\left(x_F-x^{-1}\right)x^{-3}.\label{f_Y'}
\end{eqnarray}
By Assumption \ref{smoothness and shape}, $f\left(x_F-x^{-1}\right)$ and $f'\left(x_F-x^{-1}\right)$ exist for $x>z_Y:=1/(x_F-z)>0$. Thus $F_Y$ is twice differentiable on $(z_Y,\infty)$. We also know that $f\left(x_F-x^{-1}\right)>0$ for $x>z_Y$. Then by \eqref{f_Y}, we get that $f_Y>0$ on $(z_Y,\infty)$. Moreover, $f_Y$ is decreasing on $(z_Y,\infty)$ since $f\left(x_F-x^{-1}\right)$ and $x^{-2}$ are both positive and decreasing for $x>z_Y$. Now we only need to prove that $f_Y$ is also convex on $(z_Y,\infty)$. Indeed, for any $z_Y<x_1<x_2$ and $0<\lambda<1$, we have that 
{\footnotesize

\begin{align*}
&f_Y\left(\lambda x_1+(1-\lambda)x_2\right)\\
=&f\left(x_F-\frac{1}{\lambda x_1+(1-\lambda)x_2}\right)\frac{1}{\left(\lambda x_1+(1-\lambda)x_2\right)^2}\\
\leq & f\left(x_F-\frac{\lambda}{x_1}-\frac{1-\lambda}{x_2}\right)\left[\frac{\lambda}{x_1^2}+\frac{1-\lambda}{x_2^2}\right]\\
\leq & \left[\lambda f\left(x_F-x_1^{-1}\right)+(1-\lambda) f\left(x_F-x_2^{-1}\right)\right]\left[\frac{\lambda}{x_1^2}+\frac{1-\lambda}{x_2^2}\right]\\
= & \lambda^2 f\left(x_F-x_1^{-1}\right)x_1^{-2}+(1-\lambda)^2 f\left(x_F-x_2^{-1}\right)x_2^{-2}+\lambda(1-\lambda)\left[f\left(x_F-x_1^{-1}\right)x_2^{-2}+f\left(x_F-x_2^{-1}\right)x_1^{-2}\right]\\
\leq & \lambda^2 f\left(x_F-x_1^{-1}\right)x_1^{-2}+(1-\lambda)^2 f\left(x_F-x_2^{-1}\right)x_2^{-2}+\lambda(1-\lambda)\left[f\left(x_F-x_1^{-1}\right)x_1^{-2}+f\left(x_F-x_2^{-1}\right)x_2^{-2}\right]\\
= & \lambda f_Y(x_1)+(1-\lambda) f_Y(x_2).
\end{align*}
}
Therefore, $F_Y$ also satisfies Assumption \ref{smoothness and shape}. 
\end{proof}
\subsection{Proof of Proposition \ref{conserv_2}}
\begin{proof}
Clearly, \eqref{F_Y}, \eqref{f_Y} and \eqref{f_Y'} still hold. Then we may follow the proof of Proposition \ref{conserv_1} to define $z_Y$ and prove that on $(z_Y,\infty)$, $F_Y$ is twice differentiable and $f_Y$ is positive, decreasing and convex. The remainder of this proof inspires from \citet{embrechtsModellingExtremalEvents1997}. Since $F$ satisfies Assumption \ref{smoothness and shape} and \ref{additional assumption}, it is known that $F$ is a von Mises function with auxiliary function $\bar{F}/f$. More specifically, $\bar{F}$ has the following representation:
\begin{align*}
    \bar{F}(x)=c\exp\left\{-\int_z^x \frac{f(t)}{\bar{F}(t)}\mathrm{d}t\right\},z<x<x_F
\end{align*}
where $c$ is a positive constant. Then we have that 
\begin{align*}
\bar{F}_Y(x)&=\bar{F}\left(x_F-x^{-1}\right)=c\exp\left\{-\int_z^{x_F-x^{-1}} \frac{f(t)}{\bar{F}(t)}\mathrm{d}t\right\}\\
&=c\exp\left\{-\int_{1/(x_F-z)}^x \frac{f\left(x_F-s^{-1}\right)s^{-2}}{\bar{F}\left(x_F-s^{-1}\right)}\mathrm{d}s\right\}
=c\exp\left\{-\int_{z_Y}^x \frac{f_Y(s)}{\bar{F}_Y(s)}\mathrm{d}s\right\}.
\end{align*}
By definition, $F_Y$ is also a von Mises function with auxiliary function $\bar{F}_Y/f_Y$, which implies that $F_Y\in MDA(\Lambda)$ and that \eqref{additional assumption for Y} holds. 
\end{proof}
\subsection{Proof of Theorem \ref{optimal value theorem}}
\begin{proof}
Define $h(x)=\mathbb{I}(x>b)$. Correspondingly, 
\begin{align*}
H(x)=\int_0^x\int_0^t h(s+a)\mathrm{d}s\mathrm{d}t=\frac{(x-b+a)^2}{2}\mathbb{I}(x>b-a).  
\end{align*}
Then Assumption 1 and Assumption 2 in \citet{lamTailAnalysisParametric2017} hold. Indeed, $h:\mathbb{R}\rightarrow\mathbb{R}^+$ is bounded and is nondecreasing in $[a, b)$ and nonincreasing in $(b,\infty)$. Also, $\lambda=\lim_{x\rightarrow\infty}H(x)/x^2=\frac{1}{2}$. Therefore, we can apply Theorem 4 in \citet{lamTailAnalysisParametric2017}. In particular, $z^*=\max_{x\in[0,\mu]} W(x)$ where 
\begin{align*}
W(x)=\begin{cases}
\nu\left(\frac{\sigma-\mu^2}{\sigma-2\mu x+x^2}H(x)+\frac{(\mu-x)^2}{\sigma-2\mu x+x^2}H(\frac{\sigma-\mu x}{\mu-x})\right) & \text{if } x\in[0,\mu);\\
\nu(H(\mu)+\lambda(\sigma-\mu^2)) & \text{if }x=\mu.
\end{cases}
\end{align*}

If $\mu\le b-a$, then it can be proved that $\arg\max W(x)=\mu$ and thus $z^*=W(\mu)=\frac{\nu}{2}(\sigma-\mu^2)$. If $\mu>b-a$, then it can be proved that $\arg\max W(x)=[b-a,\mu]$, and thus $z^*=\frac{\nu}{2}[\sigma-2(b-a)\mu+(b-a)^2]$.
\end{proof}
\subsection{Proof of Theorem \ref{conserv_4}}
\begin{proof}
We denote the optimal value of \eqref{quantile optimization problem} as $q_{opt}$ and then our goal is to show that $q^*=q_{opt}$. First, we prove that $q_{opt}\leq q^*$. Indeed, for any feasible function $\tilde{f}$, the $p$-quantile $q$ is the value that satisfies $\int_q^{\infty}\tilde{f}(x)\mathrm{d}x=1-p$. We see that $\tilde{f}$ is also feasible for \eqref{optimization problem}, and thus $z^*(a,q)\geq 1-p$. By the definition of $q^*$, we know that $q\leq q^*$. Hence, $q_{opt}\leq q^*$. Now we justify that $q^{opt}=q^*$. 

If $\beta-\eta^2/(2\nu)<1-p\leq \beta$, then $a\leq q^*<a+\mu$. Consider a feasible function $\tilde{f}$ that decreases on $[a,q^*]$ with derivative $-\nu$. Then the $p$-quantile for $\tilde{f}$ is exactly $q^*$. 

If $1-p=\beta-\eta^2/(2\nu)$, then $q^*=a+\mu$. Consider a feasible function $\tilde{f}$ that decreases on $[a,a+\mu-\varepsilon]$ with derivative $-\mu$ and then becomes extremely flat. Here, $\varepsilon$ is a small positive number. Then the $p$-quantile is between $a+\mu-\varepsilon$ and $a+\mu$. As $\varepsilon\rightarrow 0$, we get that $q_{opt}=q^*$.

If $1-p<\beta-\eta^2/(2\nu)$, then $q^*=\infty$. Still, we consider the same feasible function as in the above case. Now the quantile $q$ can be as large as we want, and hence $q_{opt}=\infty=q^*$.
\par 
In conclusion, $q^*$ defined in \eqref{q^*} is exactly the optimal value of \eqref{quantile optimization problem}.

Using \eqref{optimal value}, we can easily derive an explicit expression of $q^*$ since the non-constant part of $z^*$ is actually a quadratic function.  
\end{proof}
\subsection{Proof of Proposition \ref{conserv_5}}
\begin{proof}
It is known that $\bar{F}\in RV_{-1/\xi}$. We also have that $\overline{F}$ is absolutely continuous and that $(\overline{F})'=-f$ is monotone. Then we can apply the Proposition 0.7 in \citet{resnickExtremeValuesRegular1987} to get that $f\in RV_{-1/\xi-1}$ and that $\lim_{x\rightarrow\infty}-\frac{xf(x)}{\overline{F}(x)}=-\frac{1}{\xi}.$ Also, since $f$ is convex based on our assumption, $f$ is absolutely continuous and that $f'$ is monotone. We can then apply Proposition 0.7 again to get that $\lim_{x\rightarrow\infty}\frac{xf'(x)}{f(x)}=-\frac{1}{\xi}-1.$

In representation \eqref{Karamata}, we require that $\lim_{x\rightarrow\infty}u(x)/x=\xi$ which gives
\begin{align*}
&\lim_{x\rightarrow\infty}\frac{f(x)u(x)}{\bar{F}(x)}=\lim_{x\rightarrow\infty}\frac{xf(x)}{\bar{F}(x)}\frac{u(x)}{x}=1,\\
&\lim_{x\rightarrow\infty}-\frac{f'(x)u^2(x)}{(\xi+1)\bar{F}(x)}=\lim_{x\rightarrow\infty}-\frac{1}{\xi+1}\frac{xf'(x)}{f(x)}\frac{xf(x)}{\bar{F}(x)}\frac{u^2(x)}{x^2}=1.
\end{align*}
\end{proof}

\subsection{Proof of Theorem \ref{conserv_6}}
\begin{proof}
If $x>1/(\xi+1)$, then $\lim_{a\rightarrow\infty}(b-a)/\mu=(\xi+1)x>1$. Thus, for sufficiently large $a$, we must have that $b-a>\mu$. By Theorem \ref{optimal value theorem}, we get that $z^*=\frac{\nu}{2}(\sigma-\mu^2)$. Hence, 
\begin{align*}
    \lim_{a\rightarrow\infty}\frac{z^*(a,b)}{\bar{F}(a)}=\lim_{a\rightarrow\infty}\frac{z^*(a,b)}{\beta}=\lim_{a\rightarrow\infty}\left(1-\frac{\nu\mu^2}{2\beta}\right)=1-\frac{1}{2(\xi+1)}.
\end{align*}

If $x<1/(\xi+1)$, then similarly, for sufficiently large $a$, we have that $b-a<\mu$, and thus $z^*=\frac{\nu}{2}[\sigma-2(b-a)\mu+(b-a)^2]$. Hence, 
\begin{align*}
    \lim_{a\rightarrow\infty}\frac{z^*(a,b)}{\bar{F}(a)}=\lim_{a\rightarrow\infty}\frac{z^*(a,b)}{\beta}=\lim_{a\rightarrow\infty}\left(1-xu(a)\frac{\nu\mu}{\beta}+x^2u^2(a)\frac{\nu}{2\beta}\right)=1-x+\frac{\xi+1}{2}x^2.
\end{align*}

Clearly, if $x=1/(\xi+1)$, then the limit is the same as the case that $x>1/(\xi+1)$. Therefore, combining with $\lim_{a\rightarrow\infty}\bar{F}(b)/\bar{F}(a)= (1+\xi x)^{-1/\xi}$, we get \eqref{Frechet error}.
\end{proof}
\subsection{Proof of Proposition \ref{conserv_7}}
\begin{proof}
For the proof of representation \eqref{von Mises representation} see Example 3.3.23 in \citet{embrechtsModellingExtremalEvents1997}. \eqref{eta_Gumbel} follows directly from the expression of $u(x)$. \eqref{nu_Gumbel} follows from Assumption \ref{additional assumption}.
\end{proof}
\subsection{Proof of Theorem \ref{conserv_8}}
\begin{proof}
Similar to the proof of Theorem \ref{Frechet results}, using the results in Proposition \ref{conserv_7}, we get that 
\begin{align*}
    \lim_{a\rightarrow\infty}\frac{z^*(a,b)}{\bar{F}(a)}=\lim_{a\rightarrow\infty}\frac{z^*(a,b)}{\beta}
    =\begin{cases}
    \frac{1}{2} & \text{if }x\geq1;\\
    1-x+\frac{1}{2}x^2 & \text{if }x<1.
    \end{cases}
\end{align*}
Moreover, it is known that $\lim_{a\rightarrow\infty}\bar{F}(b)/\bar{F}(a)=e^{-x}$ \citep{embrechtsModellingExtremalEvents1997}, and thus we get \eqref{Gumbel error}.
\end{proof}
\subsection{Proof of Theorem \ref{conserv_9}}
\begin{proof}
We have that 
\begin{align*}
\lim_{a\rightarrow\infty}\frac{q^*}{q}=\lim_{a\rightarrow\infty}\frac{(q^*-a)/a+1}{(q-a)/a+1}.
\end{align*}

First we deal with $\lim_{a\rightarrow\infty}(q^*-a)/a$. We have shown that for sufficiently large $a$, 
\begin{align*}
q^*=a+\mu-\sqrt{\mu^2-\sigma+\frac{2(1-p)}{\nu}}.
\end{align*}

Then 
{
\footnotesize
\begin{align*}
\lim_{a\rightarrow\infty}\frac{q^*-a}{a}=\lim_{a\rightarrow\infty}\frac{\mu-\sqrt{\mu^2-\sigma+\frac{2x\beta}{\nu}}}{\mu}\frac{\mu}{u(a)/(\xi+1)}\frac{u(a)}{(\xi+1)a}=\frac{\xi}{\xi+1}\left(1-\sqrt{1-2(1-x)(\xi+1)}\right).
\end{align*}}

Next we deal with $\lim_{a\rightarrow\infty}(q-a)/a$. We define $U(t)=F^{-1}(1-t^{-1}),t>0$. Then we can write $q$ and $a$ as 
\begin{align*}
q=F^{-1}(p)=U\left(\frac{1}{1-p}\right),a=F^{-1}(1-\beta)=U\left(\frac{1}{\beta}\right).
\end{align*}

It is known that \citep{embrechtsModellingExtremalEvents1997}
\begin{align*}
\lim_{s\rightarrow\infty}\frac{U(st)-U(s)}{u(U(s))}=\frac{t^{\xi}-1}{\xi}.
\end{align*}

Note that $\frac{1}{1-p}/\frac{1}{\beta}=1/x$, so we set $s=1/\beta$, $t=1/x$ and get
\begin{align*}
\lim_{a\rightarrow\infty}\frac{q-a}{u(a)}=\frac{t^{\xi}-1}{\xi}=\frac{x^{-\xi}-1}{\xi}.
\end{align*}

Hence $\lim_{a\rightarrow\infty}(q-a)/a=x^{-\xi}-1$. By combining the two parts, we get \eqref{Frechet error q}.
\end{proof}
\subsection{Proof of Theorem \ref{conserv_10}}
\begin{proof}
For sufficiently large $a$, we know that $q^*\leq a+\mu$. On the other hand, since $p\geq 1-\beta$, we know that the true quantile $q\geq a$. Thus 
\begin{align*}
1\leq \lim_{a\rightarrow\infty}\frac{q^*}{q}\leq\lim_{a\rightarrow\infty}\frac{a+\mu}{a}=1+\lim_{a\rightarrow\infty}\frac{\mu}{u(a)}\frac{u(a)}{a}.
\end{align*}

By Proposition \ref{conserv_7}, we know that $\mu/u(a)\rightarrow 1$. It is also known that $u(a)\rightarrow a=0$. Therefore, we get \eqref{Gumbel error q}.
\end{proof}
\subsection{Proof of Theorem \ref{thm:mon}: Integration by parts method}
\begin{proof}
For the first item of the theorem, without loss of generality, we assume that $\bm{g} = \big(\mathbb{I}(x\geq a),g(x)\big)^\top$ and $\mathbb{S}=\{\bm{\Gamma}\}=\{(\beta,\Gamma)^\top\}$ for some scalars $\beta,\Gamma$ where $g(x):[a,\infty)\to \mathbb{R}$ is an integrable function over $[a,\infty)$ and later we show how the result can be generalized to any region $\mathbb{S}$. We rewrite the program $\mathfrak{P}(h,\bm{g},\{\bm{\Gamma}\},\mathscr{P}^{1}_{a,\eta})$ as
\begin{subequations}\label{eq:org_mon}
\begin{align}
\max_f \ &\int_{a}^{\infty} h(x)f(x)dx \label{eq:org_mon_a}\\
\text{s.t.} \ \ &\int_{a}^\infty f(x)dx =\beta, \label{eq:org_mon_b}\\
&\int_{a}^\infty g(x)f(x)dx =\Gamma, \label{eq:org_mon_c}\\
&f(a)\leq \eta, \label{eq:org_mon_d}\\
&f(x) \ \text{exists, non-increasing and right-continuous for }  x\geq a, \label{eq:org_mon_e}\\
&f(x)\geq 0 \ \text{for} \ x\geq a. \label{eq:org_mon_f}
\end{align}
\end{subequations}
From assumptions that $h(x)$ and $g(x)$ only take nonzero values over $[a,\infty)$, we can focus on the integration starting from $a$. We consider $f(x)$ as the right derivative of the cumulative distribution function $F$. Since the set of discontinuous points of a monotone function is at most countable which do not influence the integration values over $[a,\infty)$, we assume $f(x)$ right-continuous for $x\geq a$ for all those discontinuous points. 

Denote 
\begin{align*}
\tilde{H}(x)&=\int_a^x h(u) du,\quad \quad \tilde{G}(x)=\int_a^x g(u) du.
\end{align*}
As $\tilde{H}$ is continuous, $\tilde{H}(a)=0$ by definition, and $f$ has bounded variation because of~\eqref{eq:org_mon_d},\eqref{eq:org_mon_e} and \eqref{eq:org_mon_f}, we have, using integration by parts, that \eqref{eq:org_mon_a} is equal to
\begin{align*}
\int_{a}^\infty f(x)h(x)dx&=\int_{a}^\infty f(x)d\tilde{H}(x)\\
&=f(x)\tilde{H}(x)|_a^{\infty}-\int_{a}^\infty\tilde{H}(x)df(x)\\
&=-\int_{a}^\infty\tilde{H}(x)df(x)
\end{align*}
where the third equality follows from Lemma$~\ref{lm:ps}$ presented later with $\alpha = 0$ and that $\tilde{H}(x) = O(x)$ as $x\rightarrow \infty$ since $h$ is bounded. 

\eqref{eq:org_mon_b} can be rewritten as 
\begin{align*}
\int_{a}^\infty f(x)dx&=\int_{a}^\infty f(x)d(x-a)\\
&=f(x)(x-a)|_a^{\infty}-\int_{a}^\infty(x-a)df(x)\\
&=-\int_{a}^\infty(x-a)df(x)
\end{align*}
where the third equality follows from Lemma$~\ref{lm:ps}$ again with $\alpha = 0$.

For \eqref{eq:org_mon_c}, as $\tilde{G}$ is continuous and $\tilde{G}(a)=0$ based on definition, we can write 
\begin{align*}
\int_{a}^\infty f(x)g(x)dx&=\int_{a}^\infty f(x)d\tilde{G}(x)\\
&=f(x)\tilde{G}(x)|_a^{\infty}-\int_{a}^\infty\tilde{G}(x)df(x)\\
&=-\int_{a}^\infty\tilde{G}(x)df(x)
\end{align*}
where the third equality follows from Lemma$~\ref{lm:gf}$. 

Finally, since $f(x)\rightarrow 0$ as $x\rightarrow \infty$ by Lemma~\ref{lm:ps} with $\alpha = 0$, we can write \eqref{eq:org_mon_d} as
\begin{align*}
f(a) &= -\int_a^{\infty} df(x).
\end{align*}

Therefore,~\eqref{eq:org_mon} is equivalent to 
\begin{subequations}\label{eq:stg2_mon}
\begin{align}
\max_f \ &-\int_{a}^\infty\tilde{H}(x)df(x),\label{eq:stg2_mona}\\
\text{ s.t. } \ \ &-\int_{a}^\infty(x-a)df(x)=\beta,\label{eq:stg2_monb}\\
&-\int_{a}^\infty\tilde{G}(x)df(x) =\Gamma, \label{eq:stg2_monc}\\
&-\int_a^{\infty} df(x)\leq \eta, \label{eq:stg2_mond}\\
&f(x) \ \text{exists, non-increasing and right-continuous for } x\geq a, \label{eq:stg2_mone}\\
&f(x)\geq 0 \ \text{for} \ x\geq a, \label{eq:stg2_monf}\\
&xf(x) \rightarrow 0  \ \text{as} \ x\rightarrow \infty, \label{eq:stg2_mong}\\
&\tilde{G}(x)f(x) \rightarrow 0  \ \text{as} \  x\rightarrow \infty. \label{eq:stg2_monh}
\end{align}
\end{subequations}

The equivalence of~\eqref{eq:org_mon} and~\eqref{eq:stg2_mon} can be checked as follows: Denote the feasible region of~\eqref{eq:org_mon} as $\mathbb{F}_{1}$ and the feasible region of~\eqref{eq:stg2_mon} as $\mathbb{F}_{2}$. From the above discussion, it easily follows that $\mathbb{F}_{1}\subseteq \mathbb{F}_{2}$. For the other direction, we perform integration by parts for~\eqref{eq:stg2_mona},~\eqref{eq:stg2_monb},~\eqref{eq:stg2_monc},~\eqref{eq:stg2_mond} to obtain~\eqref{eq:org_mon_a},~\eqref{eq:org_mon_b},~\eqref{eq:org_mon_c},~\eqref{eq:org_mon_d} respectively. Hence $\mathbb{F}_{2}\subseteq \mathbb{F}_{1}$. It implies the equivalence of~\eqref{eq:org_mon} and~\eqref{eq:stg2_mon}. 

Finally, let $p(x) = -\frac{f(x)}{\eta}+1$. Then~\eqref{eq:stg2_mon} can be rewritten as 
\begin{subequations}\label{eq:stg3_mon}
\begin{align}
\max_f \ &\eta\int_{a}^\infty\tilde{H}(x)dp(x) \label{eq:stg3_mon_a}\\
\text{ s.t. } \ \ &\eta\int_{a}^\infty(x-a)dp(x)=\beta, \label{eq:stg3_mon_b}\\
&\eta\int_{a}^\infty\tilde{G}(x)dp(x) =\Gamma, \label{eq:stg3_mon_c}\\
&\eta\int_a^{\infty} dp(x)\leq \eta, \label{eq:stg3_mon_d}\\
&p(x) \ \text{exists, non-decreasing and right-continuous for }  x\geq a, \label{eq:stg3_mon_e}\\
& 0 \leq p(x)\leq 1 \ \text{for} \ x\geq a, \label{eq:stg3_mon_f}\\
& p(x) \rightarrow 1 \ \text{as} \ x\rightarrow \infty, \label{eq:stg3_mon_g}\\
& p(x) = 0 \ \text{for} \ x<a, \label{eq:stg3_mon_h}\\
& (1-p(x))x \rightarrow 0 \  \text{for} \ x\rightarrow \infty, \label{eq:stg3_mon_i}\\ 
&\tilde{G}(x)(1-p(x)) \rightarrow 0  \ \text{as} \  x\rightarrow \infty. \label{eq:stg3_mon_j}
\end{align}
\end{subequations}

Since $\tilde{H}(x)=(x-a)=\tilde{G}(x)=0$ at $x=a$, one can uniquely identify, up to measure zero, a non-decreasing, right-continuous $p$ such that $\lim_{x\rightarrow\infty}p(x)=1$ and $p(x) = 0 $ for $x<a$ with a probability measure supported on $[a,\infty)$. Finally, by Lemma~\ref{lm3}, constraints~\eqref{eq:stg3_mon_i} and~\eqref{eq:stg3_mon_j} can be derived from other constraints in this optimization problem. Constraint~\eqref{eq:stg3_mon_d} is included in~\eqref{eq:stg3_mon_f}. We have the equivalent problem as follows
\begin{align*}
\max_f \ &\eta\int_{-\infty}^\infty\tilde{H}(x)dp(x) \\
\text{ s.t. } \ \ &\eta\int_{-\infty}^\infty(x-a)dp(x)=\beta, \\
&\eta\int_{-\infty}^\infty\tilde{G}(x)dp(x) =\Gamma, \\
&p(x) \ \text{exists, non-decreasing and right-continuous for }  x\in \mathbb{R}, \\
& 0 \leq p(x)\leq 1 \ \text{for} \ x\in \mathbb{R}, \\
& p(x) \rightarrow 1 \ \text{as} \ x\rightarrow \infty,\\
& p(x) = 0 \ \text{for} \ x< a.
\end{align*}
This concludes the proof the first half of the theorem.

For the second half of the theorem, we first consider $\underline{\eta}=\bar{\eta}=\eta$ and rewrite the program $\mathfrak{P}\big(h,\bm{g},\{\bm{\Gamma}\},\mathscr{P}^{2}_{a,\eta,\eta,\nu}\big)$ as 
\begin{equation}\label{eq:org_convex}
\begin{aligned}
\max_f \ &\int_{a}^{\infty} h(x)f(x)dx\\
\text{s.t.} \ \ &\int_{a}^\infty f(x)dx =\beta,\\
&\int_{a}^\infty g(x)f(x)dx =\Gamma,\\
&f(a)=f(a+)=\eta,\\
&f_+'(a)\geq -\nu,\\
&f \ \text{convex for }  x\geq a,\\
&f(x)\geq 0 \ \text{for }  x\geq a.
\end{aligned}
\end{equation}

Based on \citet{lamTailAnalysisParametric2017}, the formulation~(\ref{eq:org_convex}) is equivalent to 
\begin{subequations}\label{eq:mid_convex}
\begin{align}
\max_f &\int_{a}^{\infty} h(x)f(x)dx \label{eq:mid_convex_a}\\
\text{ s.t. } \ \ &\int_{a}^\infty f(x)dx =\beta, \label{eq:mid_convex_b}\\
&\int_{a}^\infty g(x)f(x)dx =\Gamma,\label{eq:mid_convex_h}\\
&f(a)=\eta, \label{eq:mid_convex_c}\\
&f_+'(x)\ \text{exists and is non-decreasing and right-continuous for }  x\geq a , \label{eq:mid_convex_d}\\
&-\nu\leq f_+'(x) \leq 0 \ \text{for }  x\geq a,\label{eq:mid_convex_e}\\
&f_+'(x)\rightarrow 0 \ \text{a} \ x\rightarrow \infty,\label{eq:mid_convex_f}\\
&f(x)=\int_{a}^x f_+'(t)dt+\eta \ \text{for} \ x\geq a.\label{eq:mid_convex_g}
\end{align}
\end{subequations}
Here $f(a+)$ denotes the right limit at $a$, and $f(a)=f(a+)$ means that $f$ is right-continuous at $a$, implying a continuous extrapolation at $a$.

Denote 
\begin{align*}
\tilde{H}(x)&=\int_a^x h(u) du,&& H(x)=\int_a^x \tilde{H}(u) du,\\
\tilde{G}(x)&=\int_a^x g(u) du,&& G(x)=\int_a^x \tilde{G}(u) du.
\end{align*}

Consider the objective function \eqref{eq:mid_convex_a}. Since $\tilde{H}$ and $H$ are continuous, $f$ is absolutely continuous by \eqref{eq:mid_convex_g} and $f_{+}'$ has bounded variation because of \eqref{eq:mid_convex_e} and \eqref{eq:mid_convex_f}, we have, using integration by parts,
\begin{equation}
\label{fh-1}
\begin{aligned}
\int_{a}^\infty f(x)h(x)dx&=\int_{a}^\infty f(x)d\tilde{H}(x) \\
&=f(x)\tilde{H}(x)|_a^{\infty}-\int_{a}^\infty\tilde{H}(x)f_+'(x)dx\\
&=-\int_{a}^\infty\tilde{H}(x)f_+'(x)dx\\
&=-H(x)f_+'(x)|_a^{\infty}+\int_{a}^\infty H(x)d f_+'(x) \\
&=\int_{a}^\infty H(x)d f_+'(x)
\end{aligned}
\end{equation}
where the third equality follows from Lemma~\ref{lm:ps} with $\alpha = 0$ and that $\tilde{H}(x)= O(x)$ as $x\rightarrow \infty$ since $h$ is bounded. The fifth equality follows from Lemma~\ref{lm:ps} again with $\alpha = 0$ and that $H(x)= O(x^2)$ as $x\rightarrow \infty$.

For \eqref{eq:mid_convex_b}, we can write
\begin{align*}
\int_{a}^\infty f(x)dx&=\int_{a}^\infty f(x)d(x-a)\\
&=f(x)(x-a)|_a^{\infty}-\int_{a}^\infty(x-a)f_+'(x)dx\\
&=-\int_{a}^\infty(x-a)f_+'(x)dx\\
&=-\frac{(x-a)^2}{2}f_+'(x)|_a^{\infty}+\int_{a}^\infty \frac{(x-a)^2}{2}d f_+'(x)\\
&=\int_{a}^\infty \frac{(x-a)^2}{2}d f_+'(x)
\end{align*}
by merely viewing $h\equiv 1$ in \eqref{fh-1}.

For \eqref{eq:mid_convex_c}, note that $f(x)\rightarrow 0$ as $x\rightarrow \infty$ from Lemma~\ref{lm:ps} with $\alpha = 0$, we can use integration by parts again to write
\begin{align*}
f(a)&=-\int_{a}^\infty f_+'(x)dx=-\int_{a}^\infty f_+'(x)d (x-a)\\
&=-f_+'(x)(x-a)|_a^{\infty}+\int_{a}^\infty(x-a)df_+'(x)\\
&=\int_{a}^\infty(x-a)df_+'(x).
\end{align*}

For \eqref{eq:mid_convex_h}, since $\tilde{G}$ and $G(x)$ are continuous, we have, using integration by parts, 
\begin{align*}
\int_{a}^\infty f(x)g(x)dx&=\int_{a}^\infty f(x)d\tilde{G}(x)\\
&=f(x)\tilde{G}(x)|_a^{\infty}-\int_{a}^\infty\tilde{G}(x)f_+'(x)dx\\
&=-\int_{a}^\infty\tilde{G}(x)f_+'(x)dx\\
&=-G(x)f_+'(x)|_a^{\infty}+\int_{a}^\infty G(x)d f_+'(x)\\
&=\int_{a}^\infty G(x)d f_+'(x)
\end{align*}
where the third equality and fifth equality follow from Lemma~\ref{lm:gf}. Therefore, \eqref{eq:mid_convex} is equivalent to 
\begin{subequations}\label{eq:mid2_convex}
\begin{align}
\max_f &\int_{a}^\infty H(x)d f_+'(x) \label{eq:mid2_convex_a}\\
\text{ s.t. } \ \ &\int_{a}^\infty \frac{(x-a)^2}{2}d f_+'(x) =\beta, \label{eq:mid2_convex_b}\\
&\int_{a}^\infty(x-a)df_+'(x)=\eta, \label{eq:mid2_convex_c}\\
&\int_{a}^\infty G(x)d f_+'(x) =\Gamma, \label{eq:mid2_convex_h}\\
&f_+'(x)\ \text{exists and is non-decreasing and right-continuous for }  x\geq a , \label{eq:mid2_convex_d}\\
&-\nu\leq f_+'(x) \leq 0 \ \text{for }  x\geq a,\label{eq:mid2_convex_e}\\
&f_+'(x)\rightarrow 0 \ \text{as} \ x\rightarrow \infty,\label{eq:mid2_convex_f}\\
&f(x)=\int_{a}^x f_+'(t)dt+\eta \ \text{for} \ x\geq a,\label{eq:mid2_convex_g}\\
&xf(x),x^2f_+'(x)\rightarrow 0 \  \text{as} \ x\rightarrow \infty, \label{eq:mid2_convex_i}\\
&\tilde{G}(x)f(x),G(x)f_+'(x)\rightarrow 0 \  \text{as} \ x\rightarrow \infty \label{eq:mid2_convex_j}
\end{align}
\end{subequations}
and the constraint \eqref{eq:mid2_convex_g} states that $f$ can be recovered from $f(x)=\int_{a}^x f_+'(t)dt+\eta$. Note that this definition of $f$ has a right derivative coinciding with the obtained $f_+'(x)$.

The equivalence of~\eqref{eq:mid_convex} and~\eqref{eq:mid2_convex} can be checked as follows: 
Denote the feasible region of~\eqref{eq:mid_convex} as $\mathbb{F}_{1}$ and the feasible region of~\eqref{eq:mid2_convex} as $\mathbb{F}_{2}$. From the above discussion, it easily follows that $\mathbb{F}_{1}\subseteq \mathbb{F}_{2}$. For the other direction, we perform integration by parts for~\eqref{eq:mid2_convex_a},~\eqref{eq:mid2_convex_b},~\eqref{eq:mid2_convex_c},~\eqref{eq:mid2_convex_h} to obtain~\eqref{eq:mid_convex_a},~\eqref{eq:mid_convex_b},~\eqref{eq:mid_convex_c},~\eqref{eq:mid_convex_h} respectively. Hence $\mathbb{F}_{2}\subseteq \mathbb{F}_{1}$. It implies the equivalence of~\eqref{eq:mid_convex} and~\eqref{eq:mid2_convex}.

We now show that \eqref{eq:mid2_convex_i} and \eqref{eq:mid2_convex_j} are redundant. For $f_{+}'(x)$ satisfying \eqref{eq:mid2_convex_d}, \eqref{eq:mid2_convex_e} and \eqref{eq:mid2_convex_f}, we know from \eqref{eq:mid2_convex_b} that 
\begin{align*}
\int_{x}^\infty \frac{(t-a)^2}{2}d f_+'(t)\rightarrow 0 \ \text{as} \ x\rightarrow \infty.
\end{align*}
Note that with the non-decreasing property of $f_{+}'(x)$ via \eqref{eq:mid2_convex_d}, we have the following inequality
\begin{align*}
\int_{x}^\infty \frac{(t-a)^2}{2}d f_+'(t) \geq \frac{(x-a)^2}{2}\int_{x}^\infty d f_+'(t) \geq 0.
\end{align*}
Hence, with $f_+'(x)\rightarrow 0$ as $x\rightarrow \infty$ via \eqref{eq:mid2_convex_f}, we have
\begin{align}\label{x^2f'(x)}
\frac{(x-a)^2}{2}f_+'(x)\rightarrow 0  \ \text{as} \ x\rightarrow \infty.
\end{align} 

For \eqref{eq:mid2_convex_c}, we can write 
\begin{align*}
\eta&=\int_{a}^\infty(x-a)df_+'(x)=(x-a)f_+'(x)|_{a}^{\infty}-\int_{a}^\infty f_+'(x)dx=\int_{a}^\infty -f_+'(x)dx
\end{align*}
where the third equality follows from \eqref{x^2f'(x)}. It is easy to conclude that 
\begin{align}\label{f(infty)}
\int_{x}^\infty -f_+'(t)dt \rightarrow 0 \ \text{as} \ x\rightarrow \infty.
\end{align}

From \eqref{eq:mid2_convex_g} we can easily derive from \eqref{eq:mid2_convex_e} and \eqref{f(infty)} that
\begin{align}
f(x) &= \int_{a}^x f_+'(t)dt+\eta= \int_{a}^x f_+'(t)dt+\int_{a}^\infty -f_+'(x)dx= \int_{x}^\infty -f_+'(t)dt\geq 0\label{def_f}\\
\implies f(x)&\rightarrow 0 \ \text{as} \ x\rightarrow \infty.\notag
\end{align}

Then from \eqref{eq:mid2_convex_b}, we can write
\begin{align*}
\beta&=\int_{a}^\infty \frac{(x-a)^2}{2}d f_+'(x)\\
&=f_+'(x)\frac{(x-a)^2}{2}|_{a}^{\infty}-\int_{a}^\infty f_+'(x)(x-a) dx\\
&= \int_{a}^\infty -f_+'(x)(x-a) dx
\end{align*}
where the third equality follows from \eqref{x^2f'(x)}. Then it is easy to conclude that 
\begin{align*}
\int_{x}^\infty-f_+'(t)(t-a) dt \rightarrow 0 \ \text{as} \ x\rightarrow \infty.
\end{align*}
Note that with the non-positive property of $f_{+}'(x)$ via \eqref{eq:mid2_convex_e}, we have the following inequality
\begin{align*}
\int_{x}^\infty-f_+'(t)(t-a) dt \geq (x-a)\int_{x}^\infty-f_+'(t) dt \geq 0.
\end{align*}
Hence, with equation of $f(x)$ in \eqref{def_f}, we have
\begin{align*}
(x-a)f(x)\rightarrow 0  \ \text{as} \ x\rightarrow \infty
\end{align*}
which concludes the redundancy of constraint \eqref{eq:mid2_convex_i}.

Now we show the redundancy of constraint \eqref{eq:mid2_convex_j}. We first consider $g(x)$ is a non-negative function. Then $G(x)$ and $\tilde{G}(x)$ are non-decreasing non-negative continuous functions. 

From \eqref{eq:mid2_convex_h}, we have 
\begin{align*}
\int_{x}^\infty G(t)d f_+'(t) \rightarrow 0 \ \text{as} \ x\rightarrow \infty.
\end{align*}
Note that with the non-decreasing property of $f_{+}'(x)$ via \eqref{eq:mid2_convex_d} and $G(x)$, we have the following inequality
\begin{align*}
\int_{x}^\infty G(t)d f_+'(t) \geq G(x)\int_{x}^\infty d f_+'(t) \geq 0.
\end{align*}
Hence, with $f_+'(x)\rightarrow 0$ as $x\rightarrow \infty$ via \eqref{eq:mid2_convex_f}, we have
\begin{align}\label{G(x)f'(x)}
G(x)f_+'(x)\rightarrow 0  \ \text{as} \ x\rightarrow \infty.
\end{align} 

For \eqref{eq:mid2_convex_h}, we can write
\begin{align*}
\Gamma &= \int_{a}^\infty G(x)d f_+'(x) \\
&=f_+'(x)G(x)|_a^{\infty}-\int_{a}^\infty f_+'(x)\tilde{G}(x)dx\\
&=\int_{a}^\infty -f_+'(x)\tilde{G}(x)dx
\end{align*}
where the third equality follows from \eqref{G(x)f'(x)}. Then it leads to 
\begin{align*}
\int_{x}^\infty-f_+'(t)\tilde{G}(t) dt \rightarrow 0 \ \text{as} \ x\rightarrow \infty.
\end{align*}
Note that with the non-positive property of $f_{+}'(x)$ via \eqref{eq:mid2_convex_e} and the non-decreasing property of $\tilde{G}(x)$, we have the following inequality
\begin{align*}
\int_{x}^\infty-f_+'(t)\tilde{G}(t) dt \geq \tilde{G}(x)\int_{x}^\infty-f_+'(t) dt \geq 0.
\end{align*}
Hence, with equation of $f(x)$ in \eqref{def_f}, we have
\begin{align*}
\tilde{G}(x)f(x)\rightarrow 0  \ \text{as} \ x\rightarrow \infty.
\end{align*}

Now we consider the case when $g(x)$ is a bounded-below function for $x\geq a$ and the value of $g(x)$ can be negative. We consider $\tilde{g}(x)=g(x)+|\min_{x\geq a} g(x)|$. Clearly, $\tilde{g}(x)$ and $|\min_{x\geq a} g(x)|$ are non-negative functions, which implies we can use the results above. The results for $g(x)$ hence follow by linearity of integration and sum law of limits. A detailed exposition is as follows: 

Given 
\begin{align*}
\tilde{G}(x) & = \int_a^x \tilde{g}(t)dt + \int_a^x -|\min_{x\geq a} g(x)|dt\\
&=\int_a^x \tilde{g}(t)dt  -(x-a)|\min_{x\geq a} g(x)|,\\
G(x) & = \int_a^x\int_a^v \tilde{g}(u)dudv -\frac{(x-a)^2}{2}|\min_{x\geq a} g(x)|,
\end{align*}
we have 
\begin{align*}
&\lim_{x\rightarrow\infty} \int_a^x \tilde{g}(t)dt f(x) \rightarrow 0,\ \lim_{x\rightarrow\infty} -(x-a)|\min_{x\geq a} g(x)| f(x) \rightarrow 0\\
\implies& \  \lim_{x\rightarrow\infty} \tilde{G}(x) f(x) \rightarrow 0,\\
& \lim_{x\rightarrow\infty} \int_a^x\int_a^v \tilde{g}(u)dudv f_+'(x) \rightarrow 0,\ \lim_{x\rightarrow\infty} -\frac{(x-a)^2}{2}|\min_{x\geq a} g(x)| f_+'(x) \rightarrow 0\\
\implies& \  \lim_{x\rightarrow\infty} G(x) f(x) \rightarrow 0,
\end{align*}
which concludes the redundancy of constraint \eqref{eq:mid2_convex_j}.

Finally, let $p(x)=\frac{f_+'(x)}{\nu}+1$. Then \eqref{eq:mid2_convex} is equivalent to 

\begin{equation}\label{eq:mid3_convex}
\begin{aligned}
\max_p &\int_{a}^\infty \nu H(x)d p(x)\\
\text{ s.t. } \ \ &\int_{a}^\infty \nu\frac{(x-a)^2}{2}d p(x) =\beta,\\
&\int_{a}^\infty\nu(x-a)dp(x)=\eta,\\
&p(x)\ \text{exists and is non-decreasing and right-continuous for }  x\geq a, \\
& 0 \leq p(x) \leq 1 \ \text{for }  x\geq a,\\
&  p(x) \rightarrow 1 \ \text{for }  x\rightarrow \infty, \\
&\int_{a}^\infty \nu G(x) dp(x) =\Gamma
\end{aligned}
\end{equation}
or equivalently
\begin{equation}\label{eq:mid4_convex}
\begin{aligned}
\max_p &\int_{-\infty}^\infty \nu H(x)d p(x) \\
\text{ s.t. } \ \ &\int_{-\infty}^\infty \nu\frac{(x-a)^2}{2}d p(x) =\beta, \\
&\int_{-\infty}^\infty\nu(x-a)dp(x)=\eta,\\
&p(x)\ \text{exists and is non-decreasing and right-continuous for }  x\in \mathbb{R},\\
& 0 \leq p(x) \leq 1 \ \text{for }  x\in \mathbb{R},\\
&  p(x) \rightarrow 1 \ \text{for }  x\rightarrow \infty, \\
& p(x) = 0 \ \text{for } x< a, \\ 
&\int_{-\infty}^\infty \nu G(x) dp(x) =\Gamma.
\end{aligned}
\end{equation}

Since $H(x)=(x-a)^2=(x-a)=G(x)=0$ at $x=a$, one can uniquely identify, up to measure zero, a non-decreasing, right-continuous $p$ such that $\lim_{x\rightarrow\infty}p(x)=1$ and $p(x) = 0 $ for $x<a$ with a probability measure supported on $[a,\infty)$. Hence \eqref{eq:mid4_convex} is equivalent to \eqref{eq:mid3_convex}. 

When $\underline{\eta}\neq\bar{\eta}$, one can replace the equality constraint $\int_{-\infty}^\infty\nu(x-a)dp(x)=\eta$ in \eqref{eq:mid4_convex} by $\underline{\eta}\leq\int_{-\infty}^\infty\nu(x-a)dp(x)\leq \bar{\eta}$ which forms a rectangular constraint. Note that the above derivation holds true for any choice of $\beta$ and $\Gamma$ so that the result still holds when replacing $\{\bm{\Gamma}\}=\{(\beta,\Gamma)^\top\}$ with general $\mathbb{S}$. This concludes the result.
\end{proof}
\subsection{Proof of Theorem \ref{thm:mon}: Choquet Method}
\begin{proof}
Without loss of generality, we assume that $\bm{g} = \big(\mathbb{I}(x\geq a),g(x)\big)^\top$ and $\mathbb{S}=\{\bm{\Gamma}\}=\{(\beta,\Gamma)^\top\}$ for some scalars $\beta,\Gamma$ where $g(x):[a,\infty)\to \mathbb{R}$ is an integrable function over $[a,\infty)$ and then we show how the result can be generalized to any region $\mathbb{S}$. Notice that for any feasible solution $f$ in program (\ref{eq:org_mon}), we can utilize Theorem \ref{Choquet_mono} and obtain 
\begin{align*}
\int_a^\infty h(x)f(x)dx&=\beta \int_a^\infty\int_{0+}^\infty h(x)\frac{\mathbb{I}(a\leq x<a+z)}{z}dQ(z)dx\\
&=\beta \int_{0+}^\infty\int_a^\infty h(x)\frac{\mathbb{I}(a\leq x<a+z)}{z}dxdQ(z)\\
&=\beta \int_{0+}^\infty\frac{\int_a^{a+z} h(x)dx}{z}dQ(z)=\beta \int_{0+}^\infty\frac{\tilde{H}(z+a)}{z}dQ(z),\\
\int_a^\infty g(x)f(x)dx&=\beta \int_{0+}^\infty\frac{\tilde{G}(z+a)}{z}dQ(z),\\
\int_a^\infty f(x)dx & = \beta\int_a^\infty\int_{0+}^\infty \frac{\mathbb{I}(a\leq x<a+z)}{z}dQ(z)dx\\
&=\beta\int_{0+}^\infty\int_a^\infty \frac{\mathbb{I}(a\leq x<a+z)}{z}dxdQ(z)\\
&=\beta \int_{0+}^\infty dQ(z)=\beta,\\
f(a) &= \int_{0+}^\infty \frac{\mathbb{I}(a\leq x<a+z)}{z}dQ(z)=\int_{0+}^\infty \frac{1}{z}dQ(z).
\end{align*}

Therefore, program (\ref{eq:org_mon}) is equivalent to the following program, 
\begin{align*}
\max_{Z\sim Q} \ &\beta\mathbb{E}_Q\Big[\frac{\tilde{H}(Z+a)}{Z}\Big] \\
\text{ s.t. } \ \ &\beta\mathbb{E}_Q [1] =\beta,\\
&\mathbb{E}_Q\Big[\frac{\tilde{G}(Z+a)}{Z}\Big] =\frac{\Gamma}{\beta}, \\
&\mathbb{E}_Q\big[\frac{1}{Z}\big]\leq \frac{\eta}{\beta}, \\
&Q\in \mathscr{M}^+(0,\infty)
\end{align*}
where $\mathscr{M}^+(0,\infty)$ is the space of non-negative bounded measures on $(0, \infty)$.

Since $\mathbb{E}_Q[\frac{1}{Z}]<\infty$, we can define a distribution function $\tilde{Q}\in \mathscr{M}^+(0,\infty)$ absolutely continuous with respect to $Q$ via $\frac{d\tilde{Q}}{dQ}(z)=\frac{\beta}{\eta}\frac{1}{z}$. We convert the decision variable from $Q$ to $\tilde{Q}$. Furthermore, the feasible set can be further restricted to $\mathscr{P}[0,\infty)$. That is because the functions $\tilde{H}(z+a),\tilde{G}(z+a),z$ are by construction equal to $0$ at $z=0$. Hence we can always add an arbitrary mass at $0$ to reach the upper bound of $\mathbb{E}_{\tilde{Q}}[1]$. This in turn deduces that upon proper normalization of the measure we can impose the constraint that $\mathbb{E}_{\tilde{Q}}[1]=1$. Finally, we let $x=z+a$ and obtain that the the following program is equivalent to program (\ref{eq:org_mon}):
\begin{align*}
\max_{X\sim\tilde{Q}} \ \  &\eta\mathbb{E}_{\tilde{Q}}[\tilde{H}(X)] \\
\text{ s.t. } \ \ &\eta\mathbb{E}_{\tilde{Q}} [X-a] =\beta,\\
&\eta\mathbb{E}_{\tilde{Q}}[\tilde{G}(X)] =\Gamma,\\
&{\tilde{Q}}\in \mathscr{P}[a,\infty).
\end{align*}

For the second part of the theorem, we first consider $\underline{\eta}=\bar{\eta}=\eta$. For any feasible solution $f$ in program (\ref{eq:org_convex}), Lemma 1 of \citet{lamTailAnalysisParametric2017} implies that $f$ is non-increasing. Then based on Lemma 5.1 of \citet{popescuSemidefiniteProgrammingApproach2005}, $f(x),\forall x\geq a$ can be written as a generalized mixture of right $a$-triangular density, i.e.,
\begin{align*}
f(x)=\beta \int_0^\infty \frac{(a+z-x)^+}{z^2/2}dQ(z)
\end{align*}
where $Q$ is a probability measure on $[0,\infty)$. Since $f$ exists everywhere for $x\geq a$, we have $Q(z=0)=0$.

Then since $\tilde{H}(x),\tilde{G}(x),(a+z-x)$ are all absolutely continuous, we have, using integration by parts, 
\begin{align*}
\int_{a}^{\infty} h(x)f(x)dx &= \beta\int_{a}^{\infty} h(x)\int_0^\infty \frac{(a+z-x)^+}{z^2/2}dQ(z)dx \\
&=\beta \int_0^\infty \frac{\int_{a}^{\infty}h(x)(a+z-x)^+dx}{z^2/2}dQ(z)\\
&=\beta \int_0^\infty \frac{\int_{a}^{a+z}h(x)(a+z-x)dx}{z^2/2}dQ(z)\\
&=\beta \int_0^\infty \frac{\tilde{H}(x)(a+z-x)|_{x=a}^{x=a+z}+\int_{a}^{a+z}\tilde{H}(x)dx}{z^2/2}dQ(z)\\
&=\beta \int_0^\infty \frac{\int_{a}^{a+z}\tilde{H}(x)dx}{z^2/2}dQ(z)=\beta \int_{0+}^\infty \frac{H(a+z)}{z^2/2}dQ(z),\\
\int_{a}^{\infty} g(x)f(x)dx &= \beta\int_{0+}^\infty \frac{G(a+z)}{z^2/2}dQ(z),\\
\int_{a}^{\infty} f(x)dx &= \beta\int_{a}^{\infty} \int_0^\infty \frac{(a+z-x)^+}{z^2/2}dQ(z)dx =\beta \int_0^\infty \frac{\int_{a}^{\infty}(a+z-x)^+dx}{z^2/2}dQ(z)\\
&=\beta \int_0^\infty \frac{\int_{a}^{a+z}x(a+z-x)dx}{z^2/2}dQ(z)\\
&=\beta \int_0^\infty \frac{(x-a)(a+z-x)|_{x=a}^{x=a+z}+\int_{a}^{a+z}(x-a)dx}{z^2/2}dQ(z)\\
&=\beta \int_0^\infty \frac{\int_{a}^{a+z}(x-a)dx}{z^2/2}dQ(z)=\beta \int_{0+}^\infty \frac{z^2/2}{z^2/2}dQ(z)=\beta,\\
f(a)&=\beta \int_{0+}^\infty \frac{z}{z^2/2}dQ(z),\\
-f_+'(a)&=\lim\limits_{\delta_n \downarrow 0}\frac{-f(a+\delta_n)+f(a)}{\delta_n}=\beta\lim\limits_{\delta_n \downarrow 0}  \int_0^\infty \frac{z-(z-\delta_n)^+}{\delta_n z^2/2}dQ(z)\\
&=\beta  \lim\limits_{\delta_n \downarrow 0}\int_{0+}^\infty \frac{z-(z-\delta_n)^+}{\delta_n z^2/2}dQ(z)=\beta\int_{0+}^\infty\frac{1}{z^2/2}dQ(z).
\end{align*}

Since $0\leq \frac{z-(z-\delta_n)^+}{\delta_nz^2/2}\leq \frac{z-(z-\delta_{n+1})^+}{\delta_{n+1}z^2/2}$, the exchange of limit and integration is followed by monotone convergence theorem. Therefore, program (\ref{eq:org_convex}) is equivalent to the following program, 
\begin{align*}
\max_{Z\sim Q} \ &\beta\mathbb{E}_Q\Big[\frac{H(Z+a)}{Z^2/2}\Big] \\
\text{ s.t. } \ \ &\beta\mathbb{E}_Q [1] =\beta,\\
&\beta\mathbb{E}_Q\Big[\frac{G(Z+a)}{Z^2/2}\Big] =\Gamma, \\
&\beta\mathbb{E}_Q\big[\frac{Z}{Z^2/2}\big]=\eta, \\
&\beta\mathbb{E}_Q[\frac{1}{Z^2/2}]\leq \nu,\\
&Q\in \mathscr{M}^+(0,\infty).
\end{align*}

Since $\mathbb{E}_Q[\frac{1}{Z^2/2}]<\infty$, we can define a distribution function $\tilde{Q}\in \mathscr{M}^+(0,\infty)$ absolutely continuous with respect to $Q$ via $\frac{d\tilde{Q}}{dQ}(z)=\frac{\beta}{\nu}\frac{1}{z^2/2}$. We convert the decision variable from $Q$ to $\tilde{Q}$. Furthermore, the feasible set can be further restricted to $\mathscr{P}[0,\infty)$. That is because the functions $H(z+a),G(z+a),z,z^2/2$ are by construction equal to $0$ at $z=0$. Hence we can always add an arbitrary mass at $0$ to reach the upper bound of $\mathbb{E}_{\tilde{Q}}[1]$. This in turn deduces that upon proper normalization of the measure we can impose the constraint that $\mathbb{E}_{\tilde{Q}}[1]=1$. Finally, we let $x=z+a$ and obtain that the following program is equivalent to program (\ref{eq:org_convex}):
\begin{align*}
\max_{X\sim \tilde{Q}} \ \  &\nu\mathbb{E}_{\tilde{Q}}[H(X)] \\
\text{ s.t. } \ \ &\nu\mathbb{E}_{\tilde{Q}} \Big[\frac{(X-a)^2}{2}\Big] =\beta,\\
&\nu\mathbb{E}_{\tilde{Q}}[G(X)] =\Gamma,\\
&\nu\mathbb{E}_{\tilde{Q}}[X-a]=\eta,\\
&{\tilde{Q}}\in \mathscr{P}[a,\infty).
\end{align*}

When $\underline{\eta}\neq\bar{\eta}$, one can replace the equality constraint  $\nu\mathbb{E}_{\tilde{Q}}[X-a]=\eta$ by $\underline{\eta} \leq \nu\mathbb{E}_{\tilde{Q}}[X-a]\leq \bar{\eta}$ which gives a rectangular constraint. Note that the above derivation holds true for any choice of $\beta$ and $\Gamma$ so that the result still holds when replacing $\{\bm{\Gamma}\}=\{(\beta,\Gamma)^\top\}$ with general $\mathbb{S}$. It concludes the proof.
\end{proof}
\subsection{Proof of Theorem \ref{p2d}}
\begin{proof}
Without loss of generality, we consider $r=1$. For positive $r$ other than $1$, one could replace $\bm{\Sigma}$ by $r\bm{\Sigma}$ and the rest of the derivation is the same. We can consider $\mathbb{S}$ to be pure hyper-ellipsoid or pure hyper-rectangle separately and the final result is a combination of both. For the hyper-ellipsoid scenario, we rewrite $\mathfrak{M}_a\big( H,\bm{G},\mathbb{S}_{E}(\bm{\mu},\bm{\Sigma},1)\big)$, with given constant values $\bm{\mu},\bm{\Sigma},\bm{\bar{\mu}},\bm{\underline{\mu}},a$ as 
\begin{equation}
\label{eq:trans_one_primal_2}
\begin{aligned}
\max_{P} \ \ \  & \mathbb{E}_P[H(X)]\\
\text{ s.t. } \ \ 
&\begin{bmatrix}
\bm{\Sigma}^{-1/2}\big(\mathbb{E}_P[\bm{G}(X)]-\bm{\mu}\big)\\
1
\end{bmatrix}\in \mathbb{K},\\
&P\in \mathscr{P}[a,\infty)
\end{aligned}
\end{equation}
where $\bm{\Sigma}^{-1/2}$ is the lower-triangular square-root matrix of $\bm{\Sigma}^{-1}$ obtained by Cholesky decomposition, and $\mathbb{K}$ denotes the second order cone $\{(x_1,\dots,x_d,x_{d+1})^\top \in \mathbb{R}^{d+1}:x_{d+1}\geq \sqrt{x_1^2+\dots+x_d^2}\}$.

The dual problem is derived as follows: 
The first step is to build a Lagrangian for this program \eqref{eq:trans_one_primal_2}. Notice that with variables $\lambda\geq0,\kappa, P\in \mathscr{M}^+[a,\infty)$, the following holds
\begin{align*}
& \mathbb{E}_P[H(X)]+\lambda\big(1-\|\bm{\Sigma}^{-1/2}( \mathbb{E}_P[\bm{G}(X)]-\bm{\mu})\|_2\big)+\kappa(1-\mathbb{E}_P[1])\\
=& \mathbb{E}_P[H(X)]+\lambda\big(1-\max_{\|\bm{u}\|_2\leq 1}\bm{u}^\top \bm{\Sigma}^{-1/2}( \mathbb{E}_P[\bm{G}(X)]-\bm{\mu})\big)+\kappa(1-\mathbb{E}_P[1])\\
=& \mathbb{E}_P[H(X)]+\lambda-\max_{\|\bm{u}\|_2\leq \lambda}\bm{u}^\top\bm{\Sigma}^{-1/2}( \mathbb{E}_P[\bm{G}(X)]-\bm{\mu})+\kappa(1-\mathbb{E}_P[1])\\
=&\min_{\|\bm{u}\|_2\leq \lambda} \lambda+\kappa+\bm{u}^\top\bm{\Sigma}^{-1/2}\bm{\mu}+\mathbb{E}_P[ H(X)-\bm{u}^\top \bm{\Sigma}^{-1/2}\bm{G}(X)-\kappa]
\end{align*}
where $\|\cdot\|_{2}$ denotes the Euclidean norm in $\mathbb{R}^{p}$.

We can write a Lagrangian as follows:
\begin{align*}
L(P,\lambda,\bm{u},\kappa) = \lambda+\kappa+\bm{u}^\top\bm{\Sigma}^{-1/2}\bm{\mu}+\mathbb{E}_P[ H(X)-\bm{u}^\top \bm{\Sigma}^{-1/2}\bm{G}(X)-\kappa].
\end{align*}

One can check that $\mathfrak{P}_a\big( H,\bm{G},\mathbb{S}_{E}(\bm{\mu},\bm{\Sigma})\big)$ can be derived from 
\begin{align*}
\max\limits_{P\in \mathscr{M}^+[a,\infty)}\min\limits_{\|\bm{u}\|_2\leq \lambda,\kappa}L(P,\lambda,\bm{u},\kappa)
\end{align*}
and the corresponding dual problem can be derived from
\begin{align*}
\min\limits_{\kappa,\|\bm{u}\|_2\leq \lambda}\max\limits_{P\in \mathscr{M}^+[a,\infty)}L(P,\lambda,\bm{u},\kappa).
\end{align*}

We then write the dual problem as follows: 
\begin{align*}
\min_{\kappa, \bm{u},\lambda} \ \ \  &\lambda+\kappa+\bm{u}^\top\bm{\Sigma}^{-1/2}\bm{\mu} \\
\text{ s.t. } \ \  &- H(x)+\bm{u}^\top \bm{\Sigma}^{-1/2}\bm{G}(x)+\kappa\geq 0, x\in [a,\infty),\\
&\begin{bmatrix}
\bm{u}\\
\lambda
\end{bmatrix}\in \mathbb{K}.
\end{align*}

Then for the hyper-rectangle scenario, we rewrite $\mathfrak{M}_a\big( H, \bm{G},\mathbb{S}_{R}(\underline{\bm{\mu}},\bar{\bm{\mu}})\big)$ as 
\begin{align*}
\max_{X\sim P} \ \ \  &\mathbb{E}_P[H(X)]\\
\text{ s.t. } \ \ &\mathbb{E}_P[\bm{G}(X)] \leq \bar{\bm{\mu}},\\
&-\mathbb{E}_P[\bm{G}(X)]\leq -\underline{\bm{\mu}},\\
&P\in \mathscr{P}[a,\infty).
\end{align*}

Consider the following Lagrangian, with variables $P,\bm{\lambda_1},\bm{\lambda_2},\kappa$
\begin{align*}
L(P,\bm{\lambda_1},\bm{\lambda_2},\kappa)&=\mathbb{E}_P[H(X)]+\bm{\lambda_1}^\top(\bar{\bm{\mu}}-\mathbb{E}_P[\bm{G}(X)])+\bm{\lambda_2}^\top(-\underline{\bm{\mu}}+\mathbb{E}_P[\bm{G}(X)])+\kappa(1-\mathbb{E}_P[1])\\
&=\bm{\lambda_1}^\top \bar{\bm{\mu}}-\bm{\lambda_2}^\top \underline{\bm{\mu}}+\kappa+\mathbb{E}_P[-\bm{\lambda_1}^\top \bm{G}(X)+\bm{\lambda_2}^\top \bm{G}(X)-\kappa+ H(X)].
\end{align*}

Therefore, we have the dual problem as follows: 
\begin{align*}
\min_{\bm{\lambda_1}\geq0,\bm{\lambda_2}\geq0,\kappa} \ \ \  &\bm{\lambda_1}^\top \bar{\bm{\mu}}-\bm{\lambda_2}^\top \underline{\bm{\mu}}+\kappa\\
\text{ s.t. } \ \  &- H(x)+(\bm{\lambda_1} -\bm{\lambda_2})^\top \bm{G}(x)+\kappa\geq 0,x\in [a,\infty).
\end{align*}

It concludes the result.
\end{proof}
\subsection{Proof of Corollary \ref{coro_4cases}}
\begin{proof}
Using Theorem \ref{thm:mon}, we can easily get moment-constrained transformations as follows: 

$\mathfrak{P}(h,\bm{g},\mathbb{S}_{E}(\bm{\mu},\bm{\Sigma},r),\mathscr{P}^{1}_{a,\eta})$ $\to$ $\mathfrak{M}_a(\eta\tilde{H},\eta\tilde{\bm{G}},\mathbb{S}_{E}(\bm{\mu},\bm{\Sigma},r)),$ 

$\mathfrak{P}(h,\bm{g},\mathbb{S}_{R}(\underline{\bm{\mu}},\bar{\bm{\mu}}),\mathscr{P}^{1}_{a,\eta})$ $\to$ $\mathfrak{M}_a(\eta\tilde{H},\eta\tilde{\bm{G}},\mathbb{S}_{R}(\underline{\bm{\mu}},\bar{\bm{\mu}})),$

$\mathfrak{P}(h,\bm{g},\mathbb{S}_{E}(\bm{\mu},\bm{\Sigma},r),\mathscr{P}^{2}_{a,\underline{\eta},\bar{\eta},\nu})$ $\to$ $\mathfrak{M}_a(\nu H,\nu(x-a,\bm{G}^\top)^\top,\mathbb{S}_R(\underline{\eta},\bar{\eta})\times \mathbb{S}_{E}(\bm{\mu},\bm{\Sigma},r)),$

$\mathfrak{P}(h,\bm{g},\mathbb{S}_{R}(\underline{\bm{\mu}},\bar{\bm{\mu}}),\mathscr{P}^{2}_{a,\underline{\eta},\bar{\eta},\nu})$ $\to$ $\mathfrak{M  }_a(\nu H,\nu(x-a,\bm{G}^\top)^\top,\mathbb{S}_R(\underline{\eta},\bar{\eta})\times \mathbb{S}_{R}(\underline{\bm{\mu}},\bar{\bm{\mu}})).$

Then using Theorem \ref{p2d}, we obtain the results in the corollaries. 
\end{proof}
\subsection{Proof of Theorem \ref{d2sdp}}
\begin{proof}
Applying Proposition 3.1 from \citet{bertsimasOptimalInequalitiesProbability2005}, the constraints \eqref{polying} are equivalent to that there exists two positive semi-definite matrices $\bm{V}=[v_{ij}]_{i,j=0,\cdots,k}$ and $\bm{W}=[w_{ij}]_{i,j=0,\cdots,k}$ such that 
\begin{equation}\label{mix_F_d_sdp}
\begin{aligned}
\bm{V},\bm{W}&\succeq 0,\\
\sum\limits_{i,j:i+j=2l-1} v_{ij} &=0, \ \ l=1,\cdots,k,\\
\sum\limits_{i,j:i+j=2l} v_{ij}&=\sum\limits_{r=l}^k \binom{r}{l}y_{1r} a^{r-l},\ \ l=0,\cdots,k,\\
\sum\limits_{i,j:i+j=2l-1} w_{ij} &=0, \ \ l=1,\cdots,k,\\
\sum\limits_{i,j:i+j=2l} w_{ij}&=\sum\limits_{m=0}^l \sum\limits_{r=m}^{k+m-l} \binom{r}{m}\binom{k-r}{l-m}y_{2r} b^{r-m}c^{m}, \ \ l=0,\cdots,k.
\end{aligned}
\end{equation}
Replace \eqref{mix_F_d} with \eqref{mix_F_d_sdp} and obtain the result. 
\end{proof}


\subsection{Proof of Theorem \ref{thm:first_order_compactness}}
\begin{proof}
The proof is inspired from Theorem 2 of \citet{gobernaSensitivityAnalysisLinear2007}. For the rest of the poof, for any set $\mathbb{S}$, we denote $\text{int}(\mathbb{S})$ the interior of $\mathbb{S}$, $|\mathbb{S}|$ the cardinality, $\text{cone}(\mathbb{S}):=\{\sum\limits_{i=1}^m \lambda_i\bm{x}_i:\bm{x}_i\in \mathbb{S}, \lambda_i\in \mathbb{R}^+, m\in \mathbb{N}\}$ the conical hull. The set of all natural number is $\mathbb{N}$. $\mathbb{F}^*(\mathfrak{P})$ denotes the optimal region for program $\mathfrak{P}$.  We first review some definitions of semi-infinite linear programming and relevant results that are used later. Further information are found in \citet{gobernaFarkasMinkowskiSystemsSemiinfinite1981, miguela.gobernaLinearSemiinfiniteOptimization2000}.
 \comment{For any program $\mathfrak{P}$, we call the relation $\bm{a}^\top \bm{x}\geq \beta$ with the associated vector $(\bm{a}^\top,\beta)^\top$ a linear consequence relation of the constraints system in $\mathfrak{P}$ if every feasible point in $\mathbb{F}(\mathfrak{P})$ satisfies the relation.According to Theorem 2.1 of \citet{gobernaFarkasMinkowskiSystemsSemiinfinite1981}, a relation $\bm{a}^\top \bm{y}\geq b$ is a linear consequence relation of feasible linear constraint system represented as $\{\bm{a}(x)^\top \bm{y}\geq b(x),x\in \mathbb{S}\}$ if and only if $(\bm{a}^\top,b)^\top$ lies in the closure of the convex cone generated by $\Big\{\big(\bm{a}^\top(x),c(x)\big)^\top;c(x)\leq b(x), x\in \mathbb{S}\Big\}$.  
We say program $\mathfrak{P}$ \textit{Farkas-Minkowski} (FM) if every linear consequence relation of the constraints system of $\mathfrak{P}$ is a linear consequence relation of a finite subsystem.} According to Theorem 3.1 of \citet{gobernaFarkasMinkowskiSystemsSemiinfinite1981}, $\mathfrak{D}_{\mu}$ is FM if and only if the \textit{characteristic cone} $K_c(\mathfrak{D}_{\mu})$ is closed, 
\begin{align*}
K_c(\mathfrak{D}_{\mu}):=\text{cone}\Big(\big(G_0(x),G_1(x),\ldots,G_n(x),H(x)\big)^\top ,x\in \mathbb{\Omega}; (\bm{0}_{n+1}, -1)\Big). 
\end{align*}


The linear constraint system $\{\bm{a}(x)^\top \bm{y}\geq b(x),x\in \mathbb{S}\}$ is \textit{canonically closed} if \textbf{(1)} $\exists \bm{y}^0, \bm{a}(x)^\top \bm{y}^0> b(x),\forall x\in \mathbb{S}$; \textbf{(2)} $\exists \alpha(x):\mathbb{\Omega}\to \mathbb{R}^{++},$ the set $\{\alpha(x)\cdot (\bm{a}(x)^\top,b(x))^\top, x\in \mathbb{\Omega}\}$ is compact.
$\mathfrak{D}_{\mu}$ is FM if the constraint system is \textit{canonically closed} (Corollay 3.1.1 of \citet{gobernaFarkasMinkowskiSystemsSemiinfinite1981}). It is because \textit{canonically closed} implies that $$\text{cone}\Big(\big(G_0(x),G_1(x),\ldots,G_n(x),H(x)\big)^\top ,x\in \mathbb{\Omega}\Big)$$ is closed, which is a sufficient condition for the closedness of $K_c(\mathfrak{D}_{\mu})$. Then we have 
\begin{lemma}\label{lemma:cc2fm}
Program $\mathfrak{P}$ is FM if the constraint system is \textit{canonically closed}. 
\end{lemma}
Here we discuss the optimality and duality theory between $\mathfrak{P}_{\bm{\mu}}$ and $\mathfrak{D}_{\bm{\mu}}$. Denote \textit{first moment cone} of $\mathfrak{D}_{\bm{\mu}}$ as $\mathbb{M}(\mathfrak{D}_{\bm{\mu}}):=\text{cone}\Big(\big(G_0(x),G_1(x),\ldots,G_n(x)\big)^\top ,x\in \mathbb{\Omega}\Big)$. Then for $\mathbb{F}(\mathfrak{D}_{\bm{\mu}})\neq \emptyset$, $\mathbb{F}(\mathfrak{D}_{\bm{\mu}})$ is bounded if and only if $\mathbb{M}(\mathfrak{D}_{\bm{\mu}})=\mathbb{R}^{n+1}$ (Theorem 9.1 of \citet{miguela.gobernaLinearSemiinfiniteOptimization2000}) and $\mathbb{F}^*(\mathfrak{D}_{\bm{\mu}})$ is a nonempty bounded set if and only if $\bm{\mu}\in \text{int}(\mathbb{M}(\mathfrak{D}_{\bm{\mu}}))$ (Corollary 9.3.1 of \citet{miguela.gobernaLinearSemiinfiniteOptimization2000}). According to Corolary 4.1.1 of \citet{gobernaFarkasMinkowskiSystemsSemiinfinite1981}, given any problem $\mathfrak{P}$ and the corresponding dual problem $\mathfrak{D}$, if the problem $\mathfrak{D}$ is a feasible FM, then the following statements are true: \textup{\textbf{(1)}} $v(\mathfrak{D}) = -\infty$ if and only if $\mathbb{F}(\mathfrak{P})=\emptyset$; \textup{\textbf{(2)}} $v(\mathfrak{D}) >-\infty$ if and only if $v(\mathfrak{D})=v(\mathfrak{P})$. 

We denote $Z^* = v(\mathfrak{P}_{\bm{\mu}})$. Assume $\bm{dr}\neq \mathbf{0}$ and $\rho>0$ as $\bm{dr}=\bm{0}$ or $\rho=0$ are trivial cases. 
Since $\bm{\mu}\in \textup{int}\big(\mathbb{M}(\mathfrak{D}_{\bm{\mu}})\big)$, $\mathbb{F}^*(\mathfrak{D}_{\bm{\mu}})$ is bounded. Combined with the fact that $\mathfrak{D}_{\bm{\mu}}$ is an FM system, we obtain $\mathfrak{P}_{\bm{\mu}}$ and $\mathfrak{D}_{\bm{\mu}}$ are solvable and strong duality holds (see the second remark after Theorem 2 of \citet{gobernaSensitivityAnalysisLinear2007}).

In addition, 
{\small
\begin{align}
    v(\mathfrak{P}_{\bm{\mu}+\rho \bm{dr}}) &\leq \min\{\sum\limits_{j\in \{0\}\cup [d]}y_j(\mu_j+\rho \bm{dr}_j)|\bm{y}\in \mathbb{F}(\mathfrak{D}_{\bm{\mu}})\}\leq \min\{\sum\limits_{j\in \{0\}\cup [d]}y_j(\mu_j+\rho \bm{dr}_j)|\bm{y}\in \mathbb{F}^*(\mathfrak{D}_{\bm{\mu}})\}\notag \\    
    &= v(\mathfrak{D}_{\bm{\mu}})+\rho\min\{\bm{dr}^\top \bm{y}|\bm{y}\in \mathbb{F}^*(\mathfrak{D}_{\bm{\mu}})\}=v(\mathfrak{P}_{\mu})+\rho\min\{\bm{dr}^\top \bm{y}|\bm{y}\in \mathbb{F}^*(\mathfrak{D}_{\bm{\mu}})\}\label{proof:sensitivity:ub}
\end{align}}
where the first inequality is based on weak duality and the last equality follows from the strong duality between $\mathfrak{P}_{\bm{\mu}}$ and $\mathfrak{D}_{\bm{\mu}}$. 
Let us consider an auxiliary problem $\mathfrak{P}'$ and its corresponding dual problem $\mathfrak{D}'$ as follows. 
\newpage
{\footnotesize
\vspace{-.5in}
\begin{multicols}{2}
\begin{align*}
(\mathfrak{P}'):\max{P,\pi} \ \  &\mathbb{E}_P[H(X)]+Z^*\pi \\
\text{s.t. } \ \ &\mathbb{E}_P [G_j(X)]+\mu_j\pi= \bm{dr}_j,\forall j\in \{0\}\cup [d], \\
&P\in \mathscr{M}^+(\mathbb{\Omega}).
\end{align*}
\break
\begin{align*}
(\mathfrak{D}'):\min_{\bm{y}} \ \  &\bm{dr}^\top \bm{y} \ \\
\text{s.t. } \ \ &\sum\limits_{\{0\}\cup [d]}y_jG_j(x)\geq H(x), \ \ \forall x\in \mathbb{\Omega}, \\
&\sum\limits_{\{0\}\cup [d]} y_j\mu_j= Z^*.
\end{align*}
\end{multicols}}

Since $\mathbb{F}^*(\mathfrak{D}_{\bm{\mu}})$ is non-empty and bounded which is a feasible region for $\mathfrak{D}'$, we have $\mathbb{M}(\mathfrak{D}') = \mathbb{R}^{n+1}$, implying that every $\bm{dr}\in \textup{int}\big(\mathbb{M}(\mathfrak{D}')\big)$ and thereby strong duality holds. 
Moreover, $\mathfrak{P}'$ is solvable since the probability measure as the decision variables has compact support. The above two arguments implies that $\bm{dr}^\top \bm{y}$ can achieve its minimum value on the feasible region $\mathbb{F}^*(\mathfrak{D}_{\bm{\mu}})$. 
Let $(\bar{P},\bar{\pi})$ be an optimal solution of $\mathfrak{P}'$ and $P^*$ be an optimal solution of $\mathfrak{P}_{\bm{\mu}}$. We argue that for any sufficiently small $\rho>0$, $P^*+\rho(\bar{P}+\bar{\pi}P^*)$ is feasible for $\mathfrak{P}_{\bm{\mu}+\rho \bm{dr}}$. If $\bar{\pi}\geq 0$ , $\rho$ can be any positive number. If $\bar{\pi}< 0$, $\rho\leq -\frac{1}{\bar{\pi}}$. To see that, (1) $P^*+\rho(\bar{P}+\bar{\pi}P^*)\in \mathcal{M}^+(\mathbb{\Omega})$ and (2). $\mathbb{E}_{P^*+\rho(\bar{P}+\bar{\pi}P^*)} [G_j(X)] = \mathbb{E}_{P^*} [G_j(X)] + \mathbb{E}_{\rho(\bar{P}+\bar{\pi}P^*)} [G_j(X)] = \mathbb{E}_{P^*} [G_j(X)] + \rho\mathbb{E}_{\bar{P}} [G_j(X)]+\rho\bar{\pi}\mathbb{E}_{P^*} [G_j(X)]=\mu_j+\rho(\bm{dr}_j-\mu_j\bar{\pi}+\bar{\pi}\mu_j)=\mu_j+\rho \bm{dr}_j, j\in \{0\}\cup [d]$. 

With strong duality of $\mathfrak{P}'$ and $\mathfrak{D}'$, we have 
\begin{align}
    v(\mathfrak{P}_{\bm{\mu}+\rho \bm{dr}})&\geq \mathbb{E}_{P^*+\rho(\bar{P}+\bar{\pi}P^*)}[H(X)] \notag\\
    &= \mathbb{E}_{P^*}[H(X)]+\rho\mathbb{E}_{\bar{P}}[H(X)]+\rho\bar{\pi}\mathbb{E}_{P^*}[H(X)]\notag\\
    &=v(\mathfrak{P}_{\bm{\mu}})+\rho\mathbb{E}_{\bar{P}}[H(X)]+\rho\bar{\pi}Z^*=v(\mathfrak{P}_{\bm{\mu}})+v(\mathfrak{P}')\notag\\
    &=v(\mathfrak{P}_{\bm{\mu}})+\rho \min\{\bm{dr}^\top \bm{y}|\bm{y}\in \mathbb{F}^*(\mathfrak{D_{\bm{\mu}}})\}.\label{proof:sensitivity:lb}
\end{align}

Combining (\ref{proof:sensitivity:ub}) and (\ref{proof:sensitivity:lb}) we have 
\begin{align*}
&v(\mathfrak{P}_{\mu})+\rho\min\{\bm{dr}^\top \bm{y}|\bm{y}\in \mathbb{F}^*(\mathfrak{D}_{\bm{\mu}})\} \geq v(\mathfrak{P}_{\bm{\mu}+\rho \bm{dr}})\geq v(\mathfrak{P}_{\bm{\mu}})+\rho \min\{\bm{dr}^\top \bm{y}|\bm{y}\in \mathbb{F}^*(\mathfrak{D_{\bm{\mu}}}).\}
\end{align*}
We then have 
\begin{align*}
v(\mathfrak{P}_{\bm{\mu}+\rho \bm{dr}}) = v(\mathfrak{P}_{\mu})+\rho\min\{\bm{dr}^\top \bm{y}|\bm{y}\in \mathbb{F}^*(\mathfrak{D}_{\bm{\mu}}),\}
\end{align*}
which concludes the proof.
\end{proof}

\subsection{Some Useful Theorems and Lemmas}
\begin{lemma}
\label{lm:ps}
Given $\alpha\in\mathbb{R}$, if \textup{(i)} $\int_{a}^\infty x^{\alpha}f(x)dx=\Gamma$, \textup{(ii)} $f(x)$ is non-increasing for $x\geq a$, \textup{(iii)} $f(x)\geq 0 ,\forall x\geq a$, then we have $x^{\alpha+1}f(x)\rightarrow 0$ as $x\rightarrow \infty$.

Besides conditions above, if we have $f'_{+}(x)$ is non-decreasing, non-positive for $x\geq a $ and $f(x)=\int_a^x f_+'(t)dt +\eta$, then we have $x^{\alpha+2}f'_{+}(x)\rightarrow 0$ as $x\rightarrow\infty$.
\end{lemma}
\begin{proof}[Proof of Lemma \ref{lm:ps}]
For $\alpha \leq 0$, $x^{\alpha}$ becomes bounded function, the results easily follow.
Now we assume $\alpha>0$. Consider the function 
\begin{align*}
T(x)&=x^{\alpha+1}f(x)-(\alpha+1)\int_a^x t^\alpha f(t)dt.
\end{align*}

For any $(a\vee 0)\leq x_1\leq x_2$,
\begin{align*}
T(x_2)-T(x_1)&=x_2^{\alpha+1}f(x_2)-x_1^{\alpha+1}f(x_1)-\int_{x_1}^{x_2}f(t)dt^{\alpha+1}\\
&\leq x_2^{\alpha+1}f(x_2)-x_1^{\alpha+1}f(x_1)-f(x_2)(x_2^{\alpha+1}-x_1^{\alpha+1})\\
&=-x_1^{\alpha+1}f(x_1)+x_1^{\alpha+1}f(x_2)\\
&=x_1^{\alpha+1}(f(x_2)-f(x_1))\leq 0.
\end{align*}

From above, we have $T(x)$ is non-increasing for $x\geq a\vee 0$. Since $x^{\alpha+1}f(x)\geq 0$ and $0\leq \int_{a}^x t^{\alpha}f(t)dt\leq \Gamma$ for $x\geq a\vee 0$, we have that $T(x)$ is bounded from below and thereby converges to a limit. Given $\int_{a}^x t^{\alpha}f(t)dt\rightarrow \Gamma $ as $x\rightarrow \infty$, we must have $x^{\alpha+1}f(x)$ go to some finite non-negative value. If $x^{\alpha+1}f(x)\geq \epsilon >0$ for all large enough $x$. This means $x^{\alpha}f(x)\geq \frac{\epsilon}{x}$ for all large enough $x$, and hence $\int_{a}^{\infty} x^{\alpha}f(x) dx =\infty$. It contradicts with condition (ii) above. Therefore, $x^{\alpha+1}f(x)$ must converge to $0$.

To prove the second part, consider a function 
\begin{align*}
P(x)=-x^{\alpha+2}f_+'(x)+(\alpha+2)[-\int_x^\infty t^{\alpha+1}f_+'(t)dt].
\end{align*} 

Note that because $f$ is absolutely continuous and $\lim_{x\rightarrow \infty}x^{\alpha+1}f(x)\rightarrow 0$ as we have just proved, integration by parts yields that 
\begin{align*}
-\int_x^\infty t^{\alpha+1}f_+'(t)dt&=x^{\alpha+1}f(x)+(\alpha+1)\int_x^{\infty}t^\alpha f(t)dt <\infty,\forall x\geq (a\vee 0).
\end{align*}

For any $(a\vee0) \leq x_1\leq x_2$,
\begin{align*}
P(x_2)-P(x_1)&=x_1^{\alpha+2}f_+'(x_1)-x_2^{\alpha+2}f_+'(x_2)+\int_{x_1}^{x_2}f_+'(t)dt^{\alpha+2}\\
&\leq x_1^{\alpha+2}f_+'(x_1)-x_2^{\alpha+2}f_+'(x_2)+f_+'(x_2)(x_2^{\alpha+2}-x_1^{\alpha+2})\\
&=x_1^{\alpha+2}(f_+'(x_1)-f_+'(x_2))\leq 0
\end{align*}

From above, we have $P(x)$ is non-increasing for $x\geq (a\vee 0)$. Since $-x^{\alpha+2}f_+'(x)\geq 0$ for $x\geq (a\vee 0)$ and $-\int_{x}^\infty t^{\alpha}f_+'(t)dt$ is non-negative, we have that $P(x)$ is bounded from below and converges to a limit. Because $\lim_{x\rightarrow \infty}x^{\alpha+1}f(x)\rightarrow 0$ and $\lim_{x\rightarrow \infty}\int_{x}^{\infty}f(t)t^{\alpha}dt=0$, we have $\lim_{x\rightarrow \infty}(-\int_x^\infty x^{\alpha+1}f_+'(t)dt) =0$. It implies $-x^{\alpha+2}f_+'(x)$ must converge to some finite non-negative value as $x\rightarrow \infty$.

If $-x^{\alpha+2}f_+'(x)\geq \epsilon >0$ for all large enough $x$. This means $-x^{\alpha+1}f_+'(x)\geq \frac{\epsilon}{x}$ for all large enough $x$, and hence $-\int_x^\infty t^{\alpha+1}f_+'(t)dt=\infty$ for $x\geq a$. It contradicts with the finiteness of the integration of $-t^{\alpha+1}f_+'(t)$ above. Therefore, $-x^{\alpha+1}f_+'(x)$ must converge to $0$.
\end{proof}
\begin{lemma}
\label{lm:gf}
Given function $g(x)$ bounded below for $x\geq a$, if \textup{(i)} $\int_{a}^\infty g(x)f(x)dx=\Gamma$,
\textup{(ii)} $f(x)$ is non-increasing for $x\geq a$, \textup{(iii)} $f(x)\geq 0 ,\forall x\geq a$ is a density function, then we have $\tilde{G}(x)f(x)\rightarrow 0$ as $x\rightarrow \infty$,where $\tilde{G}(x)=\int_a^x g(u) du$.

Besides conditions above, if we have $f'_{+}(x)$ is non-decreasing, non-positive for $x\geq a $, and $f(x)=\int_a^x f_+'(t)dt +\eta$, then we have $G(x)f'_{+}(x)\rightarrow 0$ as $x\rightarrow \infty$, where $G(x)=\int_a^x \tilde{G}(u) du$.
\end{lemma}
\begin{proof}[Proof of Lemma \ref{lm:gf}]
We first assume function $g(x)$ is a non-negative function for $x\geq a$.

Consider the function
\begin{align*}
T(x)&=\tilde{G}(x)f(x)-\int_a^x g(t)f(t)dt.
\end{align*}

For any $a\leq x_1\leq x_2$,
\begin{align*}
T(x_2)-T(x_1)&=\tilde{G}(x_2)f(x_2)-\tilde{G}(x_1)f(x_1)-\int_{x_1}^{x_2}g(t)f(t)dt\\
&=\tilde{G}(x_2)f(x_2)-\tilde{G}(x_1)f(x_1)-\int_{x_1}^{x_2}f(t)d\tilde{G}(t)\\
&\leq \tilde{G}(x_2)f(x_2)-\tilde{G}(x_1)f(x_1)-f(x_2)(\tilde{G}(x_2)-\tilde{G}(x_1))\\
&=-\tilde{G}(x_1)f(x_1)+f(x_2)\tilde{G}(x_1)\\
&=\tilde{G}(x_1)(f(x_2)-f(x_1))\leq 0.
\end{align*}

From above, we have $T(x)$ is non-increasing for $x\geq a$. Since $\tilde{G}(x)f(x)\geq 0$ and $0\leq \int_{a}^x g(t)f(t)dt\leq \Gamma$ for $x\geq a$, we have $T(x)$ bounded from below and converge to a limit. Because $\int_{a}^x g(t)f(t)dt \rightarrow \Gamma$ as $x\rightarrow \infty$, $\tilde{G}(x)f(x)$ must go to some finite non-negative value as $x\rightarrow \infty$. We now show that $\tilde{G}(x)f(x)\rightarrow 0$ as $x\rightarrow \infty$.

First of all, since $\tilde{G}(x)$ is non-negative and non-decreasing function, we must have that $\tilde{G}(x)$ either blows up or converges to some positive finite value. We first consider the case when $\tilde{G}(x)$ converges to some positive finite value as $x\rightarrow \infty$. Since $f(x)$ goes to zeros as $x\rightarrow \infty$ (indeed, $f(x)\geq 0,f(x)$ is non-increasing for $x\geq a$ and $f(x)$ is a density function), we have $\tilde{G}(x)f(x) \rightarrow 0$ as $x\rightarrow \infty$.

Now we consider the second case when $\tilde{G}(x)$ blows up as $x\rightarrow \infty$. We prove by contradiction. Since $\tilde{G}(x)f(x)$ is non-negative for $x\geq a$, we suppose that $\tilde{G}(x)f(x)$ will converge to a positive value as $x\rightarrow \infty$. It implies that there exist $ \epsilon>0$ and $x_{\epsilon}>a$ such that $\tilde{G}(x)f(x)\geq \epsilon, \ \forall x\geq x_{\epsilon}$

Hence we have
\begin{align*}
\int_{x_{\epsilon}}^{\infty} f(x)g(x)dx&\geq \int_{x_{\epsilon}}^{\infty} \frac{\epsilon}{\tilde{G}(x)}g(x)dx\\
&= \int_{x_{\epsilon}}^{\infty} \frac{\epsilon}{\tilde{G}(x)}d\tilde{G}(x)\\
&=\int_{x_{\epsilon}}^{\infty} \epsilon \ d\ln{\tilde{G}(x)}\\
&=\epsilon \ln{\tilde{G}(x)}|_{\infty}-\epsilon\ln{\tilde{G}(x_{\epsilon})}=\infty
\end{align*}

However, we have 
\begin{align*}
\int_{x_{\epsilon}}^{\infty} f(x)g(x)dx\leq \int_{a}^{\infty} f(x)g(x)dx=\Gamma<\infty.
\end{align*}
We get a contradiction. Hence we must have $\tilde{G}(x)f(x) \rightarrow 0$ as $x\rightarrow \infty$.

To prove the second part, we consider the function 
\begin{align*}
P(x)&=-G(x)f_+'(x)+[-\int_x^\infty \tilde{G}(t)f_+'(t)dt].
\end{align*}

Note that because $f$ is absolutely continuous and $\lim_{x\rightarrow \infty}\tilde{G}(x)f(x)\rightarrow 0$ as we have just proved, integration by parts yields that
\begin{align*}
-\int_x^\infty \tilde{G}(t)f_+'(t)dt&=\tilde{G}(x)f(x)+\int_x^{\infty}f(t)g(t)dt <\infty.
\end{align*}

For any $a\leq x_1\leq x_2$, we have 
\begin{align*}
P(x_2)-P(x_1)&=-G(x_2)f_+'(x_2)+G(x_1)f_+'(x_1)+[\int_{x_1}^{x_2} \tilde{G}(t)f_+'(t)dt]\\
&=-G(x_2)f_+'(x_2)+G(x_1)f_+'(x_1)+f_+'(x_2)(G(x_2)-G(x_1))\\
&=G(x_1)(f_+'(x_1)-f_+'(x_2))\leq 0.
\end{align*}

From above, we have $P(x)$ is non-increasing for $x\geq a$. Since $-G(x)f_+'(x)\geq 0$ and $-\int_{x}^{\infty} \tilde{G}(t)f_+'(t)dt\geq 0$ for $x\geq a$, we have that $P(x)$ is bounded from below and converges to a limit. Because $\lim_{x\rightarrow \infty}\tilde{G}(x)f(x)$ and $\lim_{x\rightarrow \infty}\int_{x}^{\infty}f(t)g(t)dt$ are both equal to zero, we have $\lim_{x\rightarrow \infty}(-\int_x^\infty \tilde{G}(t)f_+'(t)dt) =0$. It implies that $-G(x)f_+'(x)$ must go to some finite non-negative value as $x\rightarrow \infty$. We now show that $-G(x)f_+'(x)\rightarrow 0$ as $x\rightarrow \infty$.

First of all, since $G(x)$ is non-negative and non-decreasing function, we must have that $G(x)$ either blows up or converges to some positive finite value. We first consider the case when $G(x)$ converges to some positive finite value as $x\rightarrow \infty$. Suppose that $f_+'(x)\not\rightarrow 0$. Then by non-positiveness and non-decreasing properties of $f_+'(x)$, we have $f_+'(x)\rightarrow c<0$ as $x\rightarrow \infty$. However, $f(x)=\int_a^x f_+'(t)dt +\eta \rightarrow -\infty$ as $x\rightarrow \infty$, violating the non-negative condition of $f(x)$. We then have $f_+'(x)$ go to zeros as $x\rightarrow \infty$, thereby implying $-G(x)f_+'(x) \rightarrow 0$ as $x\rightarrow \infty$.

Now we consider the second case when $G(x)$ blows up as $x\rightarrow \infty$. We prove by contradiction, since $-G(x)f_+'(x)$ is non-negative for $x\geq 0$, we suppose that $-G(x)f_+'(x)$ will converge to a positive value as $x\rightarrow \infty$. It implies that there exist $ \epsilon>0$ and $x_{\epsilon}>a$ such that $-G(x)f_+'(x)\geq \epsilon, \ \forall x\geq x_{\epsilon}$

Hence we have 
\begin{align*}
\int_{x_{\epsilon}}^{\infty} -f_+'(x)\tilde{G}(x)dx&\geq \int_{x_{\epsilon}}^{\infty} \frac{\epsilon}{G(x)}\tilde{G}(x)dx\\
&= \int_{x_{\epsilon}}^{\infty} \frac{\epsilon}{G(x)}dG(x)\\
&=\int_{x_{\epsilon}}^{\infty} \epsilon \ d\ln{G(x)}\\
&=\epsilon \ln{G(x)}|_{\infty}-\epsilon\ln{G(x_{\epsilon})}=\infty.
\end{align*}

However, we have 
\begin{align*}
\int_{x_{\epsilon}}^{\infty} -f_+'(x)\tilde{G}(x)dx=\tilde{G}(x_{\epsilon})f(x_{\epsilon})+\int_{x_{\epsilon}}^{\infty}f(t)g(t)dt <\infty.
\end{align*}
We get a contradiction. Hence we must have $G(x)f_+'(x) \rightarrow 0$ as $x\rightarrow \infty$.

Now we consider the case when $g(x)$ is a bounded-below function for $x\geq a$ and the value of $g(x)$ can be negative. We consider $\tilde{g}(x)=g(x)+|\min_{x\geq a} g(x)|$. Clearly, $\tilde{g}(x)$ and $|\min_{x\geq a} g(x)|$ is non-negative function, which implies we can use the results just above. The results for $g(x)$ hence follow by linearity of integration and sum law of limits.
\end{proof}
\begin{lemma}
\label{lm3}
If $\mathbb{E}[G(X)] = \beta$ where $\beta$ is some finite constant and $G(x):\mathbb{R}\rightarrow \mathbb{R}$ is a non-decreasing function and bounded from below for $x\geq a$ with some constant $a$. For any $P\in \mathscr{P}[a,\infty)$, we have $\mathbb{E}_P[\mathbb{I}(X\geq x)]G(x)\rightarrow 0$ as $x\rightarrow \infty$.
\end{lemma}
\begin{proof}[Proof of Lemma \ref{lm3}]
We first assume that $G(x)$ is a non-negative function. Based on $\mathbb{E}[G(X)] = \beta$, it is easy to obtain that $\mathbb{E}[G(X)\mathbb{I}(X\geq x)] \rightarrow 0$ as $x\rightarrow \infty$. Then $\mathbb{E}[\mathbb{I}(X\geq x)]G(x) \leq \mathbb{E}[G(X)\mathbb{I}(X\geq x)] \rightarrow 0$ as $x\rightarrow \infty$. The first inequality follows from the non-decreasing property of $G(x)$. And from the non-negative property of $G(x)$, we have $\mathbb{E}[\mathbb{I}(X\geq x)]G(x)\geq 0$. It implies that $\mathbb{E}[\mathbb{I}(X\geq x)]G(x)\rightarrow 0$ as $x\rightarrow \infty$.

Now we consider the case when $G(x)$ is a bounded-below function and the value of $G(x)$ can be negative. We consider $\tilde{G}(x)=G(x)+|\min_{x\geq a} G(x)|$. Here, the minimum value is taken with respect to the support of the probability space. Clearly, $\tilde{G}(x)$ and $|\min_{x\geq a} G(x)|$ are non-negative functions, which implies we can use the results above. The result for $G(x)$ hence follows by sum law of limits. Given $\lim_{x\rightarrow \infty} \mathbb{E}[\mathbb{I}(X\geq x)][G(x)+|\min_{x\geq a} G(x)|] = 0$ and $\lim_{x\rightarrow \infty} \mathbb{E}[\mathbb{I}(X\geq x)][-|\min_{x\geq a} G(x)|] = 0$, we have $\lim_{x\rightarrow \infty} \mathbb{E}[\mathbb{I}(X\geq x)] G(x)=0$.
\end{proof}

\begin{theorem}\label{Choquet_mono}
The following two statements are equivalent: 

\textup{(1).} A function $f$ is finite, non-increasing and right-continuous for $x\geq a$ and $\int_a^\infty f(x)dx=\beta$.

\textup{(2).}
$f(x),x\geq a$ is a generalized mixture of the indicator functions $\frac{\mathbb{I}(a\leq x< a+z)}{z},z>\infty$, i.e.,
  \begin{align*}
  f(x)=\beta \int_{0+}^{\infty} \frac{\mathbb{I}(a\leq x< a+z)}{z}dQ(z)
  \end{align*}
  where $Q$ is a probability measure on $(0,\infty)$.
\end{theorem}
\begin{proof}[Proof of Theorem \ref{Choquet_mono}]
$(1)\implies (2)$:

  The probability distribution function defined as $F(x):=1-\beta+\int_a^x f(x)dx,\forall x\geq a$ is absolute continuous for $x\geq a$. We consider another probability distribution $\tilde{F}$ where $\tilde{F}(x)=\frac{F(x)+\beta-1}{\beta},x\geq a$ and $F(x)=0,x<a$. Then $\tilde{F}$ is still absolutely continuous on $x\geq a$ and density $\tilde{f}(x)$ is zero for $x<a$ and $\frac{f(x)}{\beta}$ for $x\geq a$. Then $\tilde{F}$ is unimodal at point $a$. Moreover, $\tilde{F}(x) = \int_a^x \tilde{f}(s)ds$. The right derivative of $\tilde{F}(x)$ exists for $x\geq a$. We also notice that $\tilde{f}(x+\delta)\leq \frac{\int_x^{x+\delta}\tilde{f}(s)ds}{\delta}\leq \tilde{f}(x)$. Since $f$ is right continuous for $x\geq a$, we have $\tilde{f}(x) = \lim\limits_{\delta\downarrow 0}\tilde{f}(x+\delta) \leq \lim\limits_{\delta\downarrow 0}\frac{\int_x^{x+\delta}\tilde{f}(s)ds}{\delta}\leq \tilde{f}(x).$ Therefore, 
\begin{align*}
  \tilde{F}_+'(x)=\lim\limits_{\delta\downarrow 0}\frac{\tilde{F}(x+\delta)-\tilde{F}(x)}{\delta}=\lim\limits_{\delta\downarrow 0}\frac{\int_x^{x+\delta}\tilde{f}(s)ds}{\delta}=\tilde{f}(x).
\end{align*}

  Based on Theorem 1.2 of \citet{dharmadhikariUnimodalityConvexityApplications1988}, $\tilde{F}$ is a generalized mixture of the distribution functions $W_{a,z}$ where $W_{a,z}$ denote the uniform distribution function on $(a,a+z),z>0$. The reason that we do not consider $z\leq 0$ is that $\tilde{F}(x)=0$ for $x<a$ and $\tilde{F}$ is absolutely continuous on $x\geq a$. In particular, we have 
  \begin{align*}
  \tilde{F}(x) = \int_{0+}^{\infty} W_{a,z}(x)dQ(z)
  \end{align*}
  where $Q$ is a probability measure on $(0,\infty)$.  

  We then have
  \begin{align}
  \tilde{f}(x)&=\tilde{F}_+'(x)=\lim_{\delta_n\downarrow 0}\frac{\tilde{F}(x+\delta_n)-\tilde{F}(x)}{\delta_n}=\lim_{\delta_n\downarrow 0}\int_{0+}^\infty\frac{W_{a,z}(x+\delta_n)-W_{a,z}(x)}{\delta_n}dQ(z)\label{exchange_sum_limit_11}\\
  &=\lim_{\delta_n\downarrow 0}\int_{0+}^{(x-a)_-}\frac{W_{a,z}(x+\delta_n)-W_{a,z}(x)}{\delta_n}dQ(z)\notag\\
  &+\lim_{\delta_n\downarrow 0}\frac{W_{a,z}(x+\delta_n)-W_{a,z}(x)}{\delta_n}Q(x-a)+\lim_{\delta_n\downarrow 0}\int_{(x-a)_+}^{\infty}\frac{W_{a,z}(x+\delta_n)-W_{a,z}(x)}{\delta_n}dQ(z)\label{exchange_sum_limit_12}\\
  &=\lim_{\delta_n\downarrow 0}\frac{W_{a,x-a}(x+\delta_n)-W_{a,x-a}(x)}{\delta_n}Q(z=x-a)+\int_{(x-a)_+
  }^{\infty}\frac{1}{z}dQ(z)\label{mono_convergence_key_step_1},
  \end{align}
  and $\lim\limits_{\delta_n\downarrow 0} \frac{W_{a,x-a}(x+\delta_n)-W_{a,x-a}(x)}{\delta_n}=0$. 
  Hence we have $\tilde{f}(x)=\int_{(x-a)_+}^{\infty}\frac{1}{z}dQ(z)$. 

  The exchange of limit and summation follows from observing that the limit in the equality (\ref{exchange_sum_limit_11}) and the first two terms of (\ref{exchange_sum_limit_12}) exist. Since $0\leq \frac{W_{a,z}(x+\delta_n)-W_{a,z}(x)}{\delta_n}\leq \frac{W_{a,z}(x+\delta_{n+1})-W_{a,z}(x)}{\delta_{n+1}}$, the exchange of limit and integration in the equality (\ref{mono_convergence_key_step_1}) follows from monotone convergence theorem. 

  It then concludes with 
  \begin{align*}
 & \int_{0+}^{\infty} \frac{\mathbb{I}(a\leq x<a+z)}{z}dQ(z)=\int_{(x-a)_+}^{\infty} \frac{1}{z}dQ(z)=\tilde{f}(x)\\
  \implies &\quad f(x) = \beta\int_{0+}^{\infty} \frac{\mathbb{I}(a\leq x<a+z)}{z}dQ(z)
  \end{align*}

  $(2)\implies (1):$ 

  Since 
  \begin{align*}
  \int_a^\infty f(x)dx &= \int_a^\infty\beta\int_{0+}^{\infty} \frac{\mathbb{I}(a\leq x<a+z)}{z}dQ(z)dx = \beta\int_{0+}^{\infty}\int_a^\infty \frac{\mathbb{I}(a\leq x<a+z)}{z}dxdQ(z)\\
  &=\beta\int_{0+}^{\infty}1dxdQ(z)=\beta.
  \end{align*}

  For any $x_2\geq x_1\geq a$, we have $\forall z>0,\mathbb{I}(a\leq x_1<a+z)\geq \mathbb{I}(a\leq x_2<a+z)$. It is easy to obtain that 
  \begin{align*}
  f(x_1)&= \beta\int_{0+}^{\infty} \frac{\mathbb{I}(a\leq x_1<a+z)}{z}dQ(z)\geq \beta\int_{0+}^{\infty} \frac{\mathbb{I}(a\leq x_2<a+z)}{z}dQ(z)=f(x_2).
  \end{align*}

  Finally, $f(x)$ is right continuous for $x\geq a$ because $\forall x\geq a$,
  \begin{align*}
  \ \lim\limits_{\delta \downarrow 0} f(x+\delta)=&\lim\limits_{\delta \downarrow 0}  \beta\int_{0+}^{\infty} \frac{\mathbb{I}(a\leq x+\delta<a+z)}{z}dQ(z)\\
  &=\beta\int_{0+}^{\infty} \lim\limits_{\delta \downarrow 0}\frac{\mathbb{I}(a\leq x+\delta <a+z)}{z}dQ(z)\\
  &=\beta\int_{0+}^{\infty} \frac{\mathbb{I}(a\leq x<a+z)}{z}dQ(z)=f(x).
\end{align*}

Since $\mathbb{I}(a\leq x<a+z)$ is a right continuous function and $\big|\frac{\mathbb{I}(a\leq x<a+z)}{z}\big|$ is bounded by $Q$-integrable function $\frac{1}{z}$. The integrability of $\frac{1}{z}$ is from the fact that $f(a)=\beta \int_{0+}^{\infty} \frac{1}{z}dQ(z)<\infty$. The exchange of integration and limit follows from dominated convergence theorem.
\end{proof}
\end{document}